\documentclass[11pt]{article}
\usepackage{amsmath,amssymb,fullpage,graphicx}

\usepackage[authoryear,round]{natbib}

\usepackage{color}
\usepackage{wasysym}

\usepackage{tikz-cd}
\usepackage{adjustbox}

\newtheorem{theorem}{Theorem}
\newtheorem{lemma}[theorem]{Lemma}

\newtheorem{proposition}[theorem]{Proposition}

\newenvironment{proof}{{\bf Proof.}}{$\Box$}

\DeclareMathOperator*{\argmin}{argmin}

\newenvironment{enum}{
\begin{enumerate}
  \setlength{\itemsep}{1pt}
  \setlength{\parskip}{0pt}
  \setlength{\parsep}{0pt}}
{\end{enumerate}}

\usepackage{sectsty}
\allsectionsfont{\bfseries\sffamily}
\let\hat\widehat
\let\tilde\widetilde

\begin{document}

\begin{center}
\textsf{\textbf{\Large Hybrid Wasserstein Distance and Fast Distribution Clustering}}\\
\textsf{\textbf{Isabella Verdinelli and Larry Wasserman}}\\
\textsf{\textbf{Department of Statistics and Data Science}}\\
\textsf{\textbf{Carnegie Mellon University}}\\
\textsf{\textbf{December 27 2018}}
\end{center}

\begin{quote}
{\em 
We define a modified Wasserstein distance
for distribution clustering
which inherits many of the properties of
the Wasserstein distance but which can be estimated easily
and computed quickly.
The modified distance is the sum of two terms.
The first term --- which has a closed form --- measures the 
location-scale differences between the distributions.
The second term is an approximation that
measures the remaining distance after accounting
for location-scale differences.
We consider several forms of approximation with our main emphasis
being a tangent space approximation
that can be estimated using
nonparametric regression.
We evaluate the strengths and weaknesses of this approach
on simulated and real examples.}

\end{quote}

\section{Introduction}

The Wassertstein distance
has attracted much attention lately
because it has many appealing properties.
(\cite{panaretos2018statistical,sommerfeld2018inference}).
It is especially useful
as a tool for clustering
a set of distributions
$P_1,\ldots, P_N$
because it captures key shape characteristics of the distributions.
But the Wasserstein distance is difficult
to compute and difficult to estimate from samples.
In this paper we introduce 
a modified Wasserstein distance
that can be estimated and computed quickly.

{\bf Wasserstein Distance.}
If $X\in\mathbb{R}^d$ is a random vector with distribution $P$ and 
$Y\in\mathbb{R}^d$ is a random vector with distribution $Q$ then,
for $p\geq 1$, the
$p$-Wasserstein distance is defined by
\begin{equation}\label{eq::wasserstein}
W_p(P,Q) \equiv W_p(X,Y) =
\left(\inf_J \int ||x-y||^p \, dJ(x,y)\right)^{1/p}
\end{equation}
where
the infimum is over all
joint distributions $J$ for
$(X,Y)$ such that
$X$ has marginal $P$ and
$Y$ has marginal $Q$.
The minimizer $J^*$ is called the
{\em optimal transport plan}
or {\em the optimal coupling}.
Figure \ref{fig::coupling} illustrates an example
of a coupling $J$.
In this paper we will focus on the case $p=2$ 
and then we write $W(P,Q)$ or $W(X,Y)$
instead of $W_2(P,Q)$ or $W_2(X,Y)$.

The modified distance that we propose is
\begin{equation}\label{eq::this}
H^2(X,Y) = W^2(Z_X,Z_Y) + W_\dagger^2(\tilde X,\tilde Y)
\end{equation}
where
\begin{align*}
Z_X\sim N(\mu_X,\Sigma_X),\ &\ \ \ \ \
\tilde{X} = \Sigma^{-1/2}_X(X-\mu_X)\\
Z_Y\sim N(\mu_Y,\Sigma_Y),\ &\ \ \ \ \
\tilde{Y} = \Sigma^{-1/2}_Y(Y-\mu_Y),
\end{align*}
$\mu_X = \mathbb{E}[X]$, $\mu_Y = \mathbb{E}[Y]$,
$\Sigma_X = {\rm Var}[X]$ and $\Sigma_Y = {\rm Var}[Y]$
and $W_\dagger$ is a
distance between the centered and scaled variables
$\tilde X$ and $\tilde Y$.
We consider several possible choices for $W_\dagger$.
We mainly focus on the case where
$W_\dagger(\tilde X,\tilde Y)$
is a tangent space approximation to
$W(\tilde{X},\tilde{Y})$
as defined by 
\cite{wang2013linear}.
The details of this tangent approximation
are given in Section \ref{section::modified}.
Our version of the tangent space distance
is a bit different than the original implementation 
as we use a combination of density estimation, permutation smoothing 
and subsampling. We will call $H$ the {\em hybrid} distance.
We will consider other choices for $W_\dagger(\tilde X,\tilde Y)$
in Section \ref{section::other}.

The first term in (\ref{eq::this})
measures location-scale differences
between the two distributions,
is available in closed form 
(see equation \ref{eq::Gaussian})
and can be estimated at a $n^{-1/2}$ rate
where $n$ is the sample size.
The second term captures any remaining non-linear differences.

\begin{figure}
\begin{center}
\begin{tabular}{ll}
\includegraphics[height=3in]{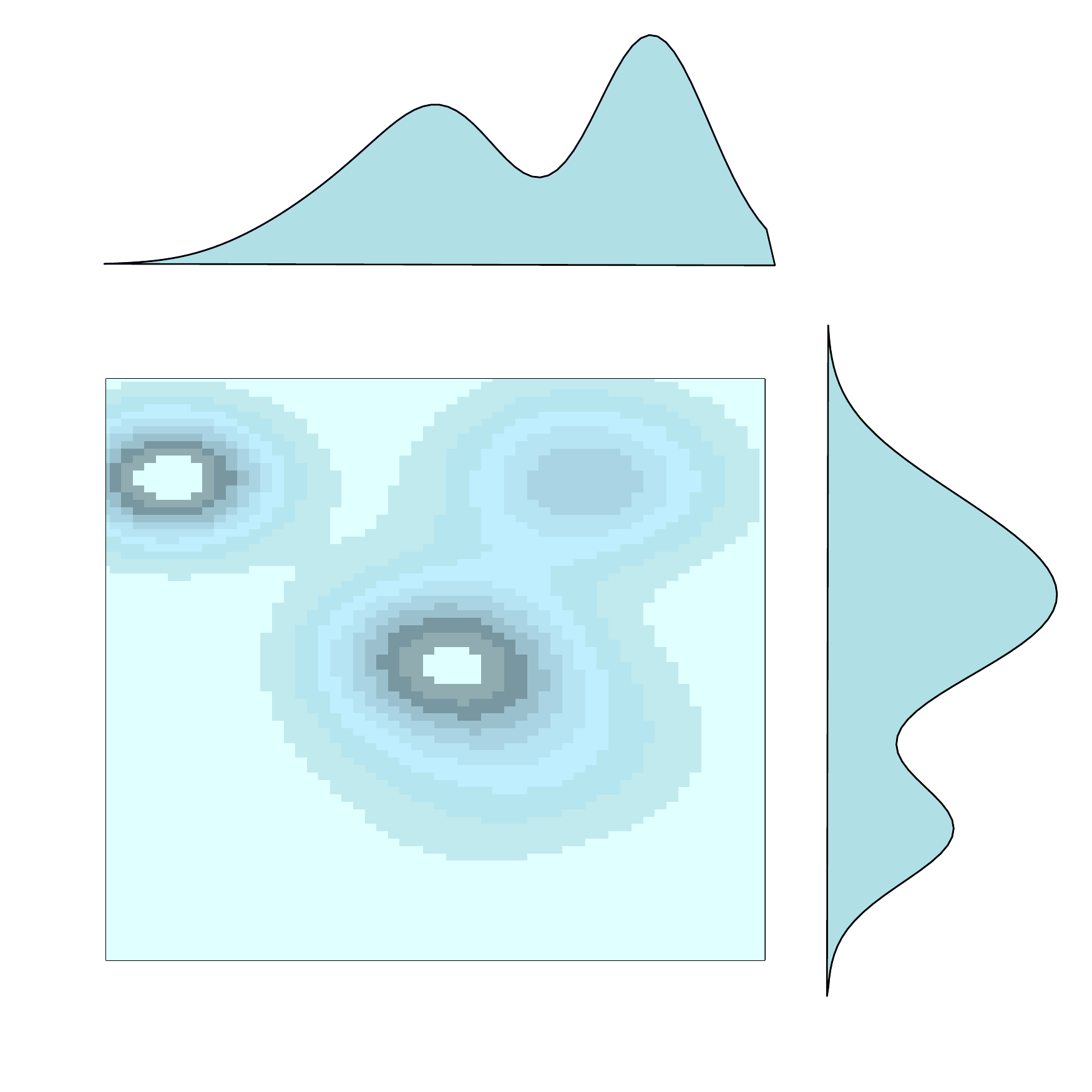}
\end{tabular}
\end{center}
\vspace{-.8cm}
\caption{\em This is an illustration of a coupling
$J$ with given marginals for $X$ and $Y$.
Generally, there are many such joint distributions.}
\label{fig::coupling}
\end{figure}

\parskip 10pt

{\bf Distribution Clustering}.
As mentioned above,
our main motivation
is ``distribution clustering''
which requires repeatedly computing distances.
Suppose, for example, that we want to cluster
a set of distributions $P_1,\ldots, P_N$.
Typically, these are empirical distributions
corresponding to 
datasets
${\cal D}_1,\ldots, {\cal D}_N$.
Given a metric $d$ on the set of probability
distributions, we can adapt existing methods ---
such as hierarchical clustering and $k$-means clustering ---
to the problem of clustering distributions.
But if we use
Wasserstein distance then
the calculations become onerous
since we need to compute many distances.
For example, suppose we want to perform 
agglomerative hierarchical clustering.
First we need to compute the
$N(N-1)/2$ distances
$W_p(P_i,P_j)$.
If we decide to cluster, say $P_1$ and $P_2$,
we need to combine the corresponding datasets
${\cal D}_1$ and ${\cal D}_2$.
Then we need to compute the distance between the new empirical measure
corresponding to
${\cal D}_1 \bigcup {\cal D}_2$ and all the other distributions.
Each stage of the hierarchical clustering involves recalculating the distances.
Similarly, if we use $k$-means clustering
we need to iterate between assigning points to clusters and computing centroids.
Computing the centroid --- also known as a barycenter --- with respect to the Wasserstein distance
is computationally expensive.
We show that replacing $W$ with $H$
significantly reduces the computational burden
without sacrificing accurate clustering.

{\bf Related Work.}
Our works builds on \cite{wang2013linear}
who introduced the idea of using a tangent space approximation to Wasserstein distance.
Their motivation was image processing and their implementation of the idea 
is quite different than our version.
We also make use of subsampling approximations which were suggested in
\cite{sommerfeld2018optimal}.
The Gaussian approximation $W(Z_X,Z_Y)$ is an example of a linear approximation.
Optimal linear approximations to
Wasserstein distance are studied in \cite{kuang2017preconditioning}.
An important reference on clustering distributions
with Wasserstein distance 
is \cite{del2017robust}
which not only introduces the idea
of distribution clustering in this way,
but also proposes a trimming procedure to create
robust clusterings.
We do not consider robustness in this paper
but we note that combining ideas from
\cite{del2017robust} with the ideas in this paper is an interesting
future direction.
Wasserstein clustering is also
studied in
\cite{ho2017multilevel}
in the context of hierarchical models.
That paper not only clusters distributions but,
simultaneously, clusters data within each distribution.

\vspace{1cm}

{\bf Paper Outline.}
In Section \ref{section::wass}
we review the Wasserstein distance.
In Section \ref{section::modified}
we give the details of the proposed modified distance.
In Section \ref{section::kmeans}
we define several versions of k-means distribution clustering.
Section \ref{section::examples}
gives some examples.
In Section \ref{section::other}
we briefly explain how our ideas can be used
for hierarchical clustering and mean-shift clustering.
In Section \ref{section::otherH}
we discuss some different versions of hybridization.
Section \ref{section::discussion} contains a discussion
and concluding remarks.

\section{Wasserstein Distance}
\label{section::wass}

In this section we give a brief review of the Wasserstein distance
and we explain why it is useful for distribution clustering.
An excellent reference on Wasserstein distance is \cite{villani2003topics}.
Recall that the Wasserstein distance is defined in
equation (\ref{eq::wasserstein}).

{\bf Explicit Expressions.}
In general, there is no closed form expression
for $W_p$.
There are three notable exceptions.

{\em (i)} 
When $d=1$,
the distance can be written explicitly as
$$
W_p(P,Q) = \left(\int_0^1 |F^{-1}(z) - G^{-1}(z)|^p\right)^{1/p}
$$
where $F(x) = P(X\leq x)$ and $G(y) = Q(Y\leq y)$.

{\em (ii)}
If $P_n$ is the empirical distribution of a 
dataset $X_1,\ldots, X_n$ and $Q_n$ is the empirical distribution 
of another dataset $Y_1,\ldots, Y_n$ of the same size,
then the distance takes the form
$$
W_p^p(P_n,Q_n) = \min_\pi \frac{1}{n}\sum_{i=1}^n ||X_i - Y_{\pi(i)}||^p
$$
where the minimum is over all permutations.
This minimization can be done using various algorithms
such as the Hungarian algorithm 
(\cite{kuhn1955hungarian})
which takes time $O(n^3)$.
When $d=1$ this further simplifies to
\begin{equation}
W_p(P_n,Q_n) = \left( \frac{1}{n}\sum_{i=1}^n |X_{(i)}-Y_{(i)}|^p\right)^{1/p}
\end{equation}
where
$X_{(1)} \leq \cdots \leq X_{(n)}$ and
$Y_{(1)} \leq \cdots \leq Y_{(n)}$ 
are the order statistics.

{\em (iii)}
The distance has a simple expression is the Gaussian case
(or more generally, for location-scale families).
Suppose that
$X\sim N(\mu_X,\Sigma_X)$ and 
$Y\sim N(\mu_Y,\Sigma_Y)$.
Then
\begin{equation}\label{eq::Gaussian}
W^2(P,Q)\equiv W^2(X,Y) =
||\mu_X - \mu_Y||^2 + {\rm B}^2(\Sigma_X,\Sigma_Y)
\end{equation}
where
\begin{equation}
{\rm B}^2(\Sigma_X,\Sigma_Y) =
{\rm tr}(\Sigma_1) + {\rm tr}(\Sigma_2) -
2 {\rm tr} \Bigl[ \Bigl( \Sigma_X^{1/2}\Sigma_Y \Sigma_X^{1/2}\Bigr)^{1/2}\Bigr]
\end{equation}
is the Bures distance 
(\cite{bhatia2018bures})
between
$\Sigma_X$ and $\Sigma_Y$.
See \cite{givens1984class} and
\cite{rippl2016limit}.
From now on, we refer to (\ref{eq::Gaussian})
as the {\em Gaussian Wasserstein distance}
and we denote this as
$G^2(X,Y)=||\mu_X - \mu_Y||^2 + {\rm B}^2(\Sigma_X,\Sigma_Y)$.

\vspace{.2cm}

{\bf The Monge Distance and Transport Maps.}
A related distance is the Monge distance defined~by
\begin{equation}\label{eq::Monge}
\left(\inf_T \int ||x-T(x)||^p dP(x)\right)^{1/p}
\end{equation}
where the infimum is over all maps $T$
such that $T(X)\sim Q$.
When a minimizer exists,
this corresponds to the Wasserstein distance
and the map $T$ is called the {\em optimal transport map}.
In this case the optimal coupling
$J^*$ is a degenerate distribution on the set
$\{(x,T(x))\}$.
But, the minimizer might not exist.
Consider
$P=\delta_0$ and
$Q = (1/2)\delta_{-1} + (1/2)\delta_{1}$
where $\delta_a$ denotes a point mass at $a$.
In this case, there is no map $T$ such that
$T(X) \sim Q$.
In contrast, an optimal coupling always exists
and can be thought of as defining a transport plan that allows the 
mass to be split and assigned to many locations.
A sufficient condition for the existence
of a unique optimal transport map is that both $P$ and $Q$ be absolutely continuous
with respect to Lebesgue measure.
In the Gaussian case,
the optimal transport map is
the linear map
$$
L(x) = \mu_Y + \Sigma_Y^{1/2}\Sigma_X^{-1/2}(x-\mu_X).
$$

{\bf Barycenters.}
Given a set of distributions
$P_1,\ldots, P_N$,
the {\em barycenter},
with respect to non-negative weights
$\lambda_1,\ldots,\lambda_N$, 
is defined to be the distribution $P$ that
minimizes
$\sum_j \lambda_j W^2(P,P_j)$.
(In this paper, we will always use $\lambda_j = 1/N$.)
There is a substantial literature
on finding methods to compute the barycenter.
In the special case that each $P_j$ is Gaussian,
the barycenter takes a special form.
Let
$P_j=N(\mu_j,\Sigma_j)$
for $j=1,\ldots, N$.
Then the barycenter $P$ is
$N(\mu,\Sigma)$ where
$\mu = N^{-1}\sum_j \lambda_j\mu_j$
and $\Sigma$ is the unique, symmetric, positive definite matrix
satisfying the fixed point equation
\begin{equation}\label{eq::fp}
\Sigma = \sum_j \lambda_j ( \Sigma^{1/2} \Sigma_j \Sigma^{1/2})^{1/2}.
\end{equation}
The same holds for any location-scale family;
see \cite{alvarez2016fixed}.
In one dimension,
the barycenter of $P_1,\ldots, P_N$
is the distribution $P$ with cdf $F$ where
$F^{-1}(u)=\sum_j \lambda_j F_j^{-1}(u)$.

{\bf Key Properties.}
There has been a surge of interest in the Wasserstein distance
in statistics and machine learning lately.
This is because the distance has a number of useful properties.
Here we review three key properties, namely:
(1) sensitivity to underlying geometry,
(2) comparability of discrete and continuous distributions and
(3) shape preservation.

\begin{enumerate}
\item {\em The Wasserstein distance is sensitive to the underlying geometry}.
Consider the distributions
$P_1 = \delta_0, P_2=\delta_\epsilon$ and
$P_3 = \delta_{100}$
where $\delta_a$ denotes a point mass at $a$
and $\epsilon>0$ is a small positive number.
Then
$W(P_1,P_2)\approx 0$,
$W(P_1,P_3)\approx W(P_2,P_3)\approx 100$.
On the other hand, consider the total variation distance
$d_{\rm TV}$.
Then
$d_{\rm TV}(P_1,P_2)=d_{\rm TV}(P_1,P_3)=d_{\rm TV}(P_2,P_3)=1$.
So the total variation distance fails to capture our intuition that
$P_1$ and $P_2$ are close while
$P_3$ is far.
The same is true for Hellinger distance and Kullback-Leibler distance.

\item {\em The Wasserstein distance permits direct comparison between discrete
and continuous distribution}.
If $P_1$ is continous and $P_2$ is discrete,
then, for example, $d_{\rm TV}(P_1,P_2)=1$.
But $W_p(P_1,P_2)$ gives reasonable values.
For example, suppose that
$P_1$ is uniform on $[0,1]$ and
$P_2$ is uniform on $\{1/N,2/N,\ldots, 1\}$.
Then
$d_{\rm TV}(P_1,P_2)=1$ for all $N$ but
$W_p(P_1,P_2)= 1/N$ which again seems quite intuitive.

\item {\em Shape Preservation}.
Suppose we have a set of distributions
$P_1,\ldots, P_N$.
Recall that the barycenter $P$ minimizes
$\sum_j \lambda_j W_2^2(P,P_j)$.
The barycenter $P$ preserves the shape of the distributions.
Specifically, if each $P_j$ can be written as a location-scale shift
of some distribution $P_0$,
the $P$ is also a location-scale shift of $P_0$.
For example,
suppose that $P_1 = N(\mu_1,\Sigma)$ and 
that $P_2 = N(\mu_2,\Sigma)$.
Then the barycenter is
$P=N((\mu_1+\mu_2)/2,\Sigma)$.
In contrast,
the Euclidean average
$(1/2)P_1 + (1/2)P_2$
looks nothing like any of the $P_j$'s.
Figure \ref{fig::BaryPlot}
shows a comparison of the Wasserstein barycenter and the
Euclidean average.
\end{enumerate}

\begin{figure}
\begin{center}
\includegraphics[width=6in,height=2in]{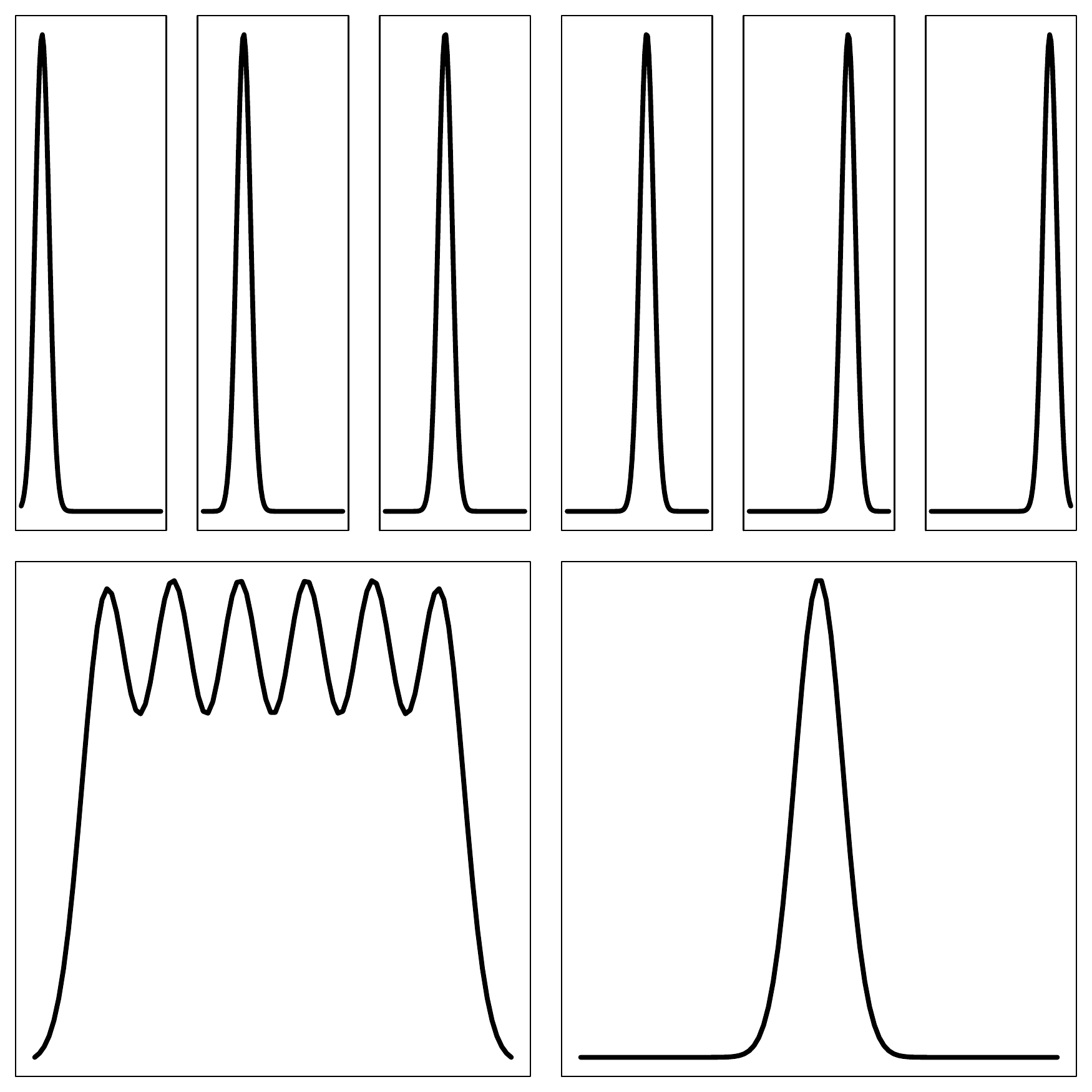}
\end{center}
\caption{\em Top row: Six Gaussian distributions.
Bottom left: The Euclidean average of the six distributions.
Bottom right: The Wasserstein barycenter.}
\label{fig::BaryPlot}
\end{figure}

\section{The Hybrid Distance}
\label{section::modified}

Let $X\sim P$ and $Y\sim Q$. 
Define $Z_X\sim N(\mu_X,\Sigma_X)$,
$Z_Y\sim N(\mu_Y,\Sigma_Y)$,
$\tilde{X} = \Sigma^{-1/2}_X(X-\mu_X)$,
and $\tilde{Y} = \Sigma^{-1/2}_Y(Y-\mu_Y)$.
Our modified distance ---
which we call the {\em hybrid distance} ---
is
\begin{equation}
H^2(X,Y) = W^2(Z_X,Z_Y) + W_\dagger^2(\tilde X,\tilde Y) =  G^2(X,Y)+ W_\dagger^2(\tilde X,\tilde Y)
\end{equation}
where
$W_\dagger$ is described below.
The first term $W^2(Z_X,Z_Y)\equiv G^2(X,Y)$
has a simple closed form, namely,
$||\mu_X - \mu_Y||^2 + {\rm B}^2(\Sigma_X,\Sigma_Y)$
where ${\rm B}(\cdot,\cdot)$ is the Bures distance.
Given samples
$X_1,\ldots, X_n \sim P$
and
$Y_1,\ldots, Y_m \sim Q$,
we can estimate
$W^2(Z_X,Z_Y)$ by plugging in sample moments.
We then have that
$\hat W(Z_X,Z_Y) = W(Z_X,Z_Y) + O_P( (n\wedge m)^{-1/2})$;
see \cite{rippl2016limit}.
We measure the remaining difference by computing the 
distance between the standardized variables 
$\tilde{X} = \Sigma^{-1/2}_X(X-\mu_X)$,
$\tilde{Y} = \Sigma^{-1/2}_Y(Y-\mu_Y)$
by adapting the method
from \cite{wang2013linear}
which we now describe.

Consider a set of distributions
$P_1,\ldots, P_N$.
Let $\mu_j = \mathbb{E}[X_j]$ and
$\Sigma_j = {\rm Var}[X_j]$
where $X_j \sim P_j$.
Let
$Z_j\sim N(\mu_j,\Sigma_j)$
and $\tilde{X}_j = \Sigma_j^{-1/2}(X_j - \mu_j)$.
Let $R$ be a reference measure with density $r$,
($R$ is defined below.)
Define
\begin{equation}
W_\dagger^2(\tilde{P}_j,\tilde{P}_k)=
\int (\psi_j(z)-\psi_k(z))^2 dz
\end{equation}
where
$\psi_j(z) = (T_j(z)-z)\sqrt{r(z)}$,
and $T_j$ is the optimal transport map from $\tilde{P}_j$ to $R$.
\cite{wang2013linear}
justify this expression as follows.
The set of probability measures
endowed with the Wasserstein metric
is a Riemannian manifold.
Then
$\int (\psi_j(z)-\psi_{k}(z))^2 dz$
is the distance between
the projections of $\tilde{P}_j$ and $\tilde{P}_k$ onto the tangent space 
at $R$; %%see Figure \ref{fig::manifold}. 
Now
$$
\int (\psi_{j}(z)-\psi_{k}(z))^2 dz =
\int (T_{j}(z)-T_{k}(z))^2 dR(z).
$$
Hence,
if $U_1,\ldots, U_m \sim R$ then
$$
\int (\psi_{j}(z)-\psi_{k}(z))^2 dz =
\frac{1}{m}\sum_{s=1}^m (T_{j}(U_s)-T_{k}(U_s))^2 + O_P(m^{-1/2}).
$$

\begin{figure}
\vspace{-1in}
\begin{center}
\includegraphics[scale=.4]{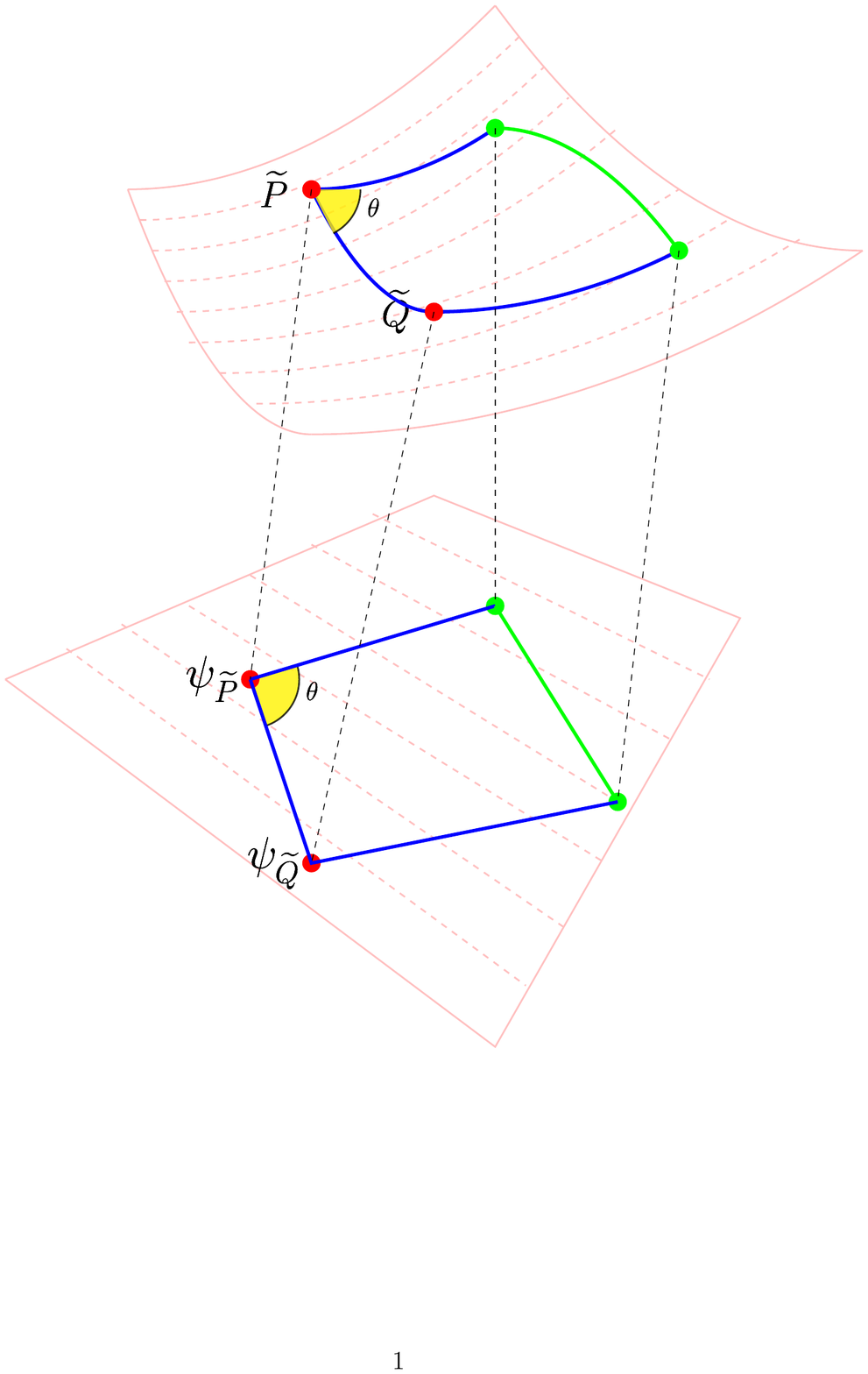} %%[width=3in,height=3in]
\end{center}
\vspace{-1in}
\caption{\em The space of probability measures equipped 
with Wasserstein distance is a manifold. The tangent distance between 
$P$ and $Q$ is the Euclidean distance between the projection
of $P$ and $Q$ into the local tangent space at the reference measure $R$.}
\label{fig::manifold}
\end{figure}

The hybrid distance is
\begin{align}\nonumber
H^2(X_j,X_k) &=
W^2 (Z_j,Z_k) + W_\dagger^2(\tilde{X}_j,\tilde{X}_k)\\
&= \nonumber
||\mu_j - \mu_k||^2 +
{\rm B}^2(\Sigma_j,\Sigma_k) + \int (\psi_j(z)-\psi_k(z))^2 dz\\
&=\nonumber
||\mu_j - \mu_k||^2 +
{\rm B}^2(\Sigma_j,\Sigma_k) + \int (T_j(z)-T_k(z))^2 dR(z)\\
& \approx \label{eq::final}
\underbrace{||\mu_j - \mu_k||^2}_{\rm location} + 
\underbrace{{\rm B}^2(\Sigma_j,\Sigma_k)}_{\rm scale} +
\underbrace{\frac{1}{m}\sum_i (T_j(U_i)-T_k(U_i))^2}_{\rm shape}
\end{align}
where
$U_1,\ldots, U_m \sim R$.

To estimate the distance we need to choose $R$
and estimate $T_j$.
\cite{wang2013linear}, motivated by applications in image processing,
suggest using
$R=N^{-1}\sum_j P_j$.
We take a slightly different approach.
Recall that we have datasets
${\cal D}_1,\ldots, {\cal D}_N$
where
${\cal D}_j= \{X_{j1},\ldots, X_{j n_j}\}$
consists of
$n_j$ observations with empirical distribution $P_j$.
Let 
$\tilde{D}_j = (\tilde{X}_{js}:\ 1\leq s \leq n_j)$
denote the normalized
observations where
$\tilde{X}_{js} = \hat\Sigma_j^{-1/2}(X_{js}-\hat\mu_j)$.
The combined dataset
$\tilde{\cal D}=\bigcup_{j=1}^N \tilde{\cal D}_j$
can be regarded as a sample from
$\sum_j \pi_j \tilde{P_j}$
where $\pi_j=n_j/\sum_j n_j$
and $\tilde P_j$ is the distribution of $\tilde X_j$.
Let $R$ be the distribution with density $r$
where $r$ is a kernel density estimate
obtained from
$\tilde{\cal D}$ using
a simple bandwidth rule such as Silverman's rule \citep{silverman2018density}.
Thus,
$R = R_n\star K_h$ (the convolution)
where
$R_n = \sum_j \pi_j \tilde{P}_j$
and $K_h$ is a kernel with bandwidth $h$.
This choice of reference measure is simple and smooth.

Next we have to estimate
$T_j$.
Here we use 
a variation of nonparametric regression
that we call
{\em permutation smoothing}.
The steps are 
given in Figure \ref{fig::alg1}.
The idea is to sample $m$ observations from each dataset and the 
reference measure $R$.
The optimal permutation for matching the samples can be found 
in $O(m^3)$ time.
This defines a map $T_j$ from $m$ points drawn from $\tilde P_j$ 
to $m$ points drawn from $R$.
We then extend $T_j$ over the whole space
by using $r$-nearest neighbor regression.
We show in the appendix that it suffices to take $r=1$
to get a consistent estimator of the transport function.
This amounts to taking $T_j$ to be the piecewise
constant
over the 
Voronoi tesselation 
$V_1,\ldots, V_m$ 
defined by
the point drawn from $\tilde{P}_j$.
We can summarize the steps as follows:

\bigskip
\fbox{\parbox{6.5in}{
$$
{\rm Gaussian\ approximation}\longrightarrow
{\rm subsample}\longrightarrow
{\rm permutation\  smoothing}\longrightarrow\
{\rm tangent\ approximation}.
$$
}}

\bigskip
\cite{wang2013linear}
point out that the tangent space approximation
is a well-defined distance and
need not be thought of as an approximation to the Wasserstein distance.
They show that, even when it does not approximate the Wasserstein distance, 
it still contains valuable information for comparing distributions.

\begin{figure}
\fbox{\parbox{6in}{
\begin{enum}
\item Fix an integer $m$.
\item Draw a sample $U_1,\ldots, U_m \sim R$ and 
draw a subsample $\tilde{X}_{j1},\ldots,\tilde{X}_{jm}$ from 
$\tilde{\cal D}_j$.
\item Find the permutation $\pi_j$ that
minimizes
$\sum_{i=1}^m ||\tilde{X}_{ji} - U_{\pi(i)}||^2$.
This can be done using the Hungarian algorithm and takes $O(m^3)$ time.
\item Let $V_{j1},\ldots, V_{jm}$ be the Voronoi tesselation defined by
$\tilde{X}_{j1},\ldots,\tilde{X}_{jm}$.
Define $\hat T_j(x) = \sum_s U_{\pi_j(s)}I(x\in V_s)$
for all $x\in \{U_1,\ldots, U_m\}$.
\end{enum}
}}
\caption{\em Permutation smoothing algorithm to estimate transport map.}
\label{fig::alg1}
\end{figure}

In the appendix we show that this method is consistent
if $n_j\to\infty$ and $m\to \infty$.
The idea of using subsamples to approximate Wasserstein distance
(rather than using subsamples to estimate the transport map to a reference measure as we are doing)
was examined carefully in \cite{sommerfeld2018optimal}.
For distribution clustering, the subsample size $m$ need not be large.
(In our examples we use $m=100$ but we get similar results even 
using $m=20$). This keeps the computation very fast. Also, note that 
we only ever evaluate $\hat T_j$ on the points $U_1,\ldots, U_m\sim R$.

{\bf Remark.}
{\em \cite{sommerfeld2018optimal}
suggest estimating Wasserstein distance by averaging over subsamples.
Similarly, we could repeat our procedure
over several subsamples and average the $\hat T_j$'s.
However, we have not found this to be necessary for distribution clustering.}

{\bf Remark.}
{\em We use $r$-nearest neighbor regression with $r=1$.
Of course, $r$ can be replaced with
a larger value if we want $\hat T_j$ to be smooth.
However, $r=1$ suffices for consistency.}

Finally, we estimate $H^2$ by
\begin{equation}
\hat H^2(P_j,P_k)=
||\hat\mu_j - \hat\mu_k||^2 +
{\rm B}^2(\hat\Sigma_j,\hat\Sigma_k) +
\frac{1}{m}\sum_{s=1}^m (\hat T_j(U_s) - \hat T_k(U_s))^2.
\end{equation}

The modified distance retains many properties of Wasserstein distance.
The following proposition summarizes these facts.
The proof is straightforward and is omitted.

\begin{proposition}
The distance $H$ has the following properties.
\begin{enum}
\item $H$ is a metric on the space of distributions with densities and
finite second moments.
In particular,
$H(P,Q)=0$ if and only if $P=Q$.
\item $H$ is exact for Gaussians: if $P$ and $Q$ are Gaussian then
$H^2(P,Q) = W^2(P,Q)$.
\item If $P_1,\ldots, P_N$ are in a location-scale family then
the barycenter is in the same family.
\end{enum}
\end{proposition}

When clustering we need to repeatedly compute averages.
In terms of the Wasserstein distance this corresponds to
computing barycenters.
The barycenter with respect to
$H$ is easy to compute.
Figure \ref{fig::barybary} summarizes the steps for computing the barycenter.

\begin{lemma}
Let $\overline{P}$ minimimize
$V(P)=\sum_j \lambda_j H^2(P,P_j)$.
Then $\overline{P}$ can be characterized as follows:
$\overline{P}$ is the distribution of the random variable
$$
Y = \overline{\mu} + \overline{\Sigma}^{1/2} \overline{T}^{-1}(U)
$$
where
$U\sim R$,
$\overline{T}(z) = z + \sum_j \lambda_j (T_j(z)-z)$,
$\overline{\mu} = \sum_j \lambda_j \mu_j$
and
$\overline{\Sigma}$ is the unique, positive definite matrix satisfying
the fixed point equation
$\Sigma = \sum_j \lambda_j ( \Sigma^{1/2} \Sigma_j \Sigma^{1/2})^{1/2}$.
\end{lemma}

\begin{proof}
Let $P$ be a distribution and
let $\mu$ be its mean, let $\Sigma$ be its covariance
and let $T$ be the optimal transport map from $P$ to $R$.
From the definition of $H$ we have that
$$
\sum_j H^2(P,P_j) = \sum_j ||\mu-\mu_j||^2 + \sum_j B(\Sigma,\Sigma_j) +
\sum_j \int (\psi_j(z) - \psi(z))^2 dz
$$
where $\psi(z) = (T(z)-z)\sqrt{r(z)}$ and
$\psi_j(z) = (T_j(z)-z)\sqrt{r(z)}$.
By minimizing each sum separately,
the optimal $P$ has mean and variance as stated and
its transport map satisfies
$\psi = \sum_j \lambda_j \psi_j$ which implies that
$\overline{T}(z) = z + \sum_j \lambda_j (T_j(z)-z)$.
\end{proof}

We can regard the triple
$(\mu_P,\Sigma_P,\psi_P)$
as a transform $\phi$ 
of $P$
where 
\begin{equation}\label{eq::transform}
\phi: P \mapsto (\mu_P,\Sigma_P,\psi_P).
\end{equation}
We call $\phi$ the
{\em hybrid transform}.
Note that the transform depends on the reference measure $R$.
Barycenters are computed by averaging each component
of the triple separately
(with the appropriate fixed point equation used for $\Sigma$.)
All clustering calculations can be carried out
in terms of the representation rather than in terms of the original distribution.
The representation is invertible:
given a 
triple $(\mu_P,\Sigma_P,\psi_P)$,
the corresponding $P$ is the distribution of the random variable
$\mu_P + \Sigma^{1/2}_P T^{-1}(U)$ with $U\sim R$
and $T(z) = \psi_P(z)/\sqrt{r(z)}+z$.
We write
$P = \phi^{-1}(\mu_P,\Sigma_P,\psi_P)$.

\begin{figure}
\fbox{\parbox{6in}{
\begin{center}
{\sf Barycenters}
\end{center}
\begin{enum}
\item Given $P_1,\ldots, P_N$ and weights $\lambda_1,\ldots, \lambda_N$.
\item Compute the means $\mu_1,\ldots, \mu_N$ and 
variances $\Sigma_1,\ldots, \Sigma_N$.
\item Use the algorithm in Figure \ref{fig::alg1}
to compute $T_1,\ldots, T_N$.
\item Let $\psi_j(z) = (T_j-z)\sqrt{r(z)}$.
\item Set
$\overline{\mu}= \sum_j \lambda_j \mu_j$ and
$\overline{\psi}= \sum_j \lambda_j \psi_j$.
Compute $\overline{\Sigma}$ from (\ref{eq::fp}).
\item Output $(\overline\mu,\overline\Sigma,\overline\psi)$.
\end{enum}
}}
\caption{Barycenter algorithm.}
\label{fig::barybary}
\end{figure}

\parskip 10pt

{\bf A Further Speed-Up Using Hypothesis Testing.}
The most expensive step in computing $\hat H^2$ is
estimating the transport map $T_j$.
Sometimes
$\tilde{P}_j$ is very close to $R$
and there is no need to compute
$T_j$.
We can check this formally by testing
$H_0: \tilde{P}_j=R$
using any convenient two-sample hypothesis test.
If the test rejects, we compute $\psi_j$.
If the test fails to reject we just set $\psi_j =0$.
A convenient nonparametric test 
with a distribution-free null distribhtion
is the cross-match test 
(\cite{rosenbaum2005exact}).
This idea is illustrated in Section
\ref{section::speedup}.

\section{$k$-means Distribution Clustering}
\label{section::kmeans}

Now we are ready to discuss the distribution clustering problem.
For concreteness, we focus here on $k$-means clustering.
In Section \ref{section::other}
we briefly consider other clustering methods.

The general outline is as follows.
Fix an integer $k$.
Then:
\begin{enum}
\item Choose $k$ distributions $c_1,\ldots, c_k$ as starting points
using ${\rm k-means}^{++}$
\citep{arthur2007k}.
\item Assign each distribution to the nearest centroid:
$$
{\cal C}_j = \Bigl\{ P_s:\ d(P_s,c_j) < d(P_s,c_\ell) \ 
{\rm for\ all\ }\ell \neq j\Biggr\}.
$$
\item Compute the barycenter $c_j$ of ${\cal C}_j$ putting equal 
weight on each distribution in the cluster.
\item Repeat steps 2 and 3 until convergence.
\end{enum}

We consider four versions:
\begin{enum}
\item {\bf Exact}: Use the Wasserstein distance as the distance at each step.
Given the computational burden, we only use the exact method
in one-dimensional examples
as a point of comparison.

\item {\bf Euclidean}. Compute $W(P_i,P_j)$ for each pair.
Use multidimensional scaling to embed the distributions into
$\mathbb{R}^a$. Thus we map each $P_j$ to a vector
$V_j \in \mathbb{R}^a$.
For ease of visualization, we use $a=2$.
We note that the set of distributions equipped with
Wasserstein distance is not isometrically embeddable into $\mathbb{R}^a$
and so there will necessarily be some distortion.
Once we have the vectors
$V_1,\ldots, V_N$ we proceed with $k$-means clustering as usual.

\item {\bf Gaussian}.
We use $W(Z_i,Z_j)$ as an approximation
of $W(P_i,P_j)$.
Barycenters are computed by averaging the means and by
iterating the fixed point equation
(\ref{eq::fp}) for the variances.

\item {\bf Hybrid.}
We use our Hybrid distance $H(P_i,P_j)$.
The steps of this approach are in Figure \ref{fig::main}.
\end{enum}

Of course, the fourth method is the focus
of this paper.
We include the others for comparison.

\begin{figure}
\fbox{\parbox{6in}{
\begin{center}
{\sf Hybrid Wasserstein $k$-Means}
\end{center}
\begin{enum}
\item Choose starting centroid distributions:
\begin{enum}
\item Choose an integer $j$ at random from $\{1,\ldots, N\}$. Set
${\cal C} = \{P_j\}$.
\item Compute $H_i^2 = \min_{P\in {\cal C}} H^2(P_i,P)$ for $i=1,\ldots, N$.
\item Choose a new $j\in \{1,\ldots, N\}$ with probability proportional 
to $H_j^2$.
Set ${\cal C} = {\cal C}\cup\{j\}$.
\item Repeat the previous steps until 
${\cal C}=\{c_1,\ldots, c_k\}$ contains $k$ distributions.
\end{enum}
\item Assign each $P_j$ to the closest centroid in ${\cal C}$.
This defines clusters
$C_1,\ldots, C_k$ where
$C_j = \{P:\ H(P,c_j) < H(P,c_s),\ s\neq j\}$.
\item Compute the barycenter $c_j$ of $C_j$.
\item Repeat steps 2 and 3 until convergence.
\item Return $C_1,\ldots, C_k$.
\end{enum}}}
\caption{\em The hybrid $k$-means distribution clustering method.}
\label{fig::main}
\end{figure}

\section{Examples}
\label{section::examples}

In this section,
we consider some examples.
In all the examples that follow,
we visualize the set of distributions by
applying multidimensional scaling to the pairwise distances.
Each point in the plots represents one distribution.

\subsection{One Dimensional Examples}

This section presents three one-dimensional examples.
These examples are chosen to illustrate
the behaviour 
of the four versions of $k$-means clustering 
methods, presented in Section \ref{section::kmeans}, and,
more specifically, for comparing the performance of the
hybrid distance of Section \ref{section::modified}
with the procedure based on Wasserstein distance. 
When $d=1$, the exact Wasserstein clustering can be 
easily computed. (The algorithm is in the appendix.)

The first example consists of 15 Normally distributed data 
sets of size $n=100$, all with variance $\sigma^2=1$. There
are three groups of five data sets each with means close 
together. This specific simple example is chosen to 
show that all four versions of $k$-means clustering methods 
identify the clusters correctly, in straightforward cases.
Figure \ref{fig::Ex1-1dClu} shows the multidimensional scaling 
projection of the pairwise distances, with different colors
used to indicate the clusters.

\begin{figure}
\vspace{-1in}
\begin{center}
\includegraphics[height=3.8in]{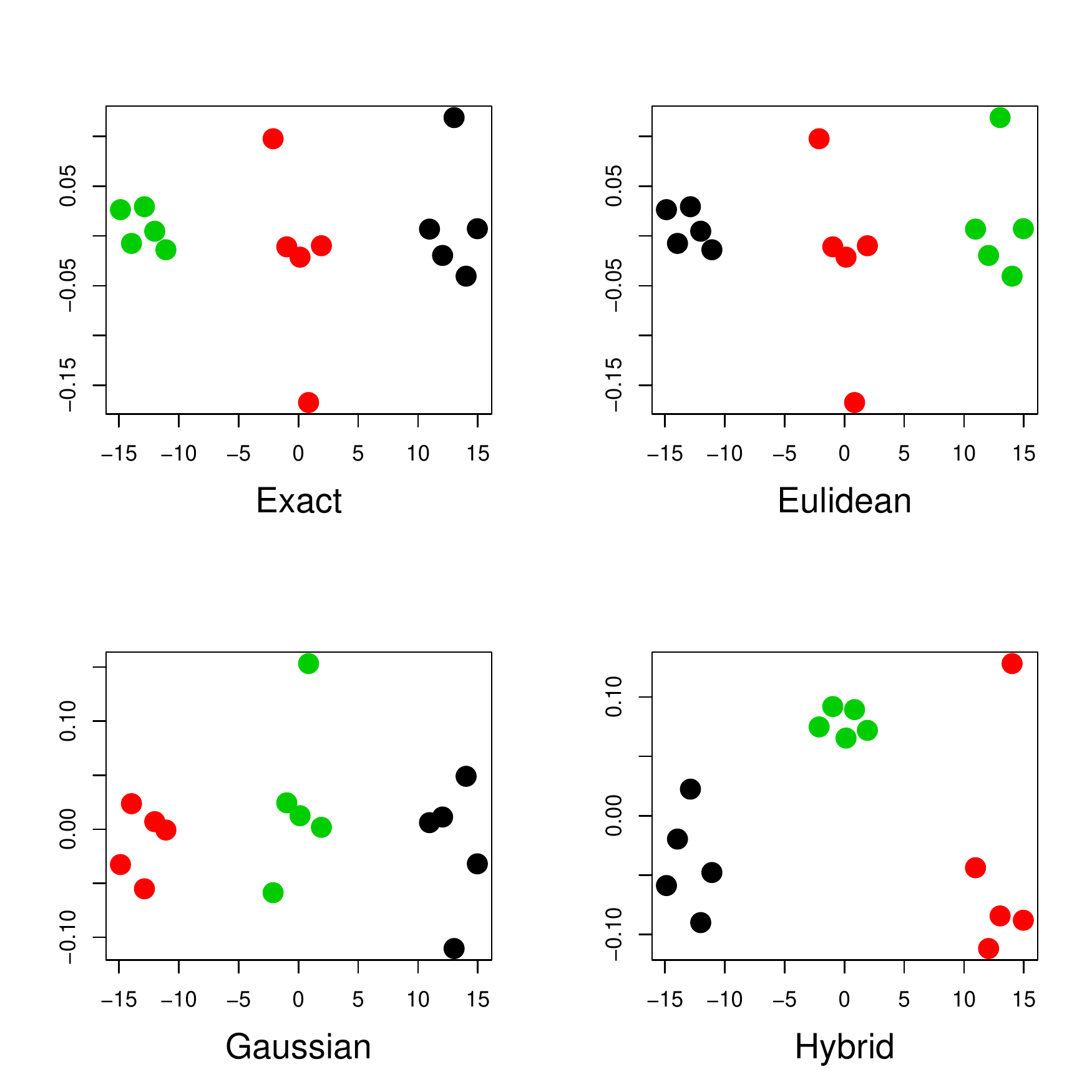}
\end{center}
\caption{\em Example 1. The plots show pairwise distances
as represented by multi-dimensional scaling. The plots show three clusters of 
five normal distribution each. All four methods are accurate.}
\label{fig::Ex1-1dClu}
\end{figure}

Our second and third examples are chosen to illustrate 
that less clean clusters of data-sets cannot be well
identified by the Gaussian approximation. Example 2 consists 
of two clusters.
The first cluster 
has twenty normal distributions
with mean $0$ and variance $1$. 
The second cluster has twenty distributions each of which
has the form $(1/2)\delta_{-1} + (1/2)\delta_{1}$.
This means that the first and second moments are also 0 and 1 so
the Gaussian approximation should be unable to find the 
two clusters.

Figure \ref{fig::Ex2Dis-1D} shows the distances computed with the 
Gaussian 
approximation and the Hybrid method, versus the Wasserstein
procedure. The $45^o$ degree line is included for illustrating that
the distances among distributions
between the Gaussian approximation and the Wasserstein method
are not close to each other, while the Hybrid and Wasserstein procedures
displays similar distances.
This plot complements the following Figure 
\ref{fig::Ex2Clu-1D}, showing the clustering
obtained in this example, and illustrates that distances from
the Gaussian approximation are a poor approxmation.
Figure \ref{fig::Ex2Clu-1D} shows indeed that, as expected,
the Gaussian approximation does 
not work well.

\begin{figure}
\begin{center}
\includegraphics[scale=.6]{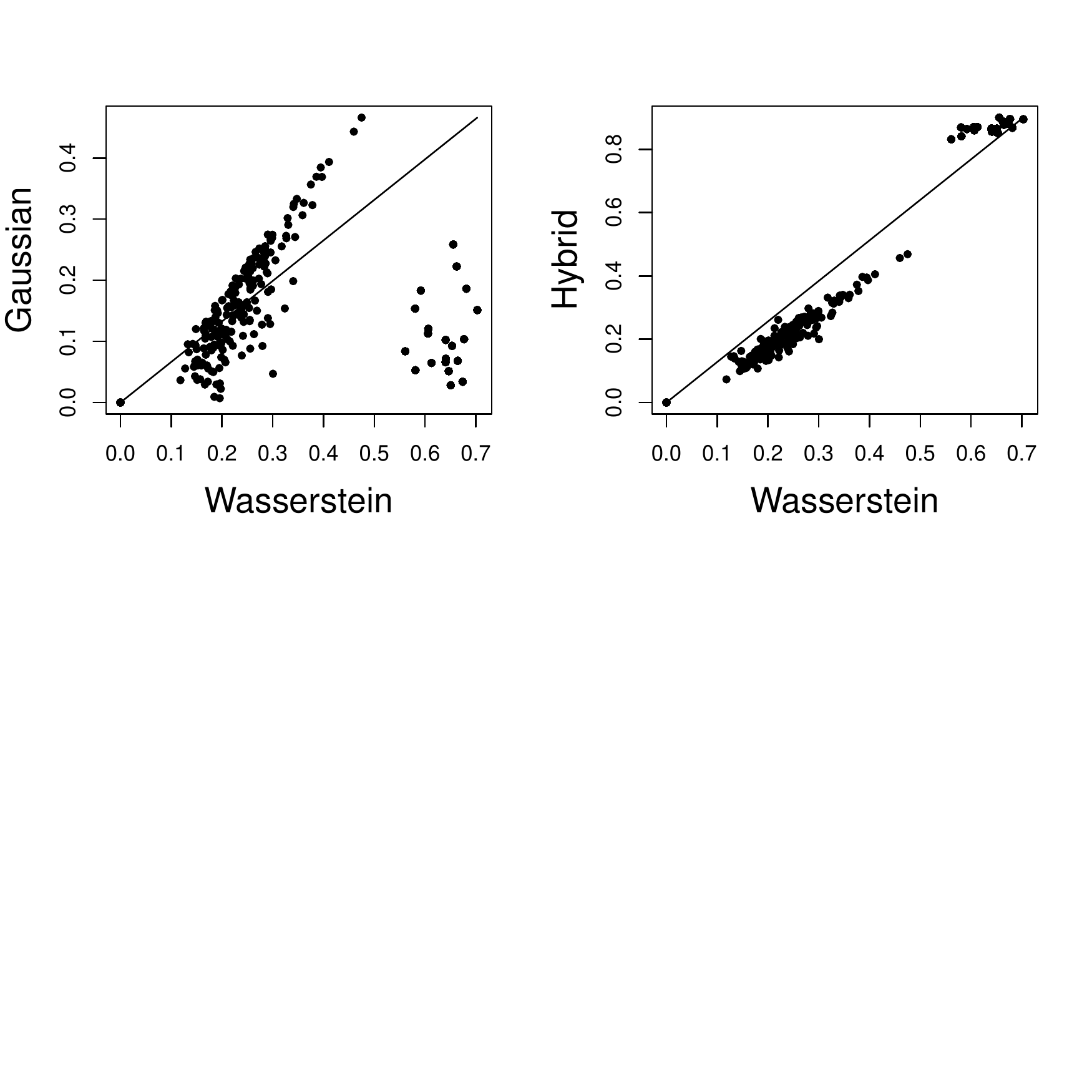}
\end{center}
\vspace{-2in}
\caption{\em Example 2. Plots of Gaussian and Hybrid 
distances of two data sets versus Wasserstein distance.
The line across the plots is the $45^o$ line. The left
hand side plot shows the distances among distriutions computed 
with the normal approximations are far from the Wasserstein distance.}
\label{fig::Ex2Dis-1D}
\end{figure}

\begin{figure}
\begin{center}
\includegraphics[scale=.6]{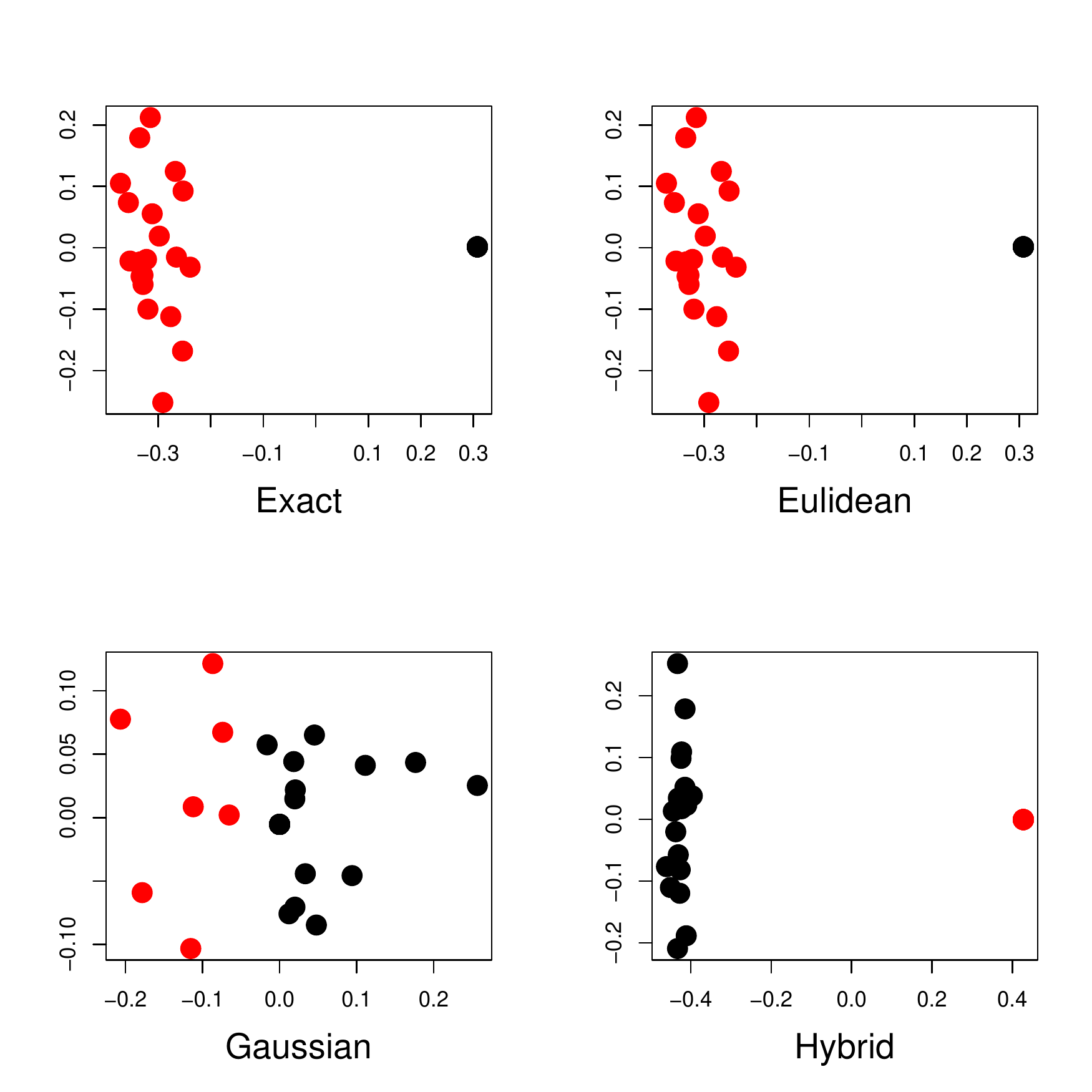}
\caption{\em Example 2. While the Gaussian
approximation does detect the two clusters, they are not well 
identified. The projection plots show random points in the 
plane. The other three methods, instead, identify the clusters 
correctly.}
\label{fig::Ex2Clu-1D}
\end{center}
\end{figure}

Our third example consists of three clusters of distributions,
each being a mixture of normal distributions. This example 
also shows that the Gaussian approximation does not identify
the clusters.
As in Example 2, we first display the plots of the 
distances, in Figure \ref{fig::Ex3Dist-1D}, and then
Figure \ref{fig::Ex3Clu-1D} with plots of  
the multidimensional 
scaling projections. This confirms the issue noted earlier
for the Gaussian approximation. Note that the clusters 
displayed in some plots are very close together and not 
easily detectable in the picture. But it is clear that
the Gaussian approximation does not fare well.

\begin{figure}
\begin{center}
\includegraphics[scale=.6]{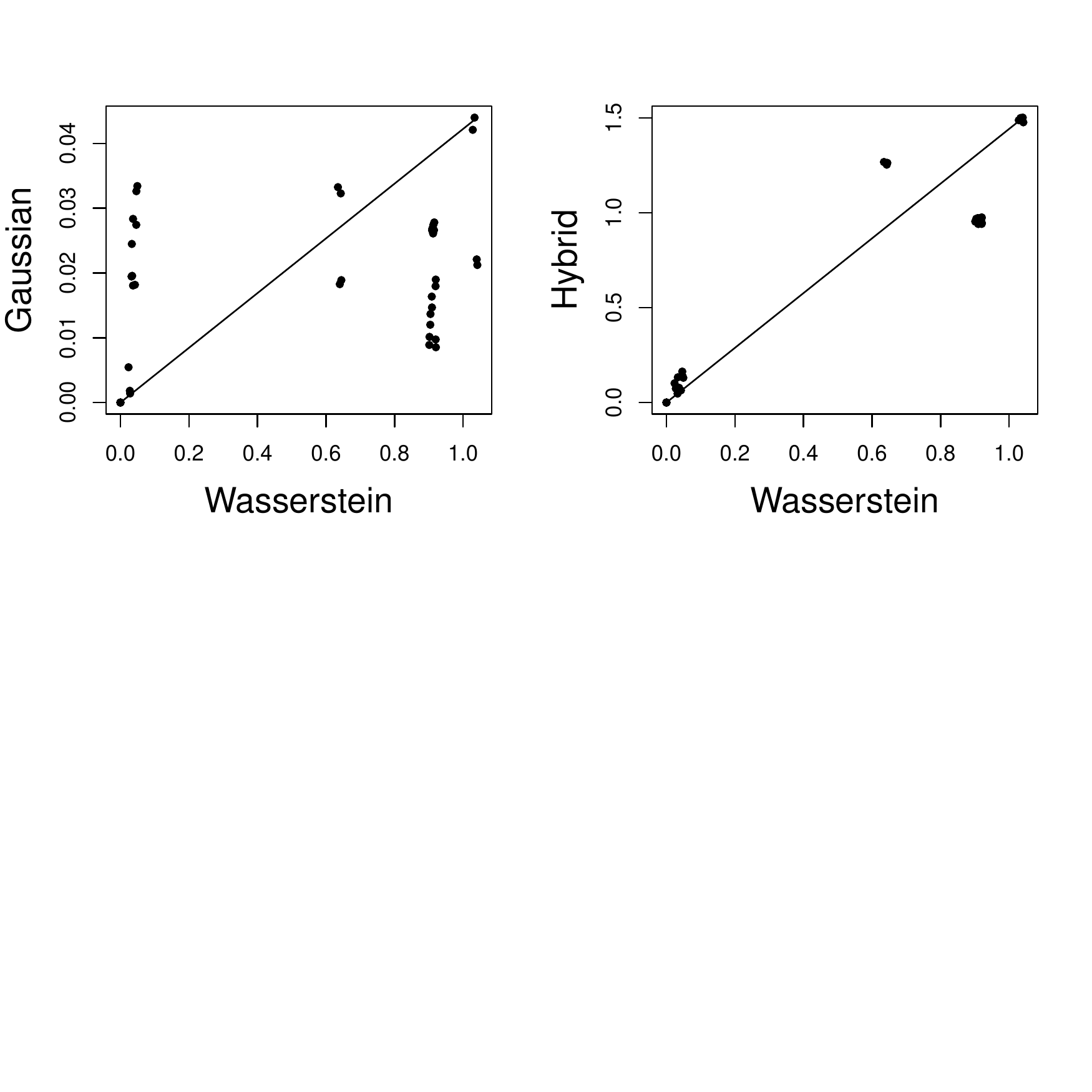}
\end{center}
\vspace{-2.2in}
\caption{\em Example 3. Plots of Gaussian and Hybrid 
distances of three data sets versus Wasserstain distance.
The left hand side plot also shows distances among distriutions 
computed by the Gaussian approximations are far from the 
Wasserstein distance.}
\label{fig::Ex3Dist-1D}
\end{figure}

\begin{figure}
\vspace{-.2in}
\begin{center}
\includegraphics[scale=.6]{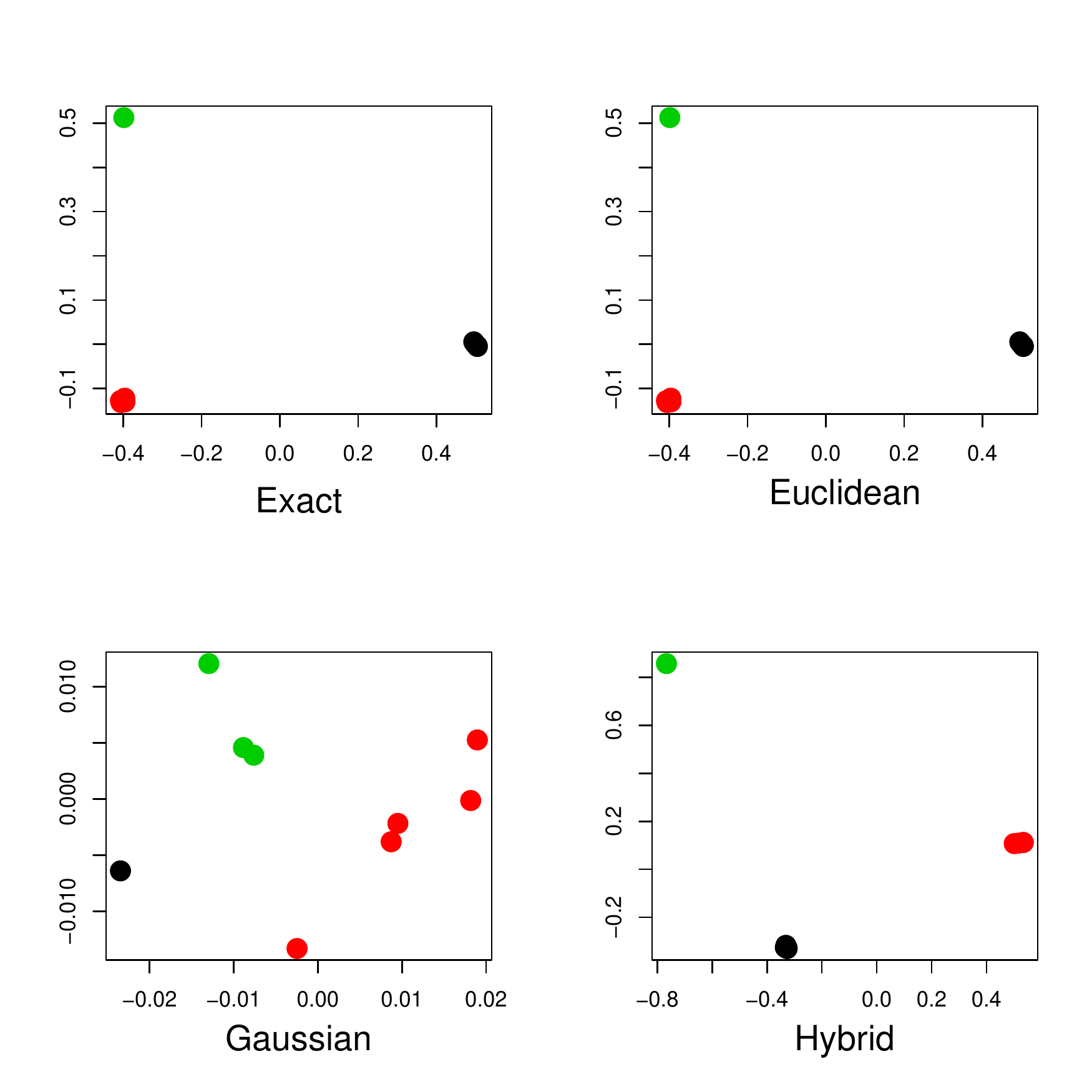}
\end{center}
\vspace{-.2cm}
\caption{\em Example 3. Clusters 
for mixtures of Normal distributions. The Gaussian approximation 
does not performs well.}
\label{fig::Ex3Clu-1D}
\end{figure}

\subsection{Multivariate Examples}
\label{section::MvarExamples}

In this section we consider three artificial data-sets in 
two dimensions. We no longer consider the exact Wasserstein method which
is computationally prohibitive.
We obtain, instead, clustering 
from the other three methods of Section \ref{section::kmeans}. 
For each example, we will plot the parwise distances
between datasets using multi-dimensional scaling.

The first example consists of $40$ bivariate Normal distributions,
with $n=100$ observations each. Twenty of them have mean $0$ and 
variance $1$, the other twenty have meann $5$ and variance $1$. 
We also include twenty bivariate uniform distributions each with 
$100$ data points, for a total of three distributional clusters.

Figure \ref{fig::Ex1-Boxp-2D} displays the boxplots of the first 
coordinates of the three data sets. The clusters appear to be well 
separated as expected. All 
three methods performs well, as it can be seen from the clusters
obtained from the Euclidean, Normal, and Hybrid methods in Figure 
\ref{fig::Ex1-Clust-2D}.

\begin{figure}
\begin{center}
\includegraphics[scale=.4]{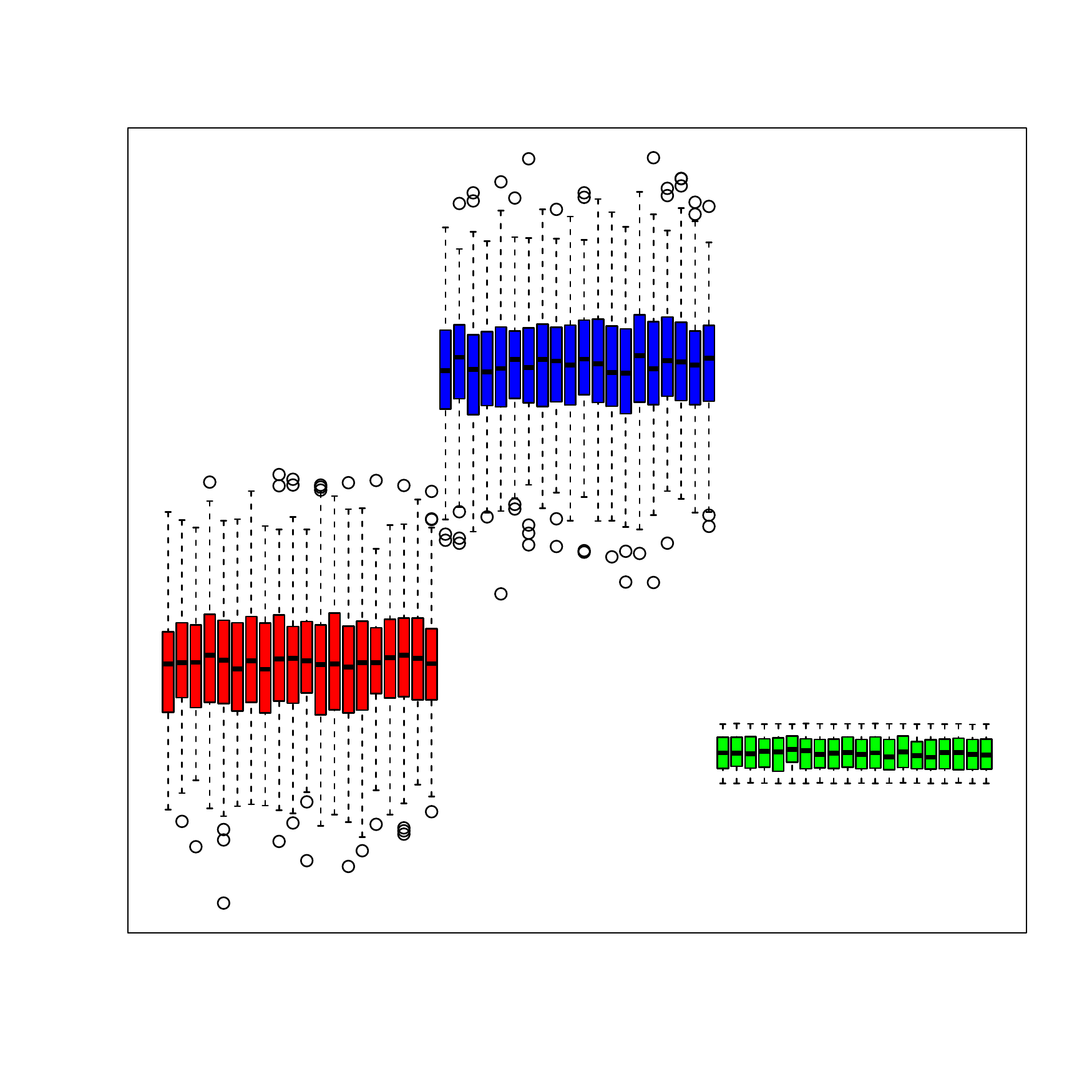}
\end{center}
\vspace{-.2in}
\caption{\em Boxplots of first coordinates in the bivariate 
example 1.}
\label{fig::Ex1-Boxp-2D}
\end{figure}

\begin{figure}
\begin{center}
\begin{tabular}{ccc}
\includegraphics[scale=.32]{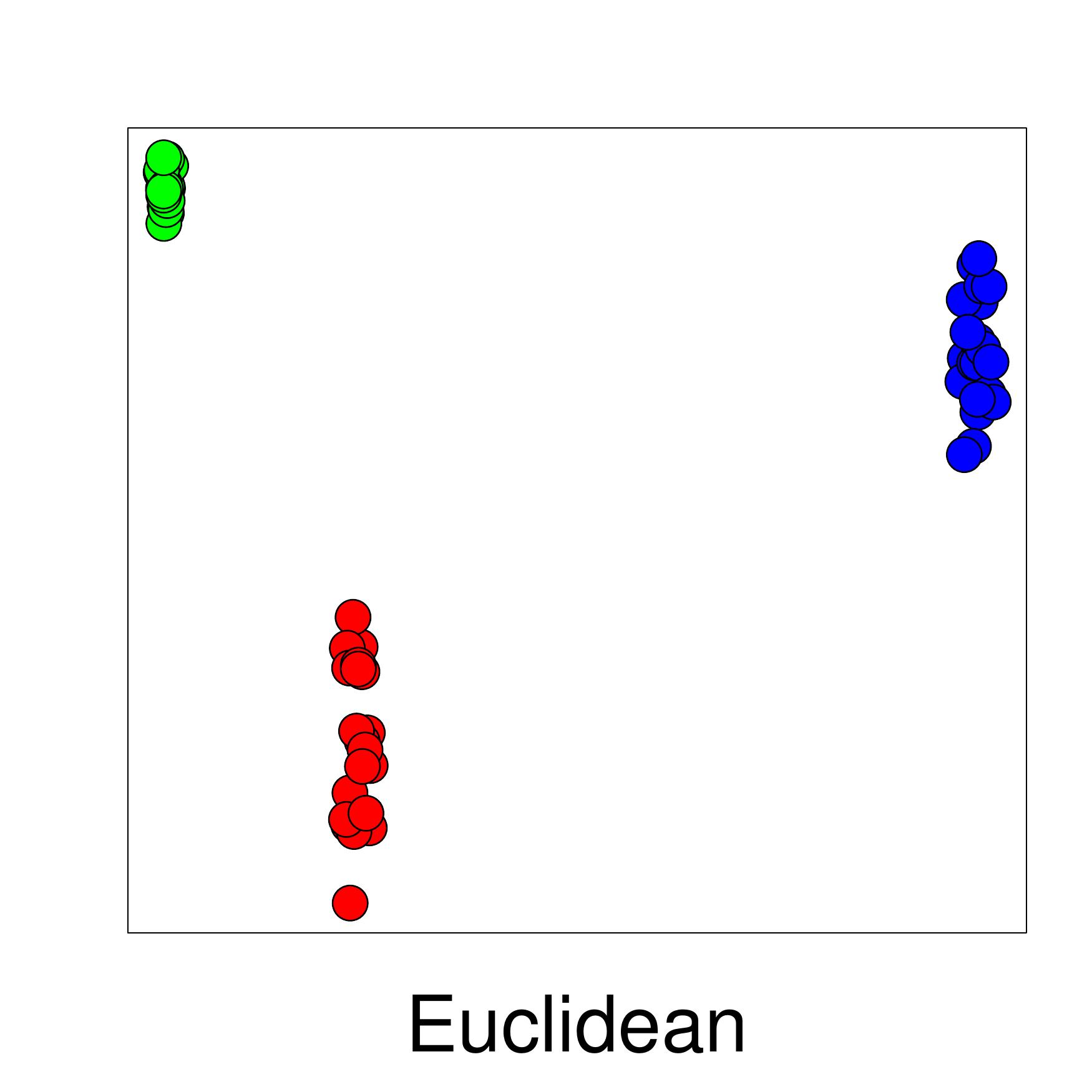}
\includegraphics[scale=.32]{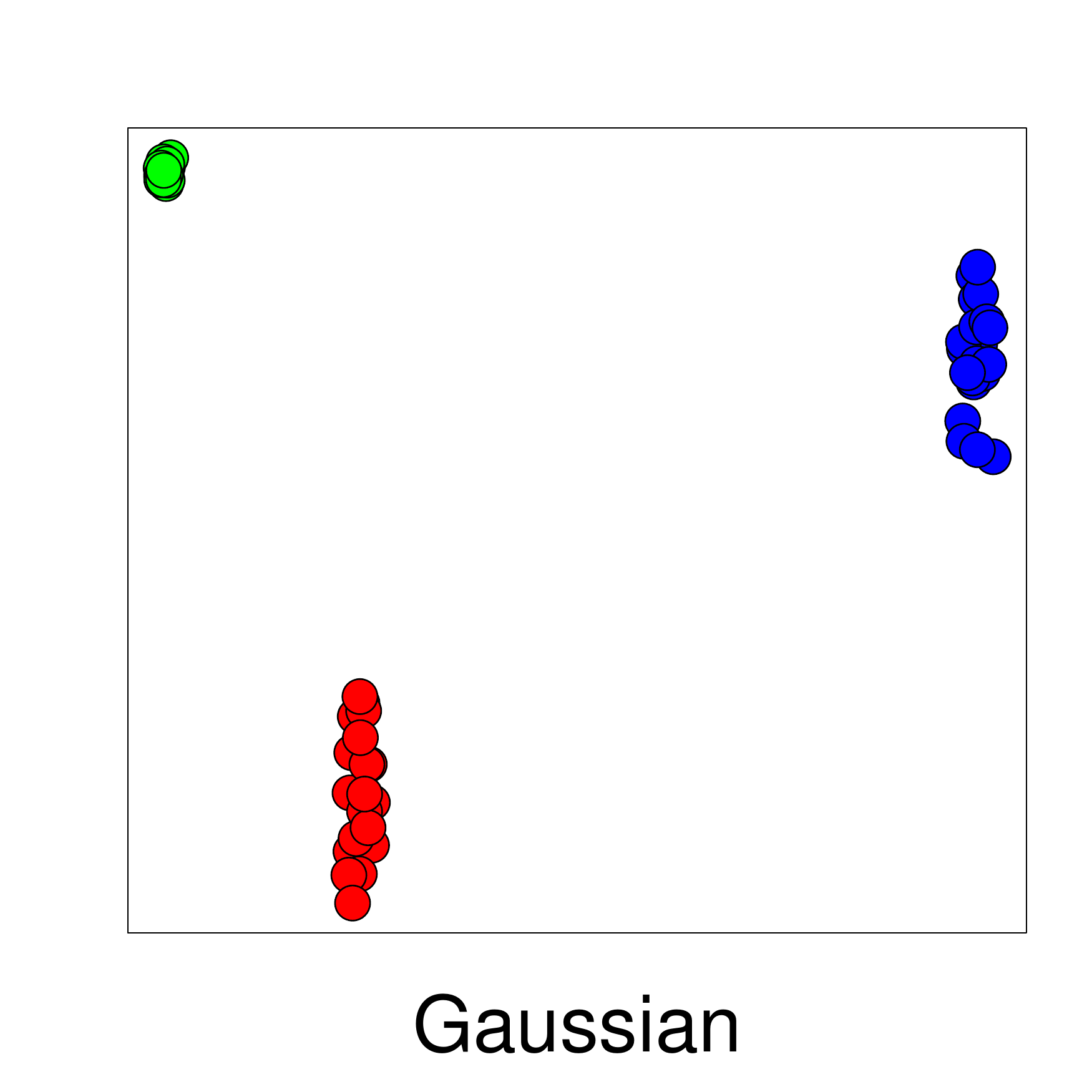}
\includegraphics[scale=.32]{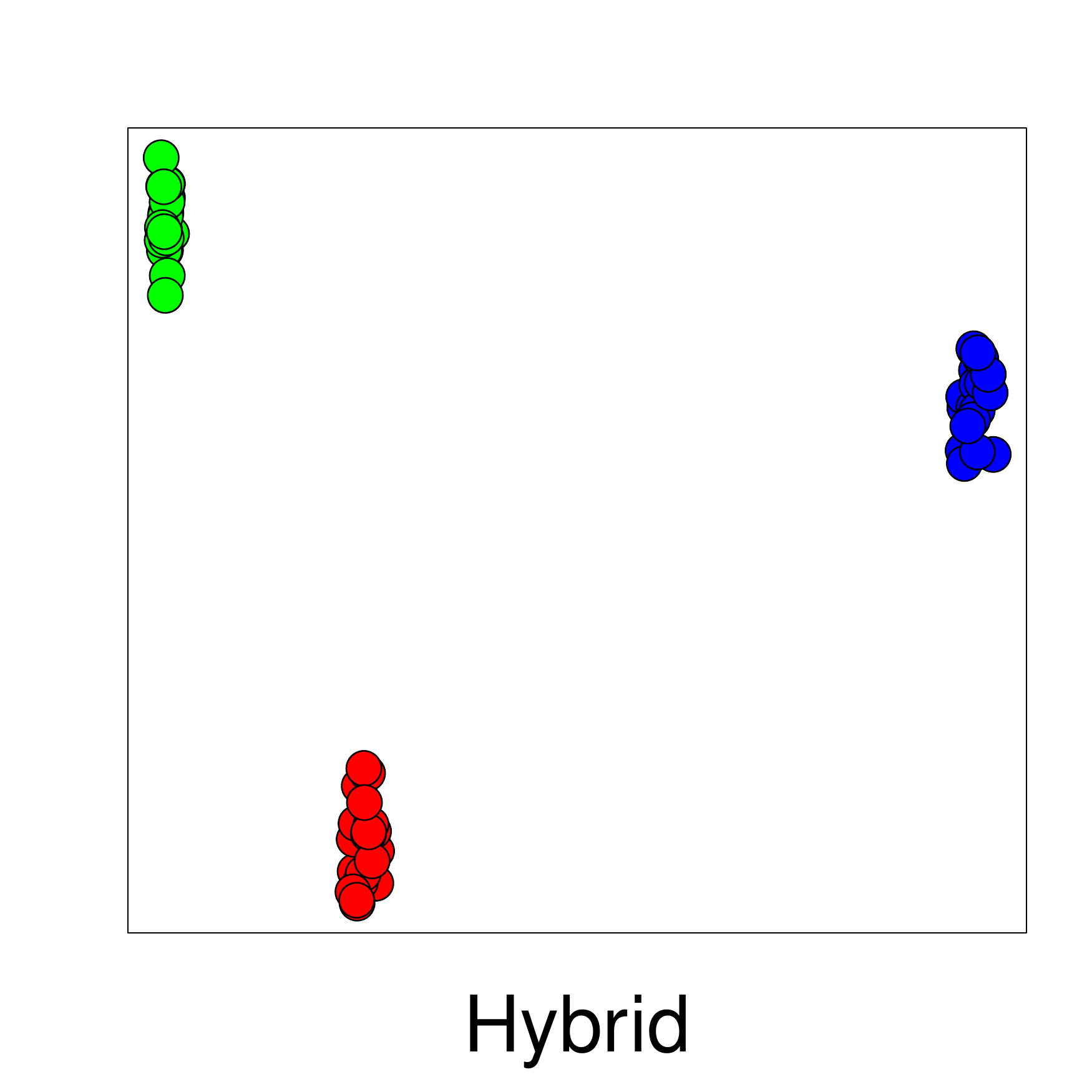}
\end{tabular}
\end{center}
\caption{\em Bivariate Example 1. Clusters obtained 
using the  procedures Gaussian and  Hybrid.}
\label{fig::Ex1-Clust-2D}
\end{figure}

The second example consists of four groups of distributions 
which are uniform on circles, each 
with $n=100$ data points. The circles' centers and
radii are randomly selected. More specifically the centers 
are generated from uniform distributions, with various ranges,
as are their radii. Figure \ref{fig::Ex2-Circles-2D} displays
the data that form four well separated clusters. 
Figure \ref{fig::Ex2-Clust-2D} shows the results for the 
three methods of clustering. All of them identify the four clusters 
correctly.

\begin{figure}
\begin{center}
\includegraphics[scale=.5]{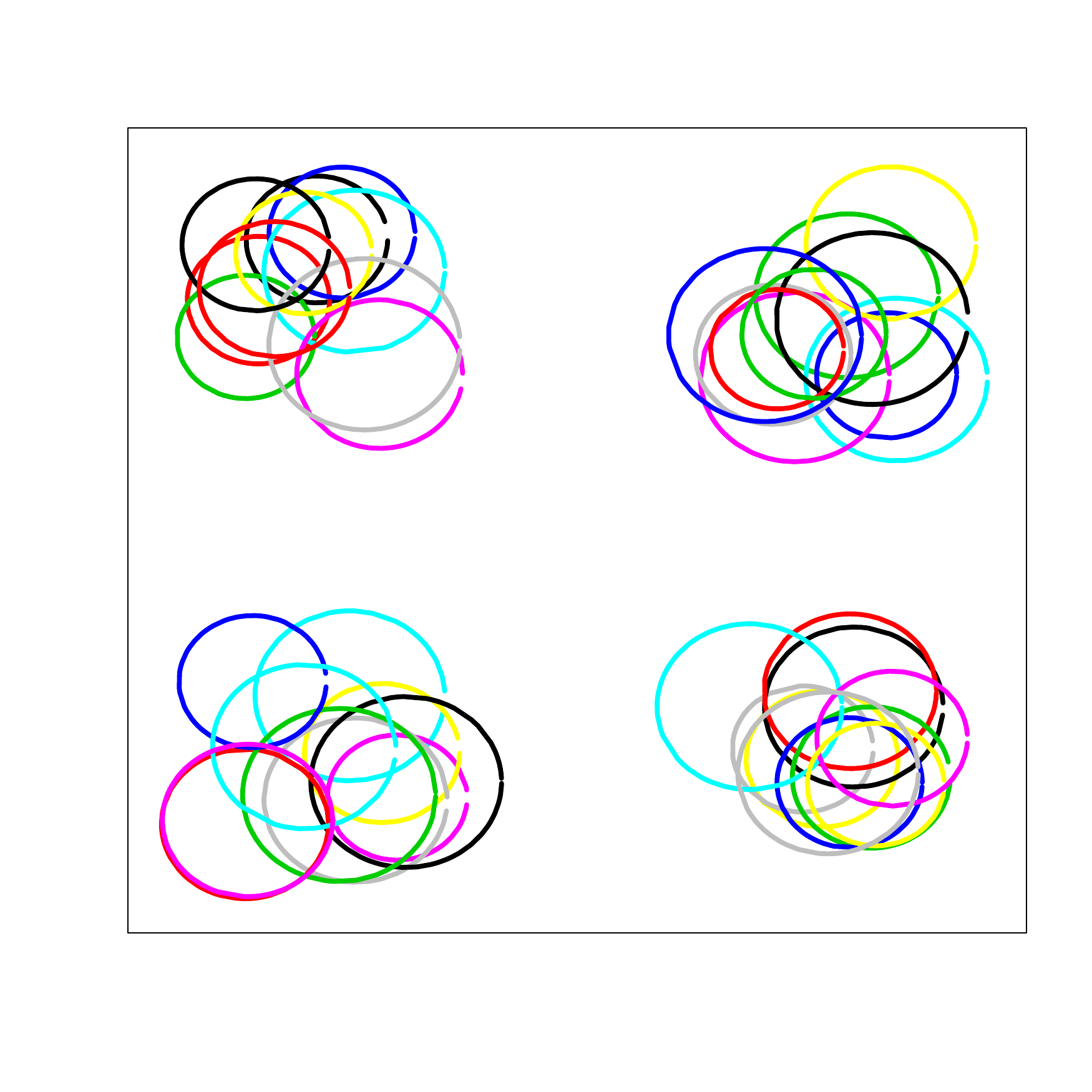}
\end{center}
\vspace{-.3in}
\caption{\em Bivariate Example 2. The datasets.}
\label{fig::Ex2-Circles-2D}
\end{figure}

\begin{figure}
\begin{center}
\begin{tabular}{ccc}
\includegraphics[scale=.32]{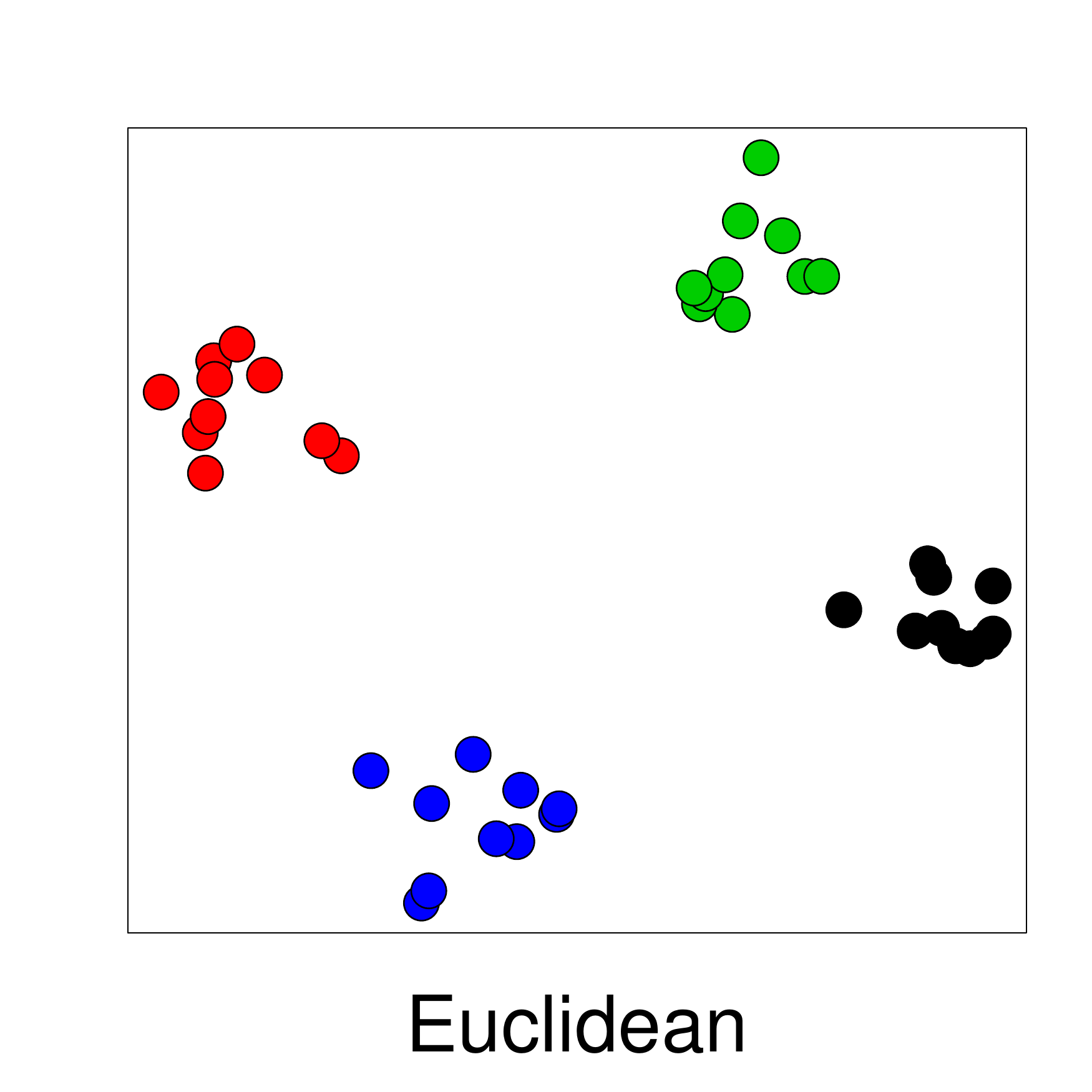}
\includegraphics[scale=.32]{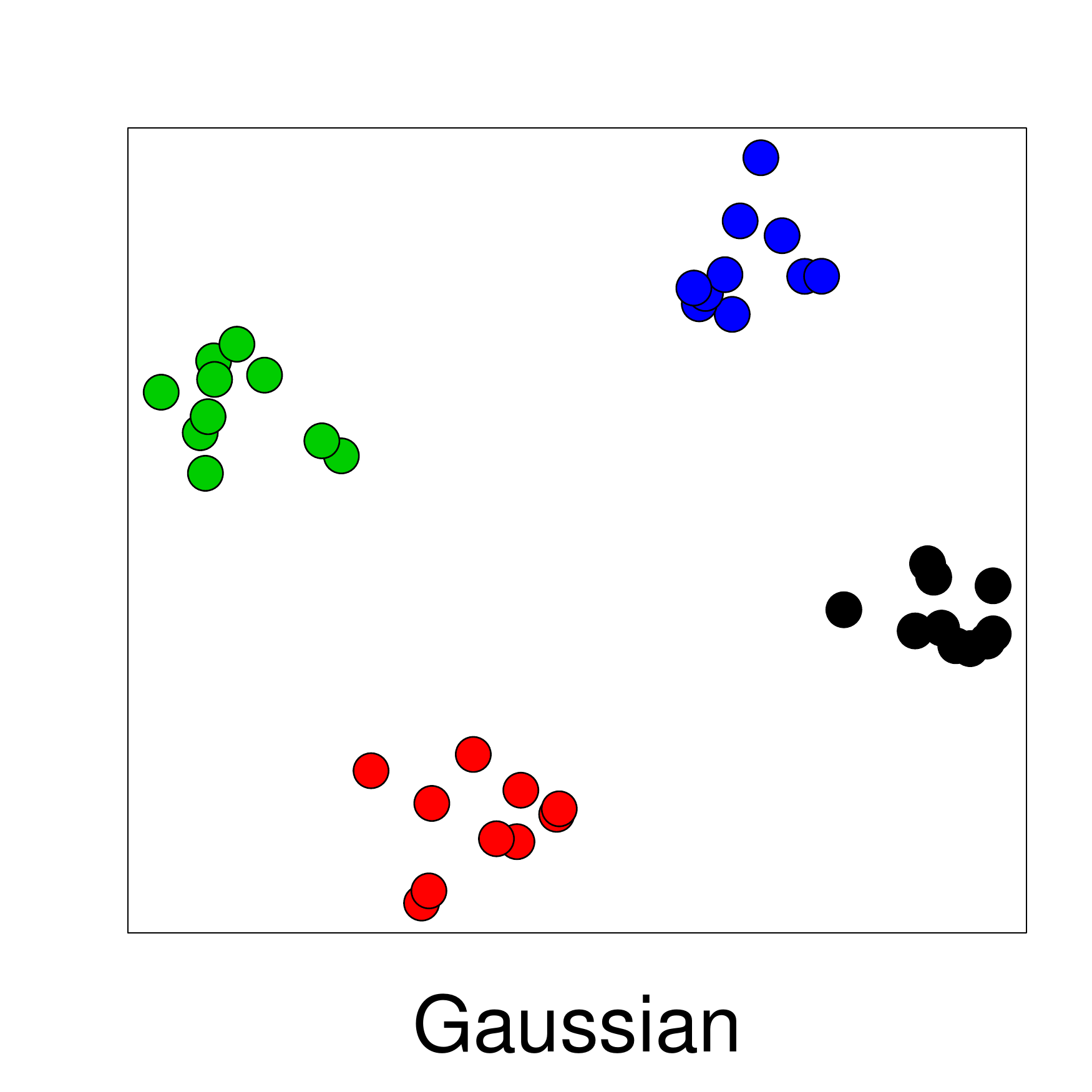}
\includegraphics[scale=.32]{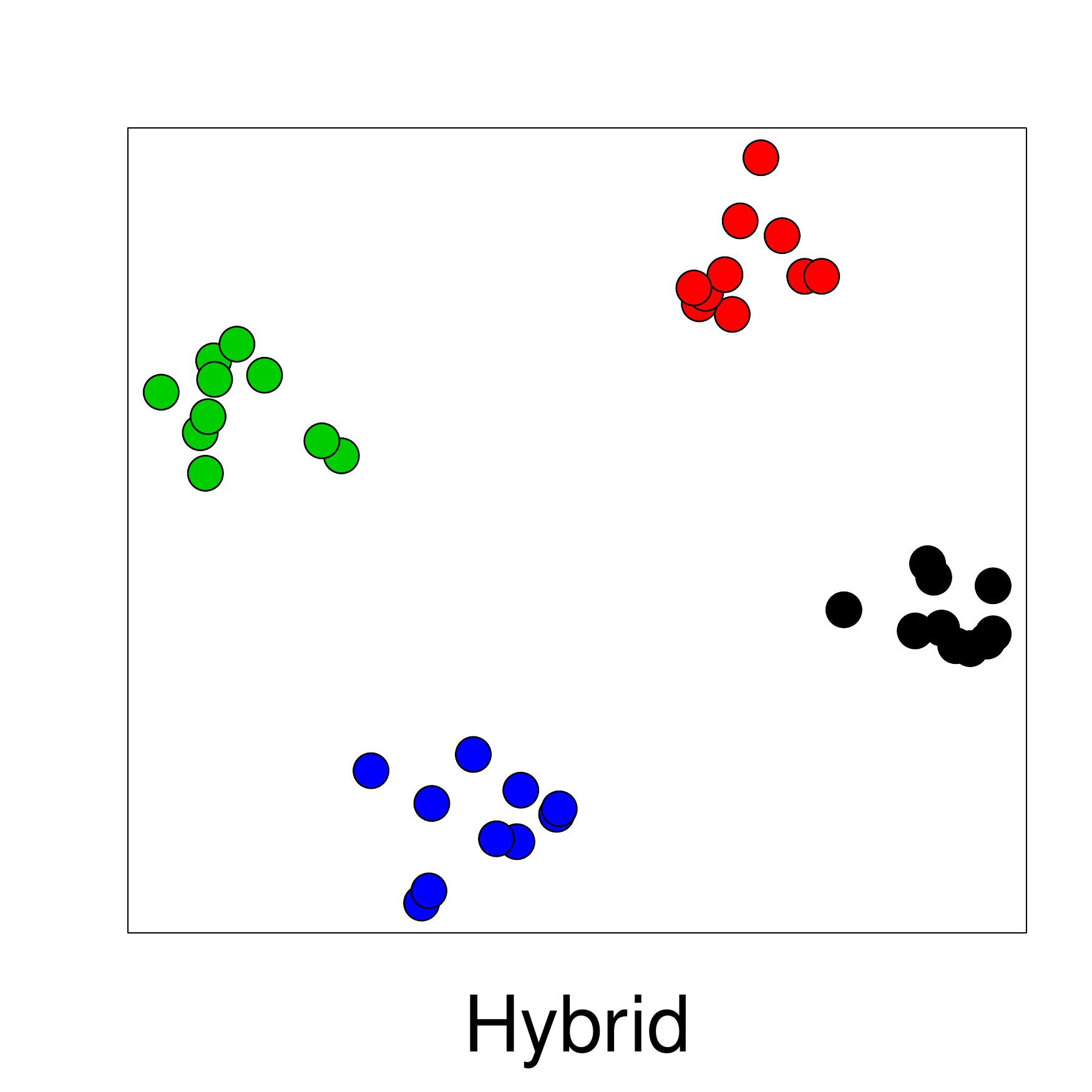}
\end{tabular}
\end{center}
\vspace{-.2in}
\caption{\em Bivariate Example 2. Clusters.}
\label{fig::Ex2-Clust-2D}
\end{figure}

Our third example consists of two groups of distributions. One group of 50
distributions, withg 100 observations each, are uniformely distributed on a 
cirle and the 50 distributions in the other group are normal with mean 
$0$ and variance $1$. Figure \ref{fig::Ex3-Boxplot-2D} shows the boxplot 
of their first coordinates. 
Figure \ref{fig::Ex3-Clust-2D} presents the results for the 
three clustering methods. As expected the Gaussian procedure 
performs quite poorly.

\begin{figure}
\begin{center}
\includegraphics[width=5in,height=3in]{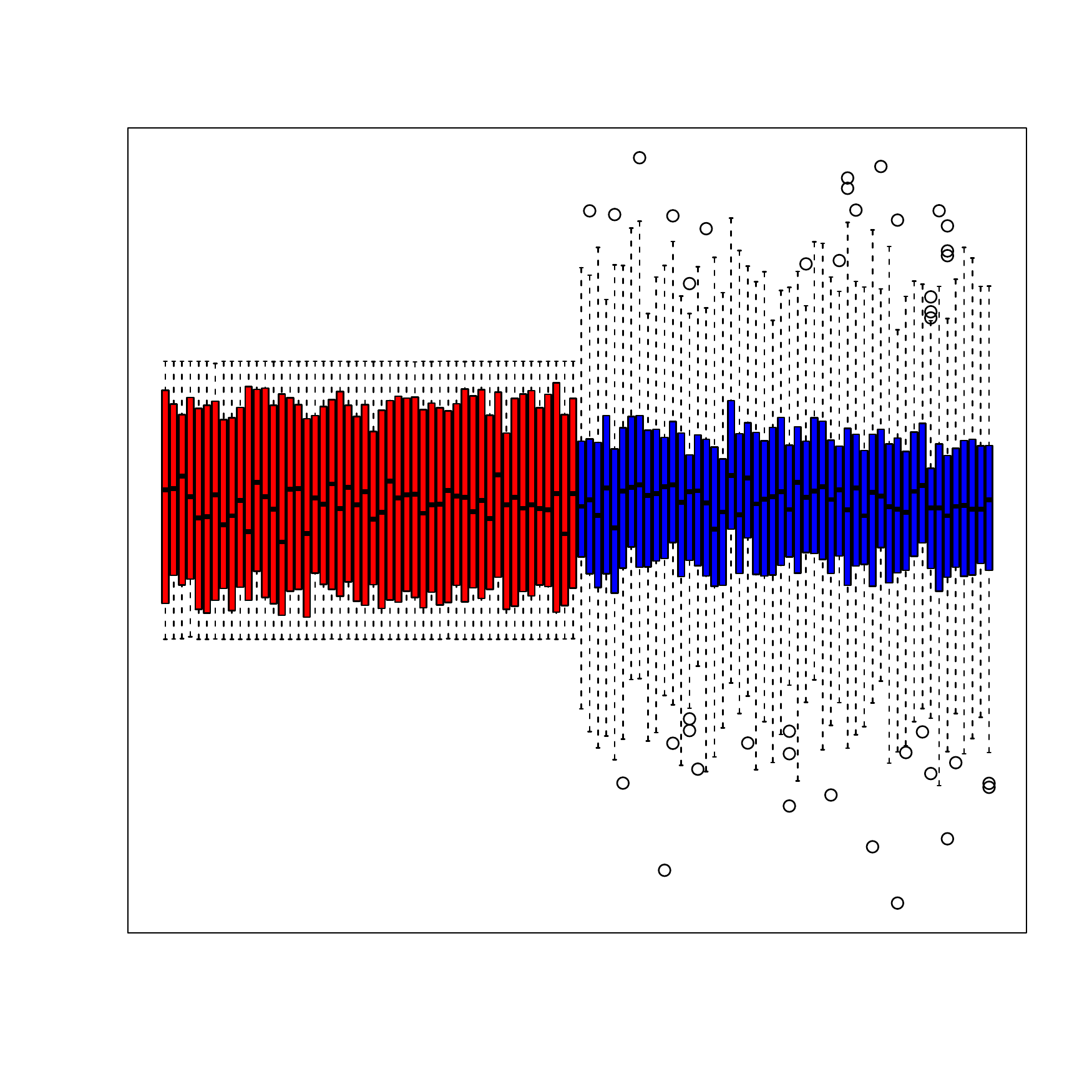}
\end{center}
\vspace{-.5cm}
\caption{\em Bivariate Example 3. Boxplots of first coordinates.}
\label{fig::Ex3-Boxplot-2D}
\end{figure}

\begin{figure}
\begin{center}
\begin{tabular}{ccc}
\includegraphics[scale=.3]{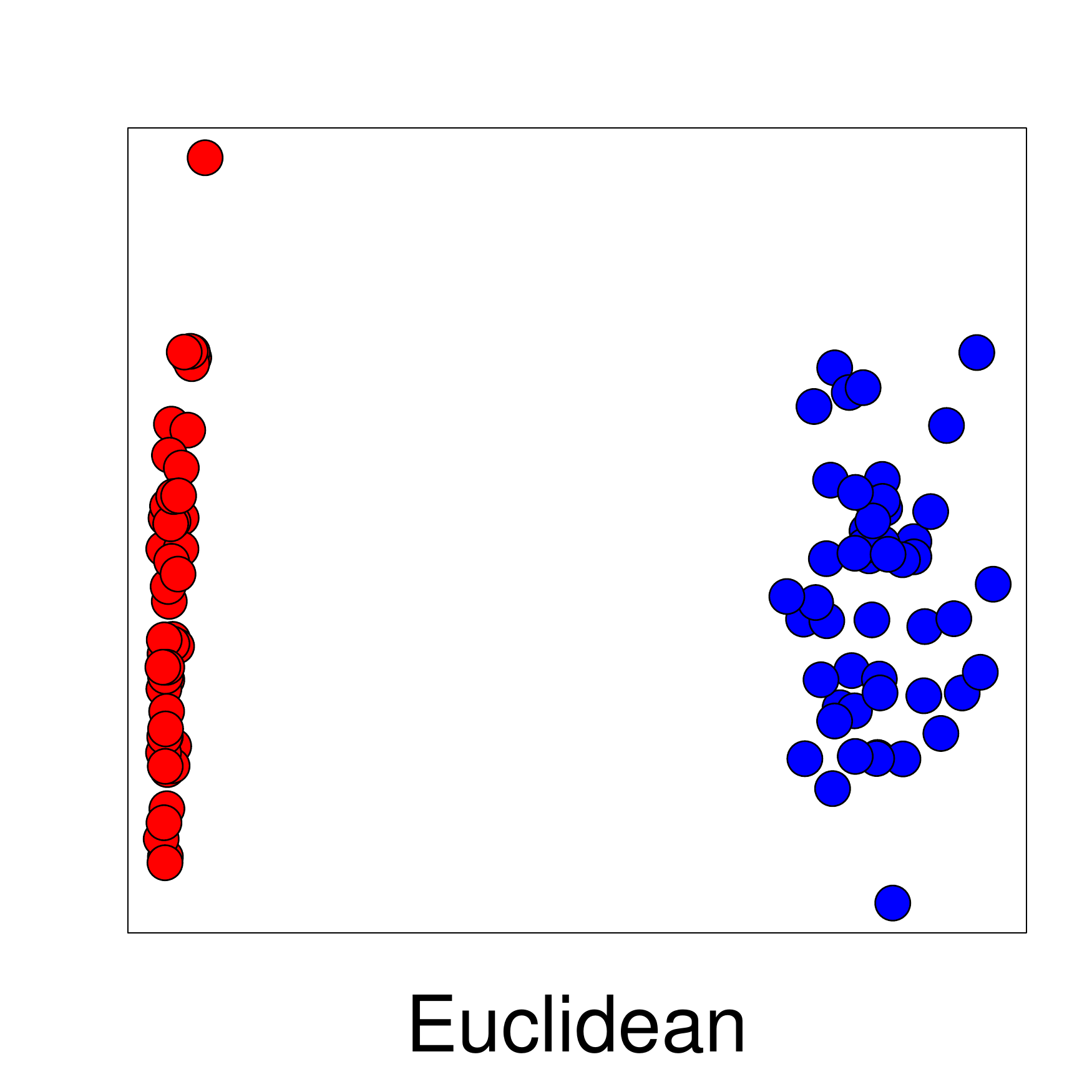}
\includegraphics[scale=.3]{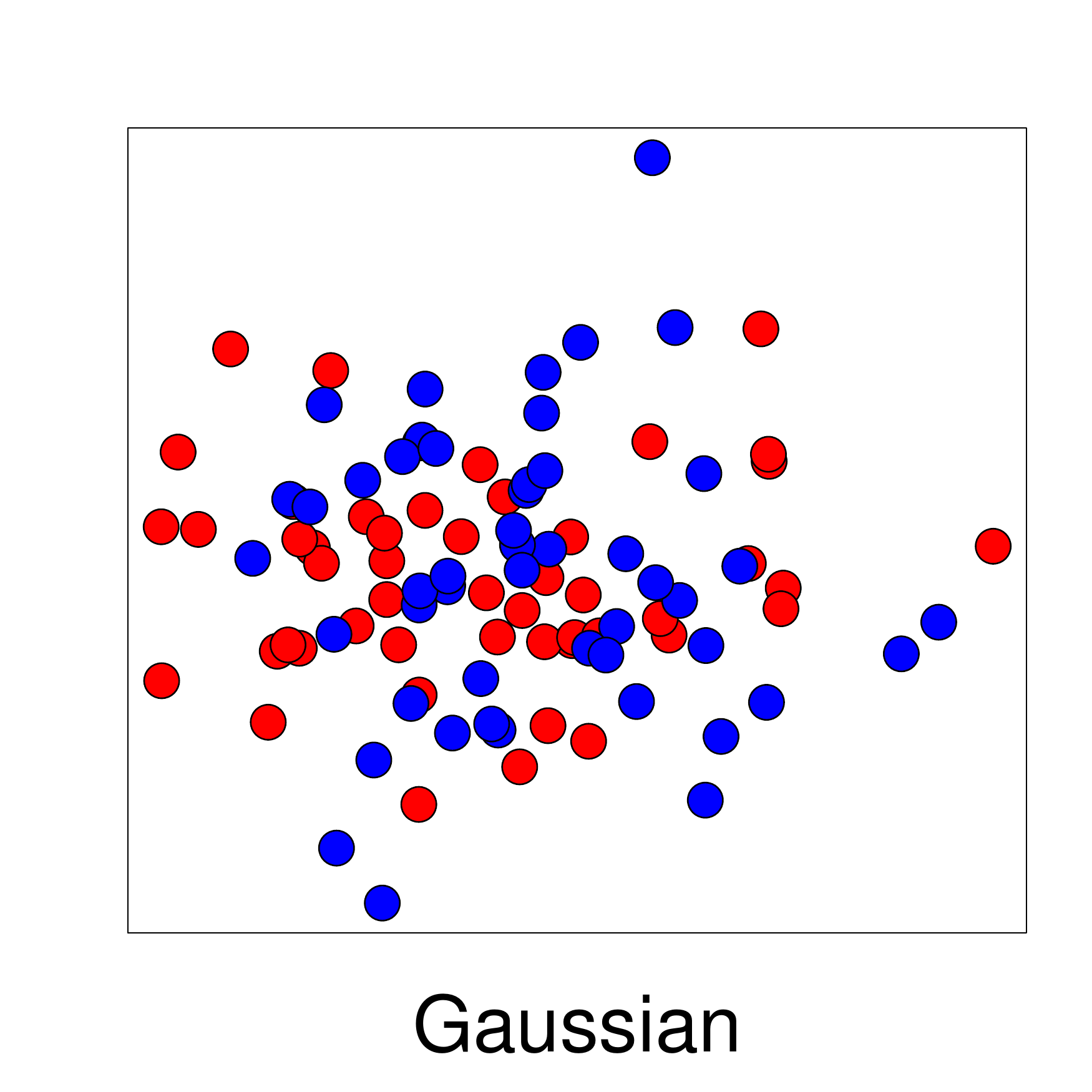}
\includegraphics[scale=.3]{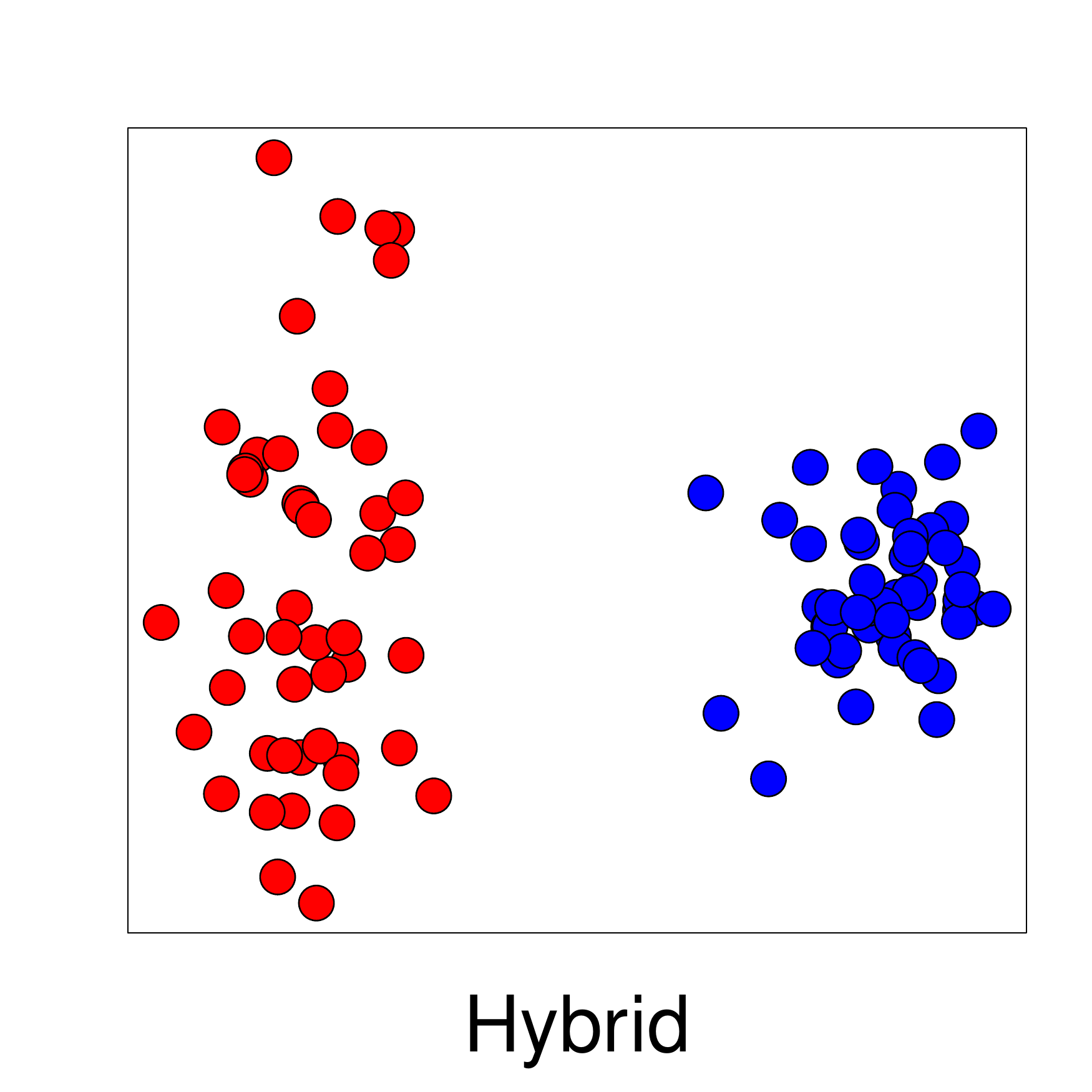}
\end{tabular}
\end{center}
\caption{\em Bivariate Example 3. Clusters.}
\label{fig::Ex3-Clust-2D}
\end{figure}

\subsection{Multivariate Example: Cytometric Data}

We now apply the Gaussian and Hybrid methods for clustering distributions, 
presented in Sections \ref{section::modified} and \ref{section::kmeans}
to a collection of data sets from cytometric genetic research by 
\cite{maier2007allelic}.
The authors obtained fluorescence intensity measures of
fluorophore-conjugated reagents on whole blood, stained with an antibody
marker. The data record the luminosity of four proteins linked to 
the T-cells, that are part of the adaptive immune system. Luminosity was 
measured on the proteins SLP76, ZAP70, CD4, and CD45RA before and after 
stimulation with the antibody anti-CD3. Two sets of blood samples were 
collected. One sample consists in 13 four-dimension data-sets, stained 
prior to anti-CD3 stimulation. The second sample, of 30 more data-sets, 
stained five minutes after stimulation, for a total of 43 data-sets in 
four dimension. Figure \ref{fig::CytoData} shows pairwise scatterplots 
for the first dataset.

\begin{figure}
\begin{center}
\includegraphics[width=4.in,height=4.in]{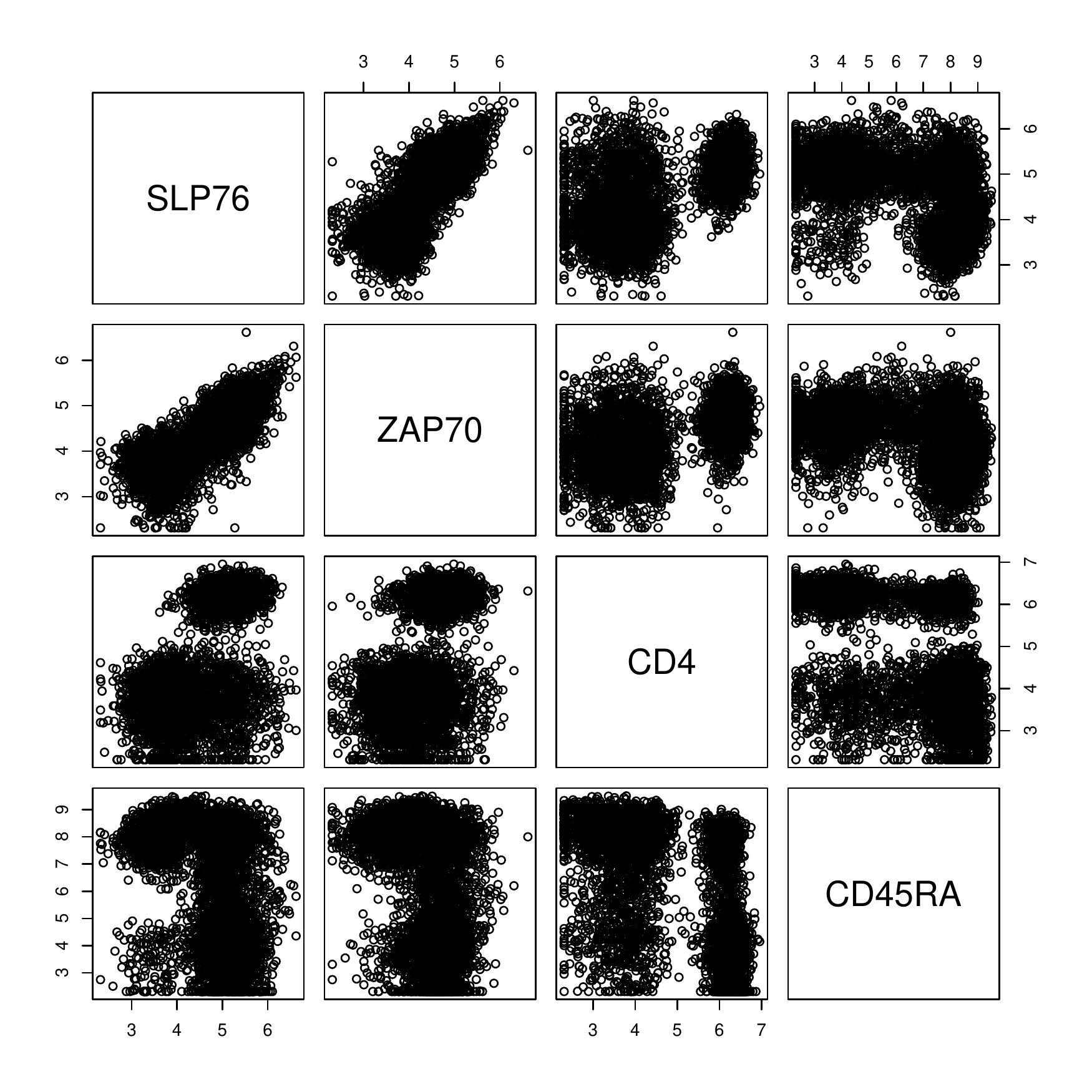}
\end{center}
\vspace{-.5cm}
\begin{center}
\caption{\em Cytometric Data. Scatterplots of joint luminosity 
recorded for the first distribution of 13 blood samples stained prior to 
the anti-CD3 stimulation.}
\end{center}
\label{fig::CytoData}
\end{figure}

In this example, for reasons explained in Section 
\ref{section::elbows}, we searched for six clusters.
Figure \ref{fig::CytoBoxplot1} and Figure \ref{fig::CytoBoxplot2} 
show the boxplots of the luminosity values recorded of the proteins,
SLP76, and ZAP70 in the 43 data-sets.

\begin{figure}
\begin{center}
\includegraphics[width=6.in,height=2.4in]{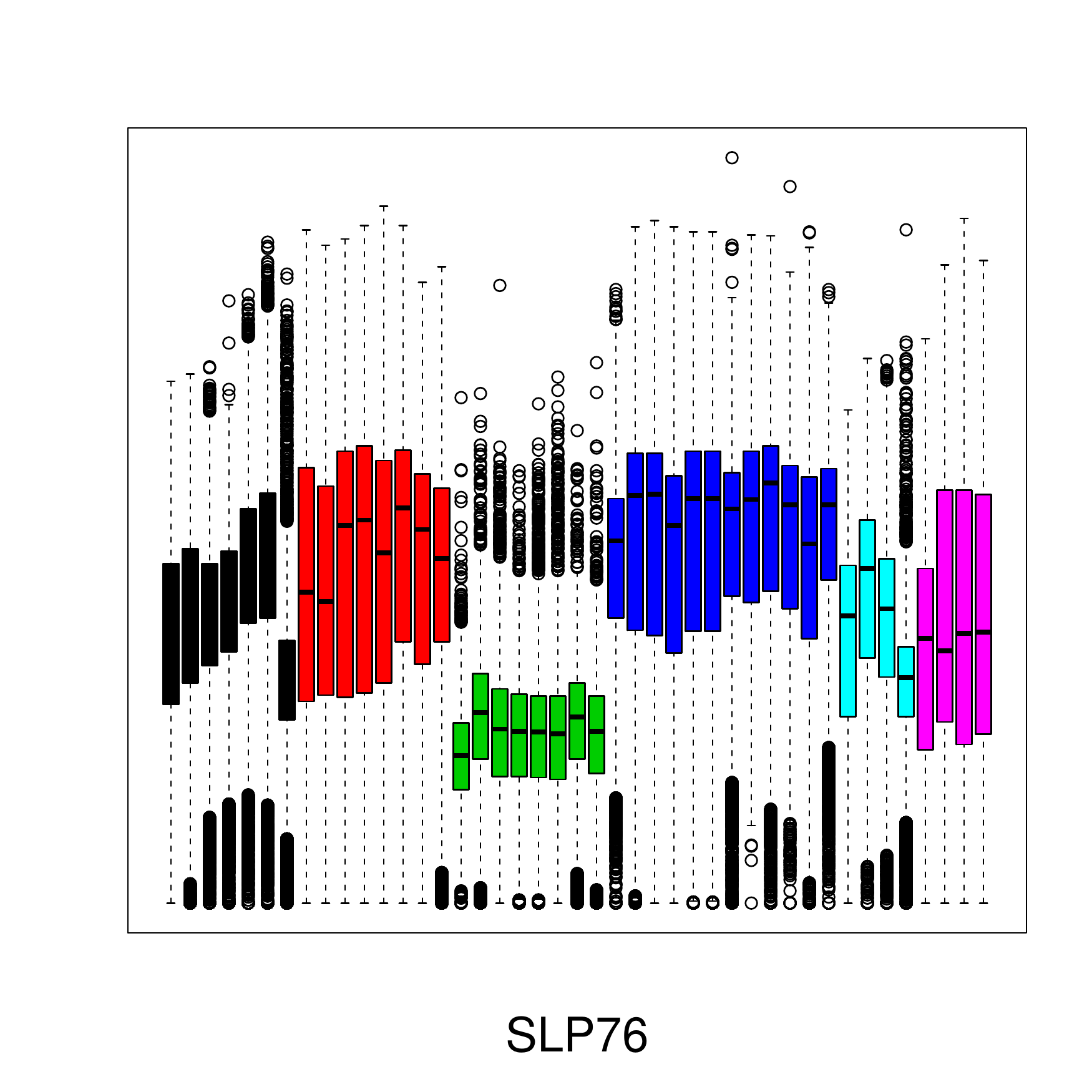}
\end{center}
\vspace{-.2in}
\caption{\em Cytometric Data. Boxplots of SLP76.}
\label{fig::CytoBoxplot1}
\end{figure}

\begin{figure}
\begin{center}
\includegraphics[width=6in,height=2.4in]{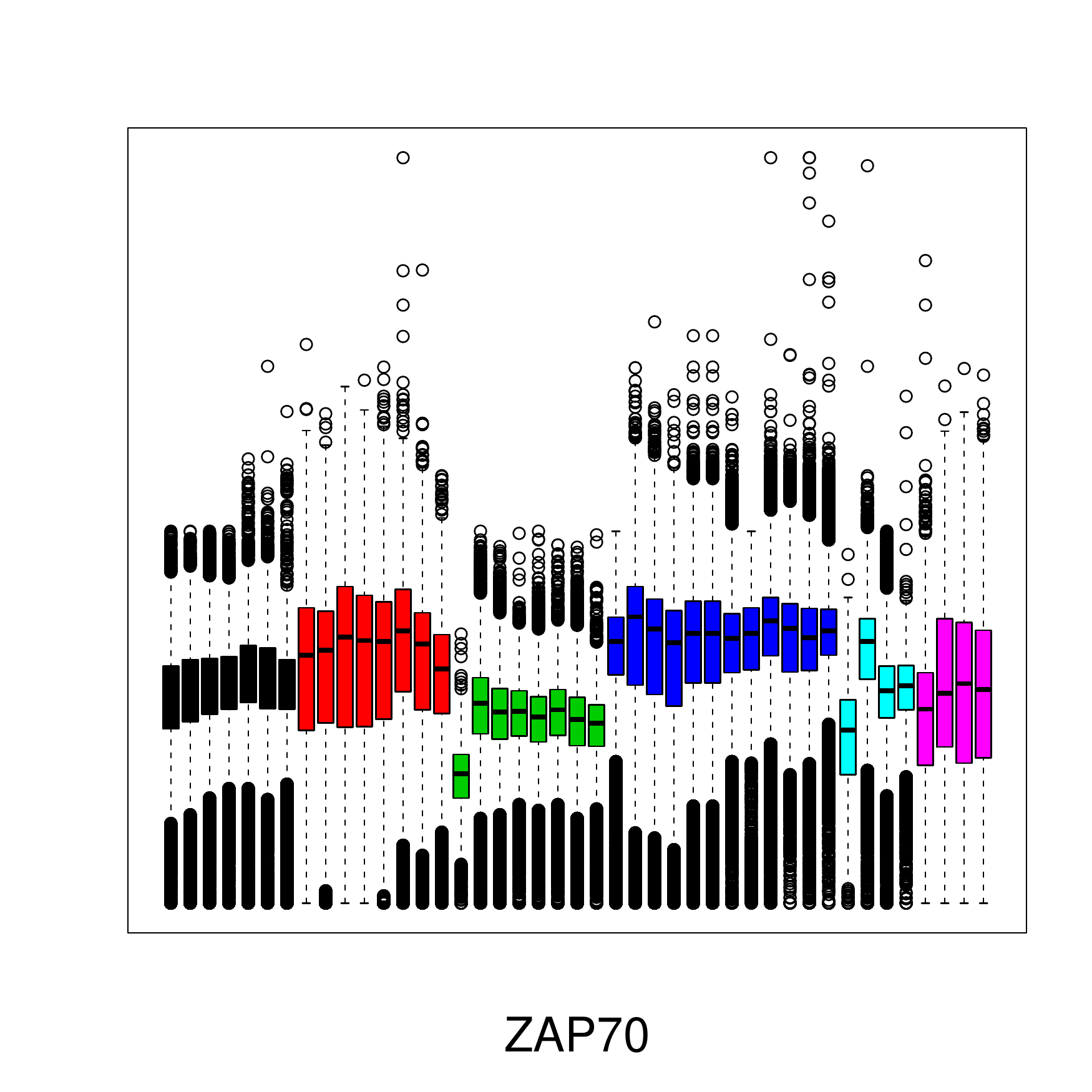}
\end{center}
\vspace{-.1in}
\caption{\em Cytometric Data. Boxplots of ZAP70.}
\label{fig::CytoBoxplot2}
\end{figure}

The boxplots are color coded by our clustering. Note that
the boxplots of the SLP6 and ZAP70 proteins correspond nicely 
to the six clusters. 
The third and fourth proteins , CD4, and CD45RA in Figure 
\ref{fig::CytoBoxplot2}, instead do not seem to be well clustered. 
We conclude that the clustering information is deriving mainly 
from  SLP6 and ZAP70.

\begin{figure}
\vspace{-.2in}
\begin{center}
\includegraphics[width=6in,height=2.4in]{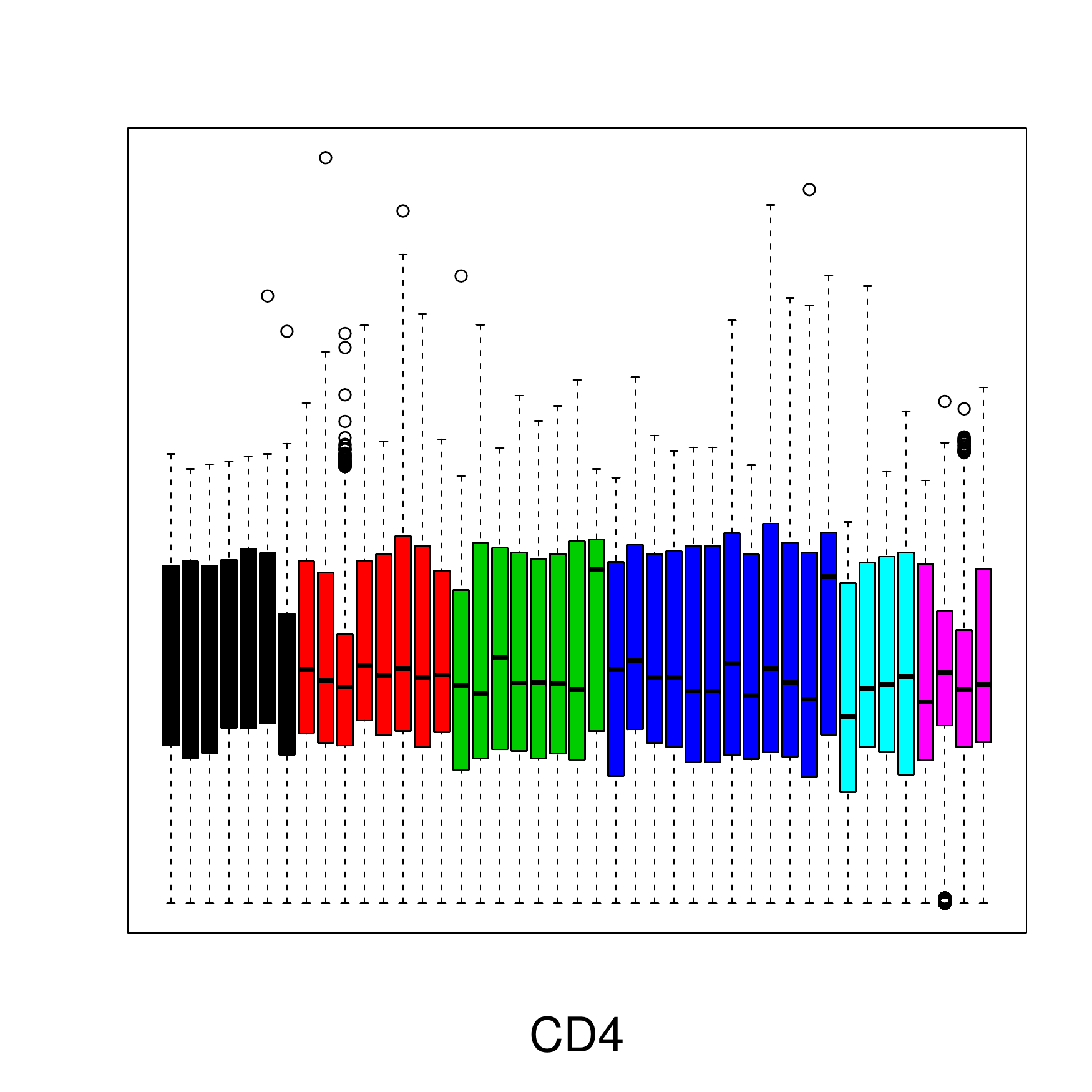}
\includegraphics[width=6in,height=2.4in]{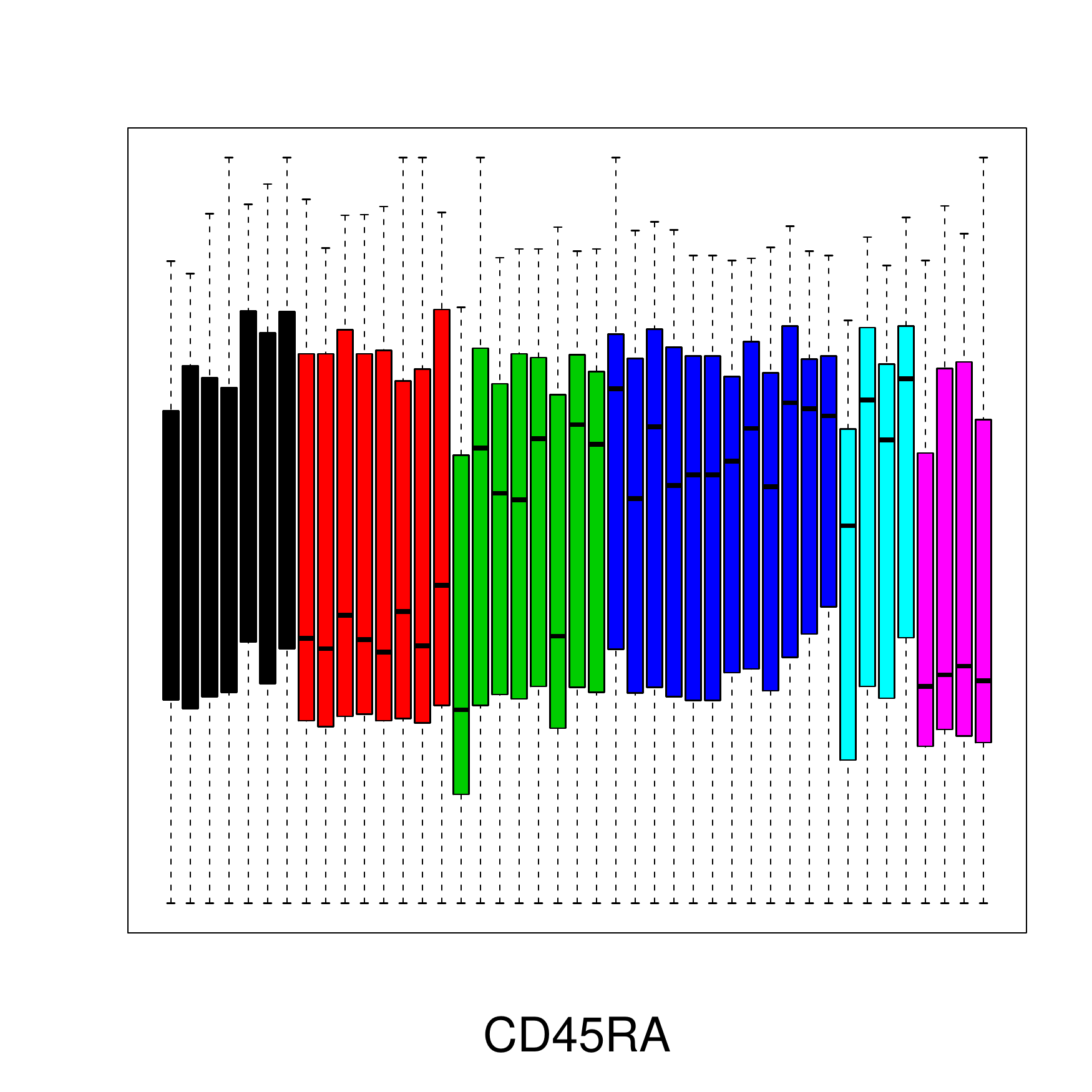}
\vspace{-.1in}
\caption{\em  Cytometric Data. Boxplots of CD4 and CD45A.}
\label{fig::CytoBoxplot3}
\end{center}
\end{figure}

The following Figure \ref{fig::CytoCluster} shows the six clusters
obtained with the Gaussian and the Hybrid procedures. Although the clusters 
appear to be not so well separated, we recall that multidimensional 
scaling projections from several dimensions to two might not be very
accurate. Our choice, in searching for six clusters, has been suggested 
by the plots in Figure \ref{fig::cytometric-elbow} of Section 
\ref{section::elbows}, where the issue of choosing the number of 
clusters is examined, and reinforced by the boxplots in Figure
\ref{fig::CytoBoxplot1} and Figure \ref{fig::CytoBoxplot2}

\begin{figure}
\begin{center}
\begin{tabular}{cc}
\includegraphics[scale=.3]{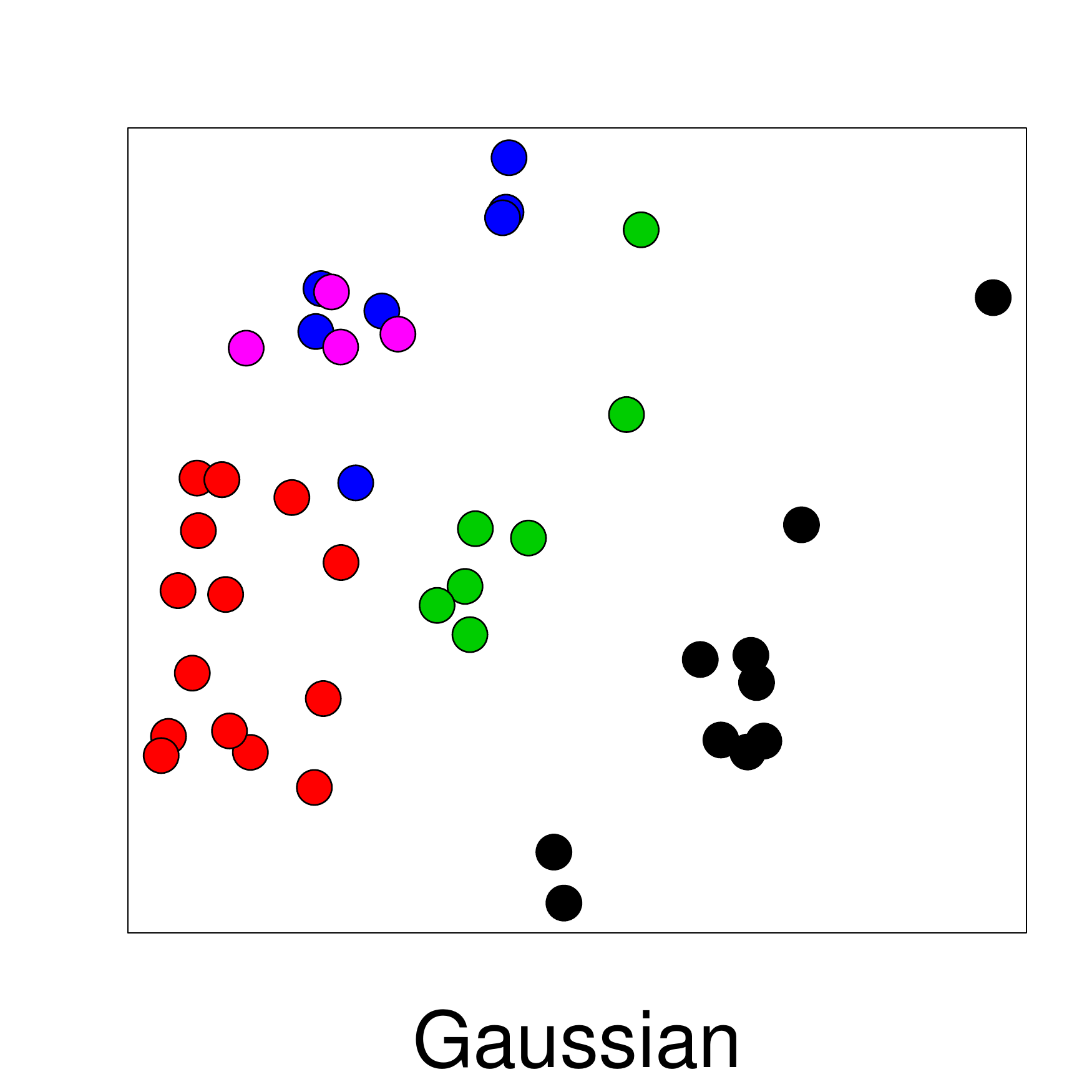}
\includegraphics[scale=.3]{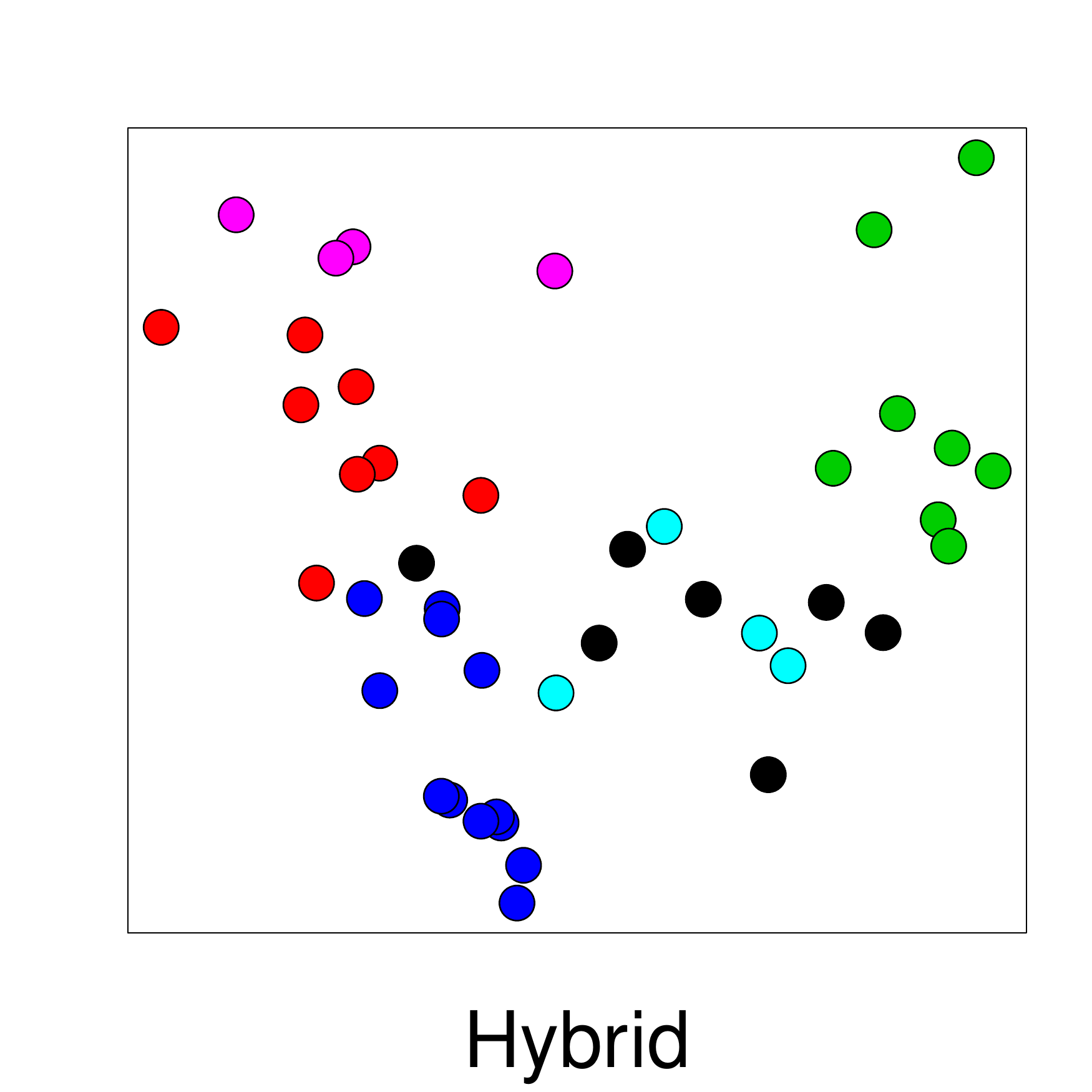}
\end{tabular}
\end{center}
\caption{\em Cytometric Data. Clusters
Clusters obtained using the  procedures Gaussian and  Hybrid.}
\label{fig::CytoCluster}
\end{figure}

\subsection{Choosing $k$: Elbows}
\label{section::elbows}

The issue of choosing $k$ in $k$-means clustering
is a well-studied problem. Perhaps the oldest and simplest method 
is to plot the sums of squares versus $k$ and look for an elbow.
For distribution clustering, we plot
$S_k=\sum_j \sum_{s\in {\cal C}_j}W^2(P_s,c_j)$ versus $k$.
In some cases, the elbow is clearer if we plot
$1/S_k$ versus $k$. Figure \ref{fig::elbow} shows the
plot of $S_k$ versus $k$ for the second example in Section
\ref{section::MvarExamples}  where data were generated as four 
groups of circles. We see a clear elbow at $k=4$.
Figure \ref{fig::cytometric-elbow} shows another plot of cluster quality 
versus $k$. In this case, for the cytometric data sets, we plot 
$1/\sum_j \sum_{s\in {\cal C}_j} W^2(P_s,c_j)$ as the inverse plot shows 
a clearer signal. The elbow is at $k=6$.

\begin{figure}
\begin{center}
\includegraphics[scale=.3]{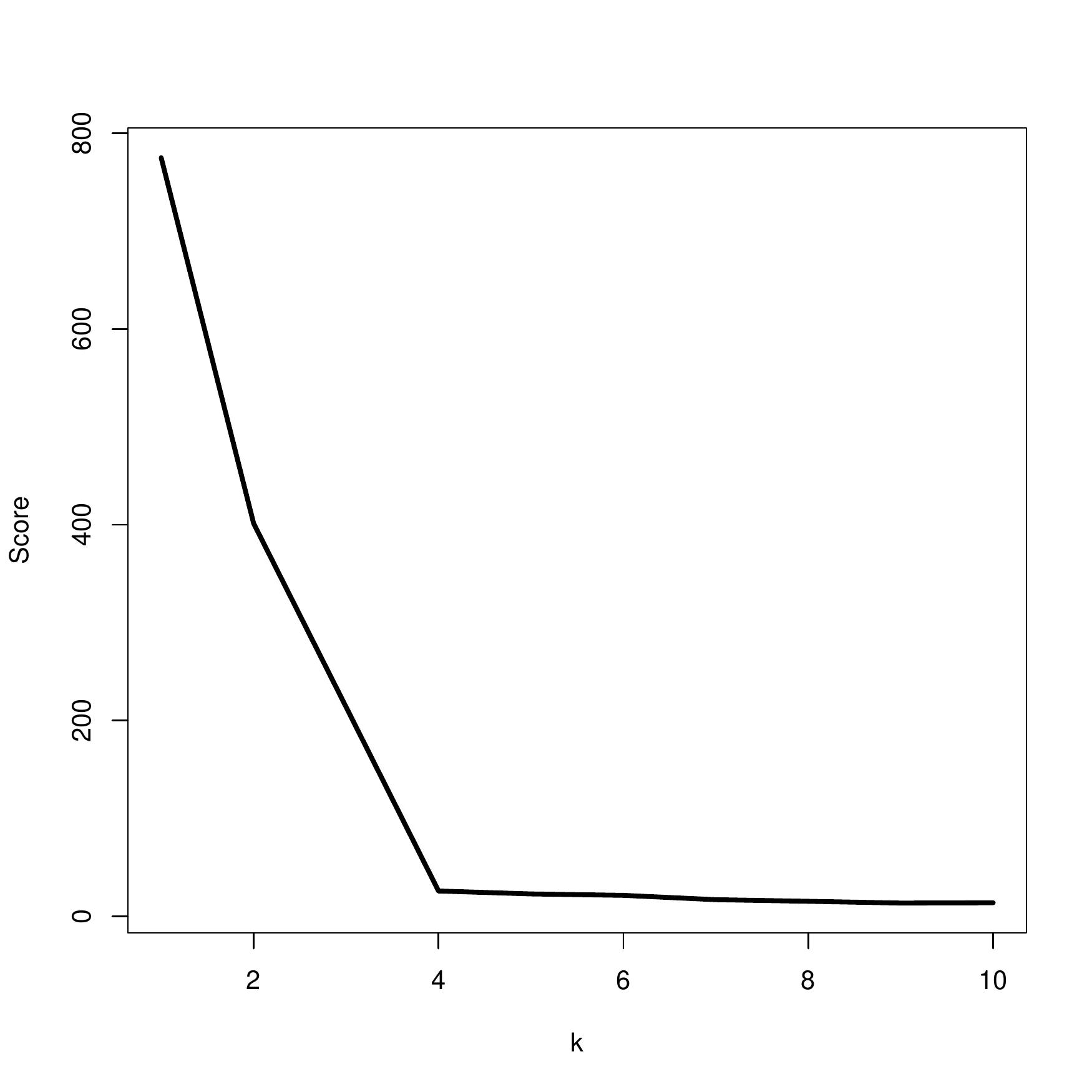}
\end{center}
\vspace{-.1in}
\caption{\em Plot of $\sum_{j=1}^k \sum_{s\in {\cal C}_j}W^2(P_s,c_j)$
for $k=1,\ldots, 10$ for the example with four groups of circles.
Note the pronounced elbow at $k=4$.}
\label{fig::elbow}
\end{figure}

\begin{figure}
\begin{center}
\includegraphics[scale=.4]{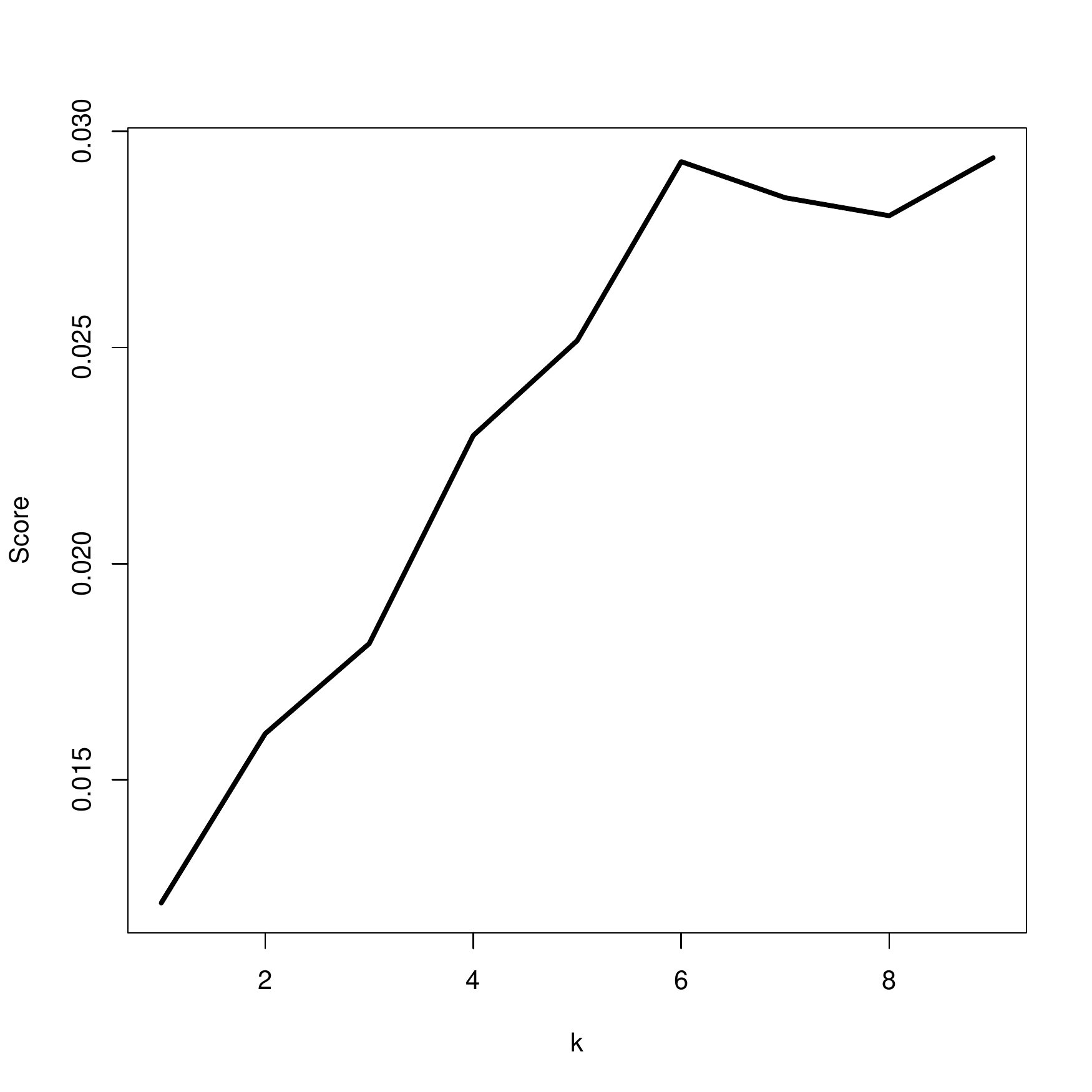}
\end{center}
\vspace{-1cm}
\caption{\em Plot of $1/\sum_{j=1}^k \sum_{s\in {\cal C}_j}W^2(P_s,c_j)$
for $k=1,\ldots, 9$ for the cytometric data.
Note the elbow at $k=6$.}
\label{fig::cytometric-elbow}
\end{figure}

\subsection{The Pre-testing Speedup}
\label{section::speedup}

We generated 100 observations from each of 30 multivariate Normal 
distributions
that lie in three distinct groups.
After computing the rescaled data
$\tilde{X}$, we apply the cross-match test from
\cite{rosenbaum2005exact}
to see if the rescaled data differ significantly from the sample
from the reference distribution.
We used level $\alpha=0.10$.
In the majority of cases, the test does not reject
and we set $\psi_j=0$ which saves considerable computing time.
Typically, the clustering is perfect
(adjusted Rand index of 1).
But when the same procedure is applied to
the Normal and circles data,
the null is almost always rejected (as expected)
and we need to do the full computations.

\section{Other Clustering Methods}
\label{section::other}

There are, of course, many other clustering methods besides
$k$-means clustering. In this section we give a sense of
how the hybrid distance can be used for two other clustering methods.

{\bf Hierarchical Clustering.}
The hybrid distance can easily be used for hierarchical clustering.
These are the steps for single linkage hierarchical clustering:

\begin{enum}
\item Compute the transform
$(\mu_j,\Sigma_j,\psi_j)=\phi(P_j)$ for each distribution.
\item Compute the pairwise distances
$H_{jk}=H(P_j,P_k)$.
\item Find the pair $(j,k)$ that minimizes $H_{j,k}$.
Merge $P_j$ and $P_k$ into a single distribution
$P_{jk} = \pi P_j + (1-\pi)P_k$
where $\pi = n_j/(n_j+n_k)$
and $n_j$ and $n_k$ are the sample sizes of the two datasets.
\item Compute the representation
$(\mu_{jk},\Sigma_{jk},\psi_{jk})$ of $P_{jk}$
by computing the barycenter
of $P_j$ and $P_k$ with weights
$\pi_{jk}$ and $1-\pi_{jk}$.
\item Repeat steps 2-4 until all distributions are merged.
\end{enum}

By working with the hybrid representations
we do not have to keep recomputing Wasserstein distances.

\bigskip

{\bf Mean-Shift and Medioid-Shift Clustering.}
Next we discuss mean shift clustering 
and medioid shift clustering \citep{cheng1995mean,chacon2018multivariate,
jiang2018quickshift++,chacon2015population}.
First we review these methods when applied to a single dataset.
Give a sample $Y_1,\ldots, Y_n$ from a distribution $P$ with density $p$,
mean shift clustering works by
first computing an estimate $\hat p$ of $p$.
Let $m_1,\ldots, m_k$ be the modes of $\hat p$.
Given any point $y$, if we follow the gradient of $\hat p$
starting at $y$ we will end up at one of the modes.
In this way the modes define a partition of the sample space.
This partition defined the mode-based clustering.
There is a simple iterative algorithm called the 
mean-shift algorithm to implement this idea.
Recently, 
\cite{duong2016nearest}
showed that if we use $r$-nearest neighnor density estimation,
then the iterative algorithm takes the following form.
Pick any starting value $x$.
Define $y^{(0)}, y^{(1)},\ldots$ by
$y^{(0)}=y$ and
\medskip
$$
y^{(s)} = \frac{1}{r}\sum_{N_r(y^{(s-1)})}Y_i
$$
\medskip
where
$N_r(y)$ denotes the $r$-nearest neighbors of the point $y$.
This iteration leads to a mode of the estimated density.
This assigns any point $y$ to a mode.
The iteration can be applied to any set of starting points
although these starting points are usually taken to be the data points.

\medskip
Returning to distribution clustering,
we have a set of distributions $P_1,\ldots, P_N$
which we now regard as a sample from a measure $\Pi$ on the space 
of distributions.
Unfortunately, $\Pi$ does not have a density in any meaningful sense.
Nonetheless, we can formally apply the mean shift clustering algorithm.
Given any $P$, let
$N_r(P)$ denote the $r$ closest distributions in
$\{P_1,\ldots, P_N\}$ to $P$
under the hybrid distance.
Let
${\rm Bary}(N_r(P))$ denote the hybrid barycenter
of the distributions in
${\rm Bary}(N_r(P))$.
We then define
$P^{(0)}=P$ and
\medskip
$$
P^{(s)} = {\rm Bary}(N_r(P^{(s-1)})).
$$
\medskip
A faster approach, which avoids computing the barycenter,
is mediod-shift clustering.
In this case,
we begin by estimating 
$\rho(P_j) = 1/d(P,P_r(P))$ where
$P_r(P)$ is the $r^{\rm th}$ nearest neighbor.
This can be thought of as a pseudo-density.
We move $P$ to the distribution
in $N_r(P)$ with highest pseudo-density.
This is repeated until there is no change.
This can be regarded as an approximation to mean-shift clustering.

As an example,
we consider
a collection of 80, two-dimensional datasets.
Each dataset has $n=100$ observations from
a bivariate Normal.
We construct the data so that there are four, well-defined clusters.
Figure
\ref{fig::mediod}
shows mediod-shift clustering using the hybrid distance.
The plots correspond to $r=2, 10, 15$ and 20
nearest neighbors.
The plots show the MDS of the data sets and the paths of the data
as the clustering proceeds.
The red dots show the final destinations, that is, the pseudo-modes.
When $r=2$ we get extra clusters.
When $r$ reaches 20 we start to oversmooth and end up with 2 clusters.
There is a large range of values of $r$ that lead to the correct 
answer of 4 clusters.

Currently, we do not have a theoretical basis for applying
the mean shift algorithm to
distributions.
We believe that it may be possible to define
a pseudo-density on the space of distributions simialr to the approach used by
In the context of clustering functional data,
\cite{ferraty2006nonparametric}
justify density clustering by defining a pseudo-density.
We conjecture that density clustering in Wasserstein space can be similarly justified
by regarding the density estimator as estimating a pseudo-density.
For example,
if $P$ is a random distribution,
one can define the pseudo-density at $P_0$ by
$\mathbb{P}(P\in B(P_0,h))$
where
$B(P_0,h) = \{ P:\ H(P_0,P)\leq h\}$.
We leave the details for future work.

\begin{figure}
\begin{center}
\begin{tabular}{cc}
\includegraphics[scale=.4]{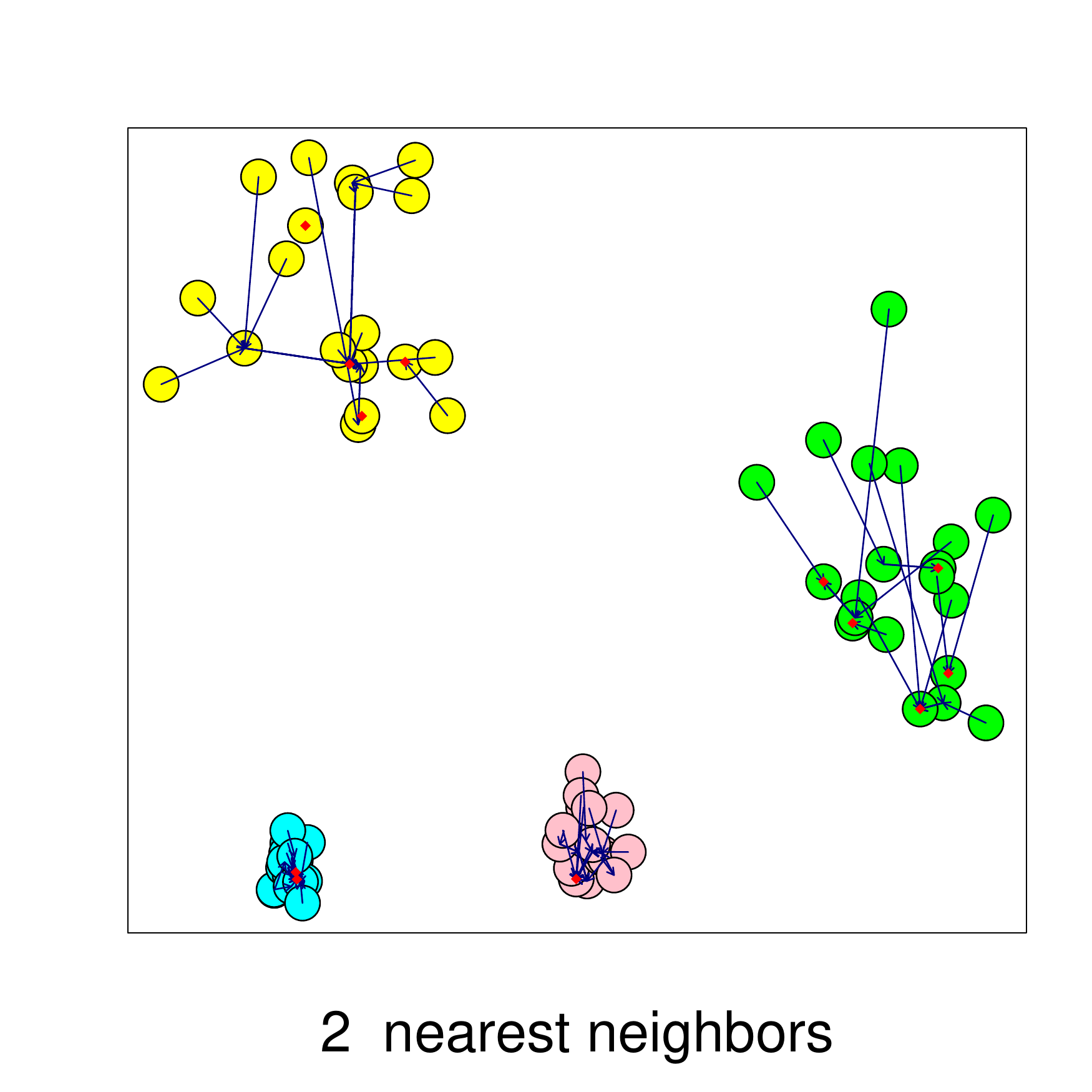}
\includegraphics[scale=.4]{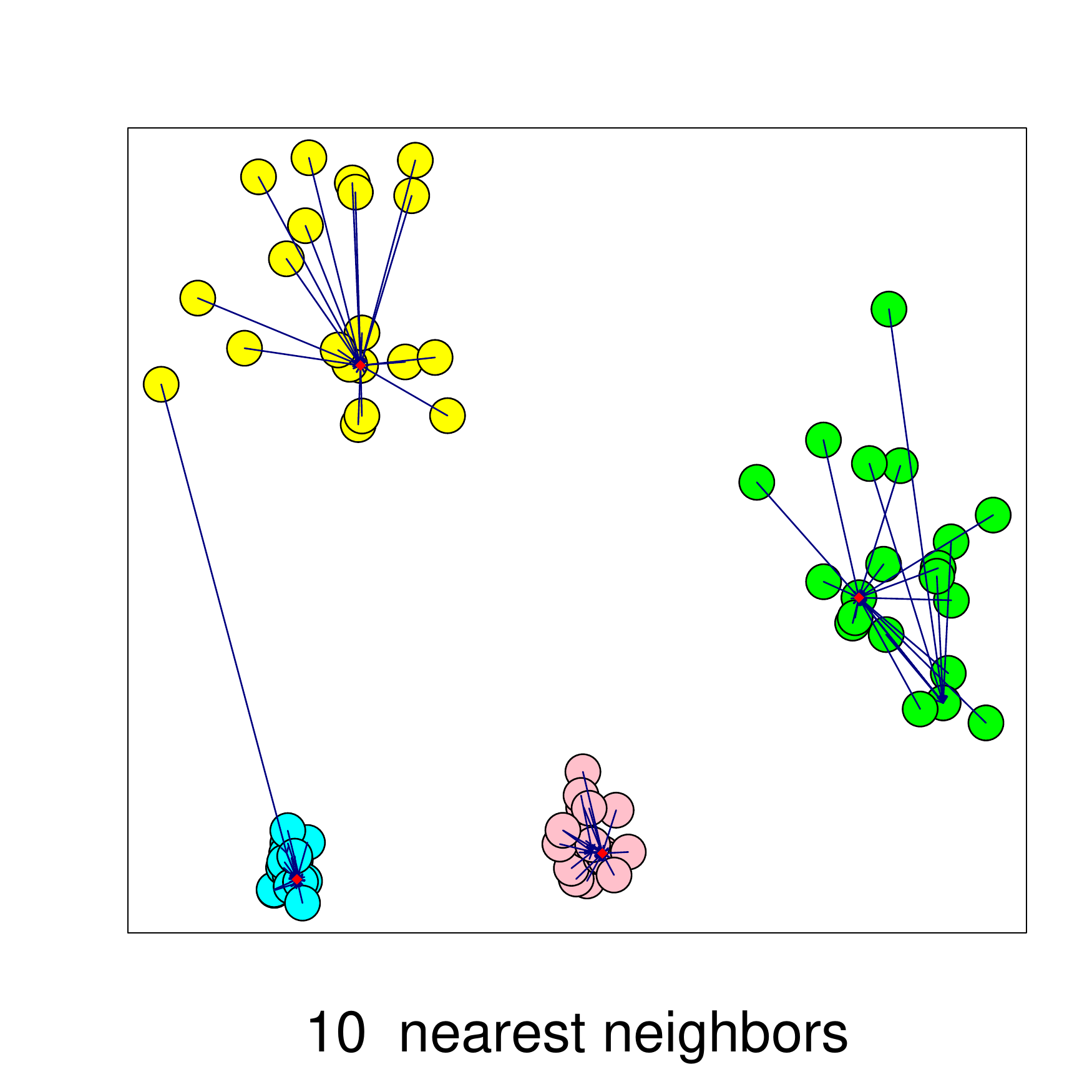}
\end{tabular}

\begin{tabular}{cc}
\includegraphics[scale=.4]{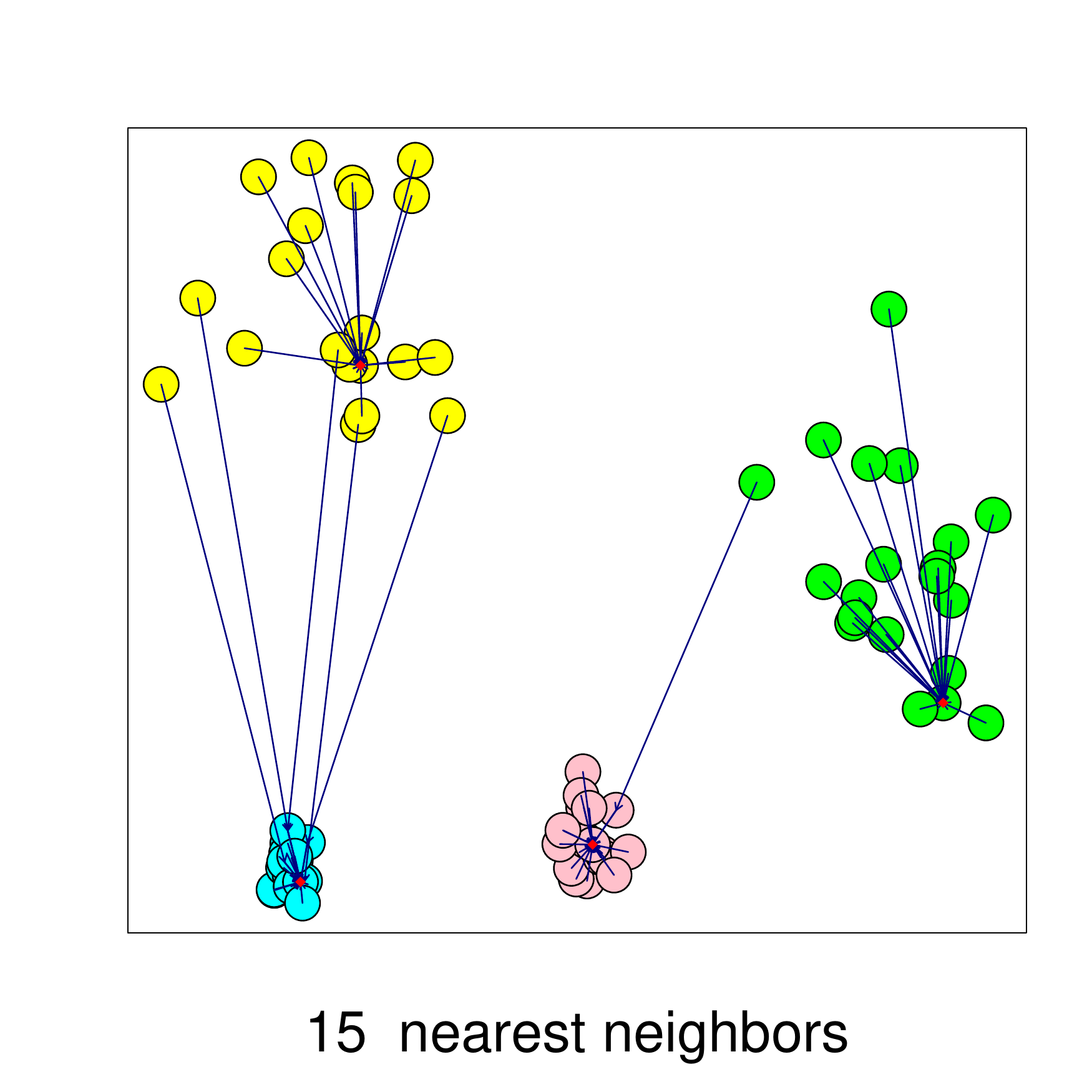}
\includegraphics[scale=.4]{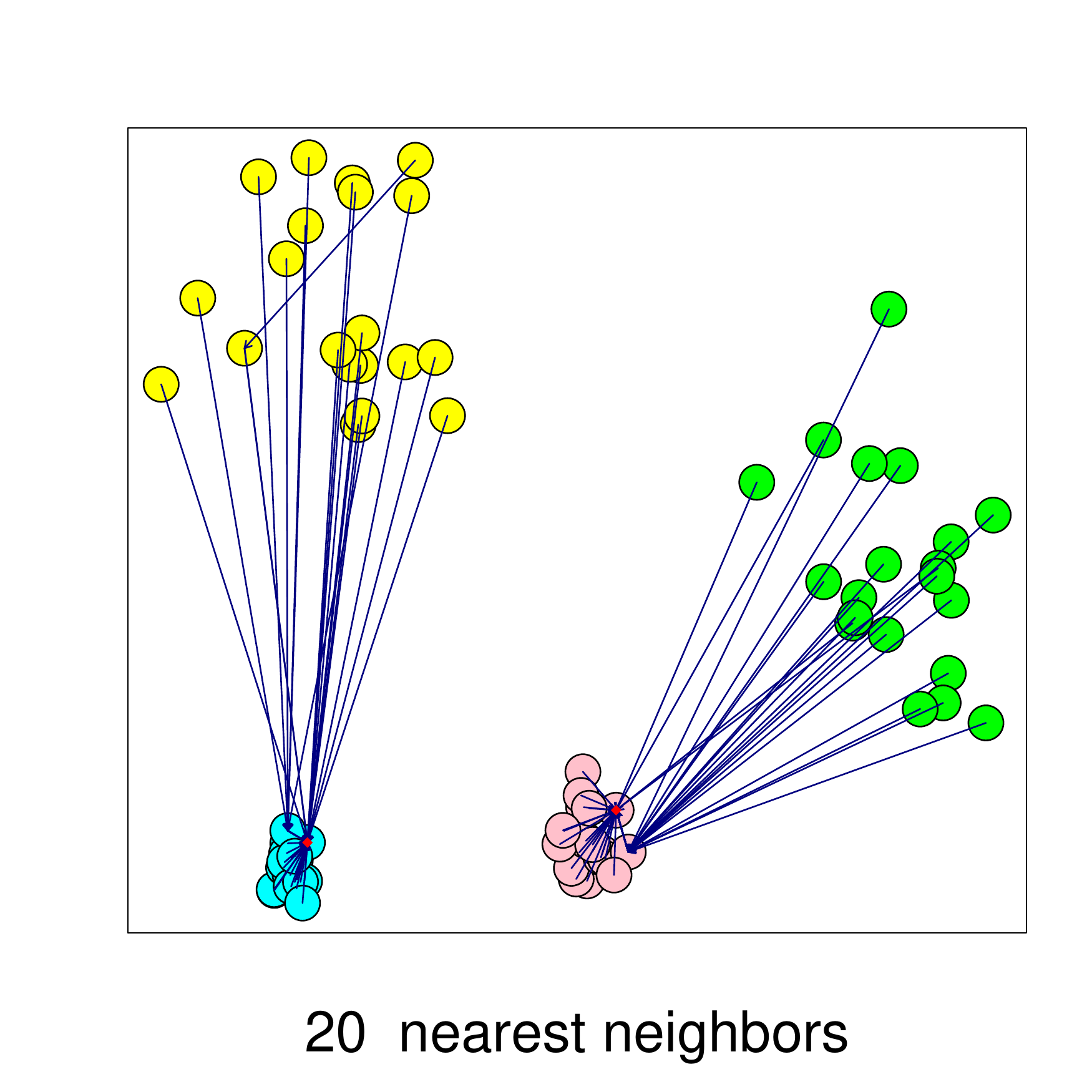}
\end{tabular}
\end{center}
\caption{\em Mediod-Shift Clustering using the hybrid distance.
The plots show the MDS of the data sets and the paths of the data
as the clustering proceeds
for$r=2, 10, 15$ and 20
nearest neighbors.}
\label{fig::mediod}
\end{figure}

\section{Other Hybridizations}
\label{section::otherH}

Now we consider other hybrid distances of the form,
$$
H^2(X,Y) = W^2(Z_X,Z_Y) + W_\dagger^2 (\tilde X, \tilde Y)
$$
where
we explore other choices besides the tangent distance for $W_\dagger$.
In each case,
we want an approximate distance that
can be computed quickly.

\bigskip
{\bf Marginal Distance.}
Here we take
$$
W_\dagger^2(X,Y) = \sum_{j=1}^d W^2(X(j),Y(j))
$$
where $X(j)$ and $Y(j)$ are the
$j^{\rm th}$ components of $X$ and $Y$.
It then follows that
$$
W_\dagger^2(X,Y) = \sum_{j=1}^d \int_0^1 |F^{-1}_j(u)-G^{-1}_j(u)|^2 du
$$
where
$F_j(t) = P(X(j) \leq t)$ and
$G_j(t) = Q(Y(j) \leq t)$.
Hence,
$$
H^2(X,Y) =
||\mu_X-\mu_Y||^2 + {\rm B}^2(\Sigma_X,\Sigma_Y) + 
\sum_{j=1}^d \int_0^1 |F^{-1}_j(u)-G^{-1}_j(u)|^2 du.
$$
The estimate is
$$
\hat H^2(X,Y) =
||\hat \mu_X-\hat \mu_Y||^2 + {\rm B}^2(\hat\Sigma_X,\hat\Sigma_Y) + 
 \sum_{j=1}^d \int_0^1 |\hat F^{-1}_j(u)-\hat G^{-1}_j(u)|^2 du.
$$
Assuming we have samples $m$ observations from each dataset,
we have the further simplification that
$$
\int_0^1 |\hat F^{-1}_j(u)-\hat G^{-1}_j(u)|^2 du=
\frac{1}{m}\sum_i (X_{(i)}(j)-Y_{(i)}(j))^2
$$
where
$X_{(1)}(j),\ldots, X_{(n)}(j)$ and
$Y_{(1)}(j),\ldots, Y_{(n)}(j)$ are
the order statistics for the $j^{\rm th}$ coordinate of $X$ and $Y$,
respectively.

Next we consider an example.
We generate 100 datasets. The first 50 are standard bivariate Normal.
The second 50 are uniform on a circle, scaled to have the same mean and 
covariance as the Normal.The left plot of Figure \ref{fig::marginal}
shows the Gaussian-Wasserstein distances using multidimensional scaling.
The colors indicate Normal (black) and circular (red).
The Gaussian-Wasserstein distance cannot distinguish the two types of datasets.
The right plot shows the marginal hybrid distance.
Here, the two types of datasets are clearly distinguished.

\begin{figure}
\begin{center}
\begin{tabular}{cc}
\includegraphics[scale=.4]{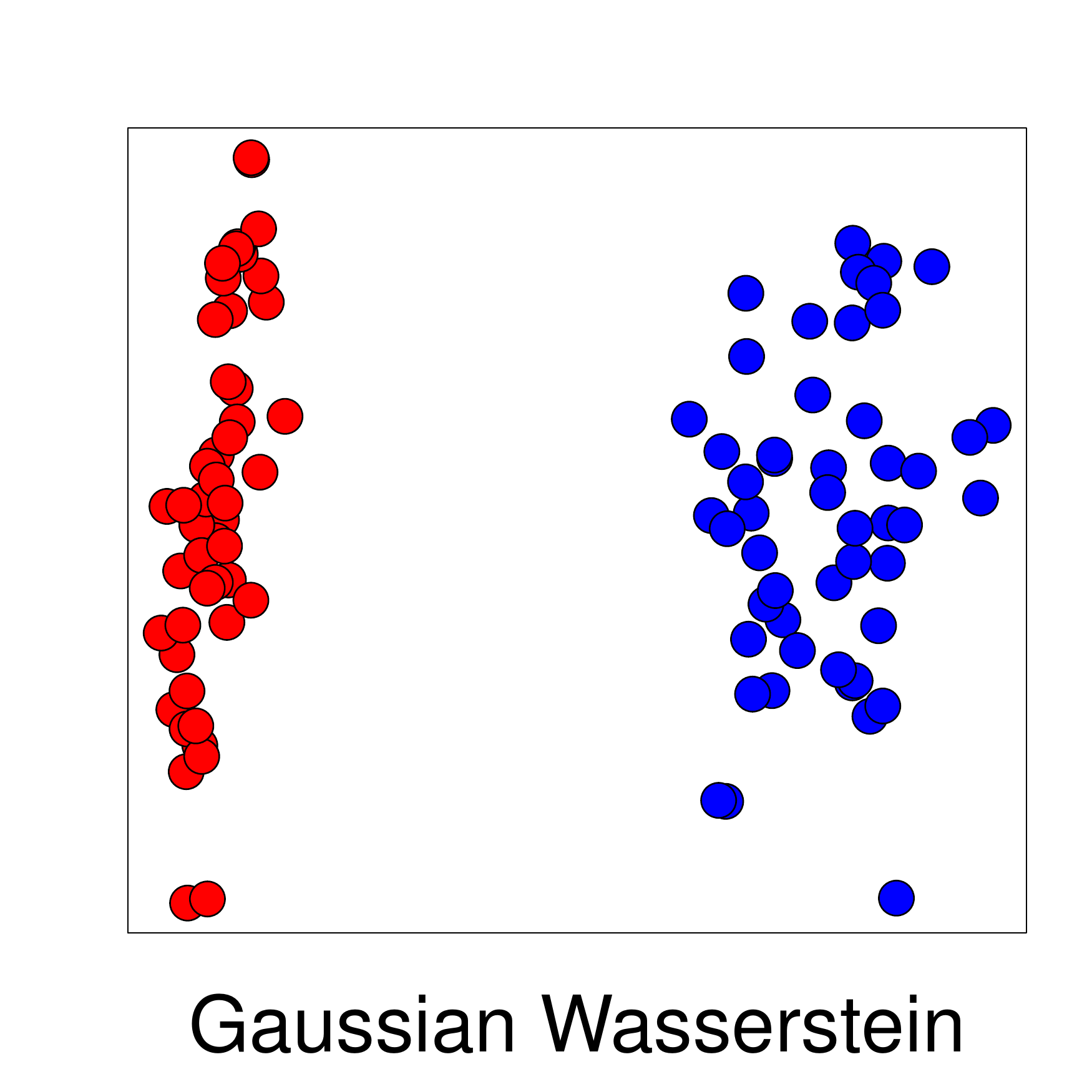}
\includegraphics[scale=.4]{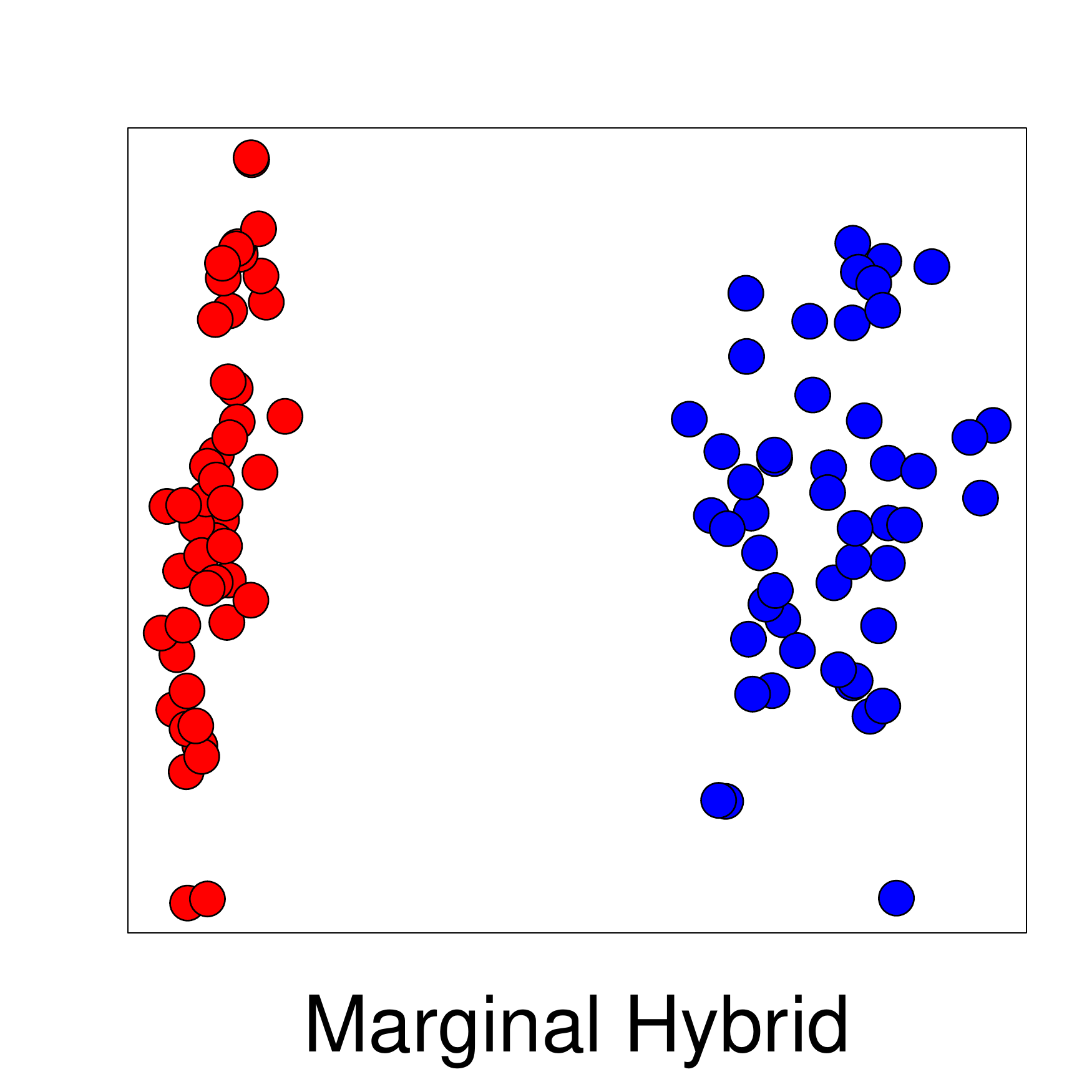}
\end{tabular}
\end{center}
\caption{\em There are 100 distributions. The first 50 
are standard bivariate Normal. The second 50 are uniform on a circle, 
scaled to have the same mean and covariance as the Normal.
The left plot shows the Gaussian-Wasserstein distance.
The colors indicate the Normal data (red) and the circular (blue).
As expected, the Gaussian-Wasserstein distance cannot distinguish 
the two types of datasets. The right plot shows the marginal 
hybrid distance. Here, the two types of datasets are clearly 
distinguished.}
\label{fig::marginal}
\end{figure}

{\bf Transformed Gaussian Approximation.}
Recall that the Gaussian-Wasserstein distance is
$$
G^2(X,Y) = ||\mu_X-\mu_Y||^2 + B^2(\Sigma_X,\Sigma_Y).
$$
A hybrid distance that leverages the simplicity of $G$
can be obtained by repeatedly applying a nonlinear transformation to
the variables, using the $G$, standardizing the variables and repeating.
Let $L$
be the map that takes $X$ to
$\tilde{X} = \Sigma_X^{-1/2}(X-\mu_X)$ and let
$\Phi$ be some fixed nonlinear map.
Define
$H^2(X,Y) = G^2(X,Y) + W_\dagger(X,Y)$ where
$$
W_\dagger(X,Y)=\sum_{j=1}^k G^2(X_j,Y_j)
$$
with
$$
X_j = \underbrace{[(\Phi\circ L)\circ \cdots (\Phi\circ L)]}_{k\ {\rm times}} X,\ \ \ 
Y_j = \underbrace{[(\Phi\circ L)\circ \cdots (\Phi\circ L)]}_{k\ {\rm times}} Y.
$$
In the case $k=1$ this simplifies to
\begin{align*}
H^2(X,Y) &= G^2(X,Y) + G^2(\Phi(\tilde X),\Phi(\tilde Y))\\
&=
||\mu_X-\mu_Y||^2 + B^2(\Sigma_X,\Sigma_Y) +
||\mu_{\Phi(\tilde X)}-\mu_{\Phi(\tilde Y)}||^2 + 
B^2(\Sigma_{\Phi(\tilde X)},\Sigma_{\Phi(\tilde Y)})
\end{align*}

There are many possible choices for $\Phi$. The only requirement 
is that $\Phi$ be non-linear otherwise the transformation adds no 
information beyond that already captured by $G$. A convenient 
choice is the polynomial transform
$$
\Phi(x) = ([x]_2,\ldots, [x]_k)
$$
where
$[x]_2 = (x_{i_1}x_{i_2}:\ 1\leq i_1 \leq i_2 \leq d)$,
$[x]_3 = (x_{i_1}x_{i_2}x_{i_3}:\ 1\leq i_1 \leq i_2 \leq i_3 \leq d)$, etc.

As an example, 
consider a one-dimensional case.
We suppose that
$P_1,\ldots,P_{50}=N(0,1)$
and
$P_j=(1/2)\delta_{-1}+(1/2)\delta_1$
for
$j=51,\ldots, 100$.
We take $\Phi(X) = (X^2,X^3, X^4)$.
Since the first two moments 
of all the distributions match,
the Gaussian distance is unable to distinguish the 100 distributions.
Figure \ref{fig::transformation}
shows the pairwise distances using
multidimensional scaling (MDS).
The top left plot is based on parwise Wasserstein distances.
The two groups are clearly visible.
The top rightplot used Gaussian Wasserstein distance.
As expected, the distributions are mixed up.
The bottom left uses hybrid distance with polynomial transformation.
Here we see that again, the distributions are clearly separated.

\begin{figure}
\begin{center}
\begin{tabular}{ccc}
\includegraphics[scale=.3]{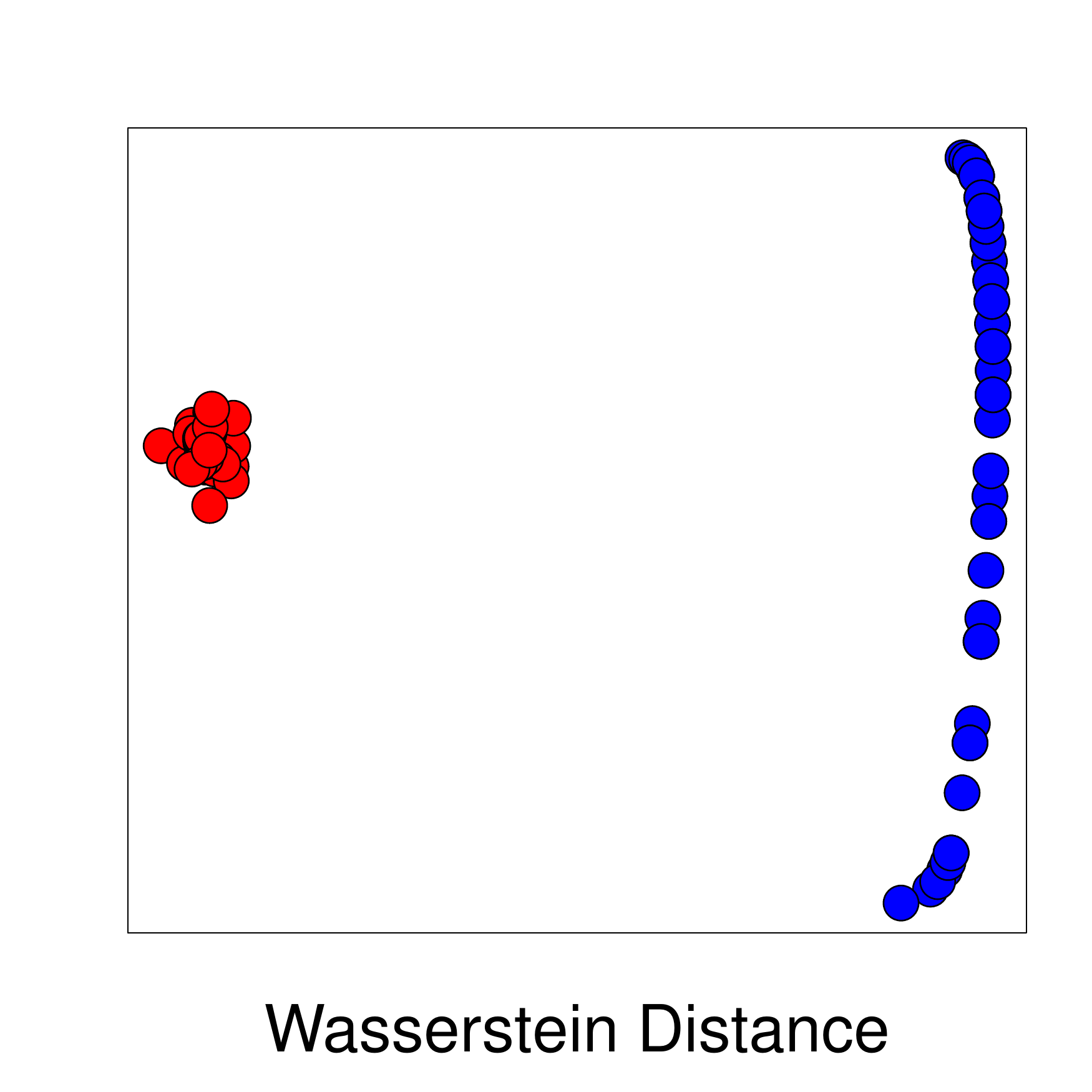}
\includegraphics[scale=.3]{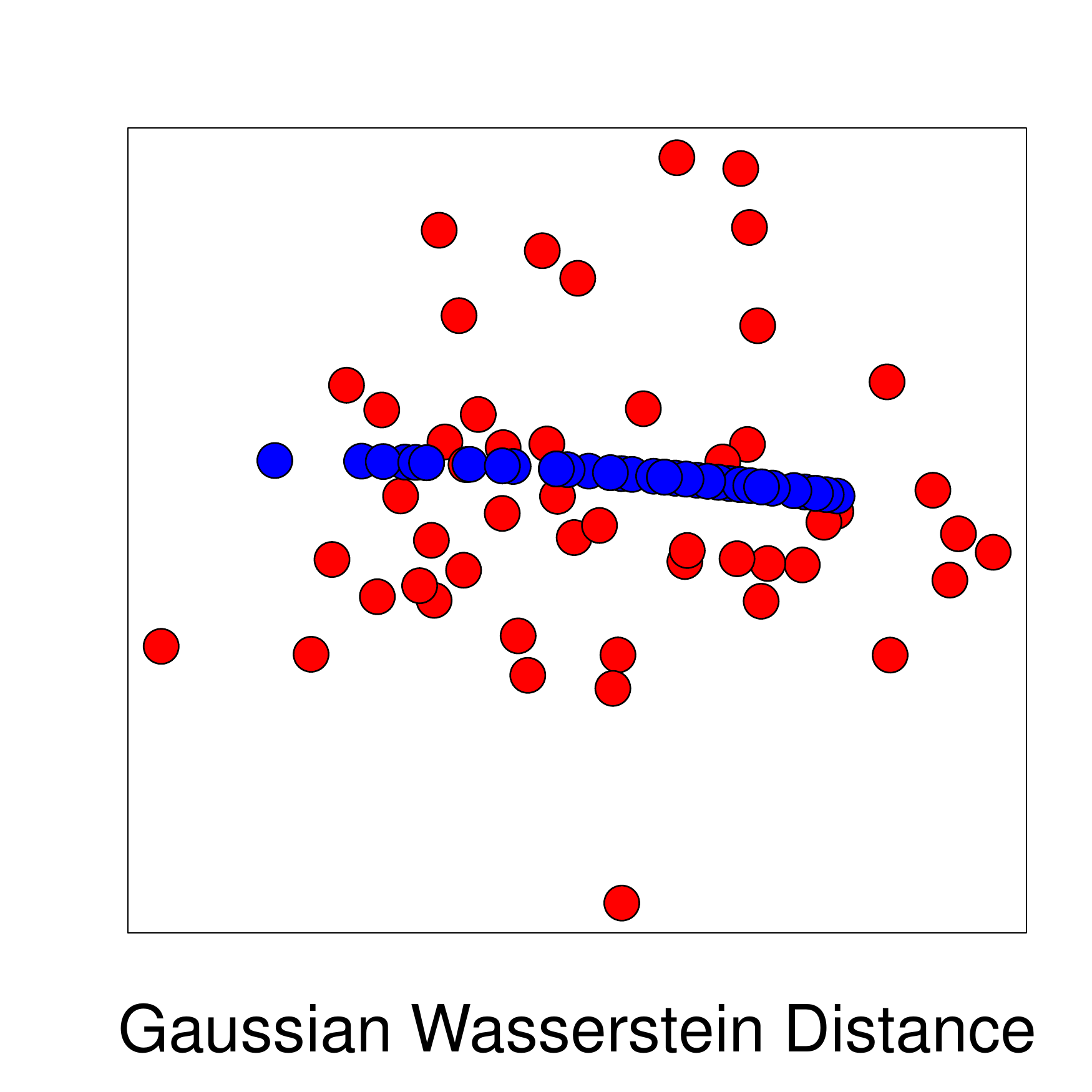}
\includegraphics[scale=.3]{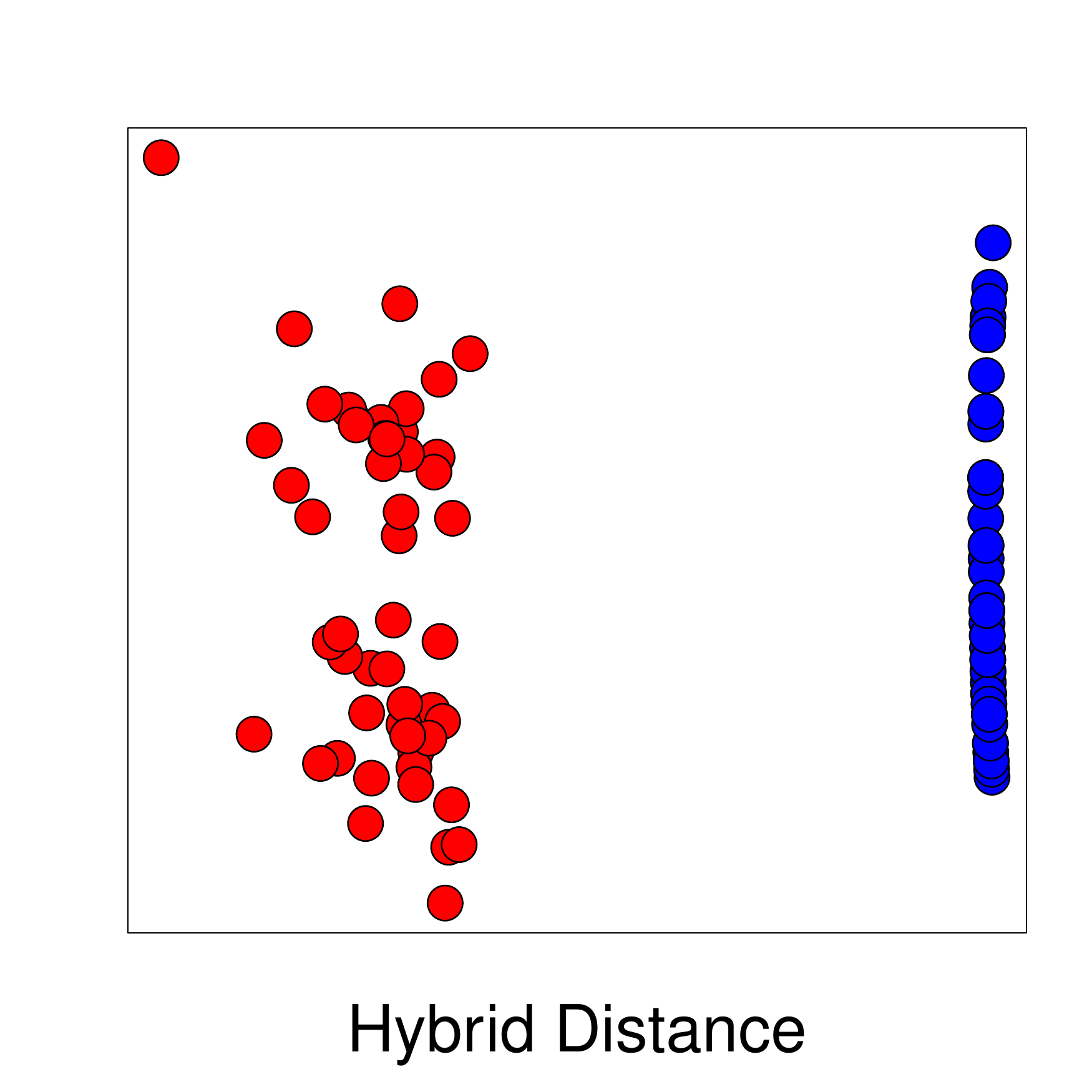}
\end{tabular}
\end{center}
\caption{\em 50 distributions are N(0,1) (black)
and 50 are mixtures of the form
$\delta_{-1}+(1/2)\delta_1$ (red).
The sample size for each distribution is 1,000.
Left: MDS based on parwise Wasserstein distance.
The two groups are clearly visible.
Middle: MDS based on the Gaussian Wasserstein distance.
As expected, the mixtures distributions are mixed in with the Normals.
Right: Hybrid distance with polynomial transformation.
Here we see that again, the distributions are clearly separated.}
\label{fig::transformation}
\end{figure}

\section{Discussion}
\label{section::discussion}

We have proposed a modified version
of the Wasserstein distance called the hybrid distance
that can be used for clustering sets of distributions.
The distances and the barycenter can be computed quickly.
The slowest part of the computation is finding the optimal
matching permutation which takes $O(m^3)$ operations.
There is large literature on approximate matching
in the computer science literature.
It would be interesting to incorporate some of those methods
to further speed up the computations.

The hybrid distance
can be used for other tasks as well
such as shrinkage estimation,
modeling random effects, domain adaptation,
and image processing.
We will report on these applications in future work.
We conclude by discussing a few issues.

\subsection{Energy Distance}
\label{section::energy}

The energy distance
(\cite{szekely2013energy})
is another metric that takes the underlying geometry
of the sample space into account.
The metric defined by
\begin{equation}\label{eq::energy}
{\cal E}(X,Y) \equiv {\cal E}(P,Q) =
2\mathbb{E}||X-Y|| - \mathbb{E}||X-X'|| - \mathbb{E}||Y-Y'||
\end{equation}
where
$X,X'\sim P$ and
$Y,Y'\sim Q$.
The energy distance has many of the desirable properties that the Wasserstein distance has
including the ability to compare discrete and continuous distributions.

A sample estimator of the distance based on
$X_1,\ldots, X_n \sim P$ and
$Y_1,\ldots, Y_m \sim Q$
is
\begin{equation}\label{eq::energy2}
\hat{\cal E}(X,Y) =
\frac{2}{nm}\sum_{i,j}||X_i - Y_j|| -
\binom{n}{2}\sum_{i\neq j} ||X_i - X_j|| - 
\binom{m}{2}\sum_{i\neq j} ||Y_i - Y_j||.
\end{equation}
In fact, there are very fast approximations that speed up
the calculations; see \cite{huang2017efficient}.

To the best of our knowledge,
there is no fast way to compute barycenters with this metric.
But we can get around this by using $k$-mediods
in place of $k$-means.
This, if
$P_1,\ldots, P_r$ is a set of distributions
in a cluster,
we define the centroid to be $P_t$ where
$P_t = \argmin_{1\leq j \leq r} \sum_s {\cal E}(P_j,P_s)$.
That is, we restrict the search for a centroid to be over the
observed distributions.

Hence, if we use $k$-mediods,
distribution clustering with the energy distance is feasible.
However, the energy distance is not
shape preserving.
To see this, consider the following example.
$P_1=\delta_{-a}$ and
$P_2=\delta_{a}$ 
where $a>0$ and $\delta$ denotes a point mass.
The barycenter $P$
minimizes
${\cal E}(P,P_1) + {\cal E}(P,P_2)$.
Recall the the Wasserstein barycenter is $\delta_0$.
But this is not true
for the energy distance.
To see this, consider distributions of the form
$P = (1/2)\delta_{-b} + m(1/2)\delta_{b}$.
It is easy to show that,
over such distributions,
${\cal E}(P,P_1) + {\cal E}(P,P_2)$
is minimized by choosing
the mixture
$P = (1/2)\delta_{-a} + m(1/2)\delta_{a}$.
This implies that the barycenter cannot be $\delta_0$ as desired.

In summary, we see that when distributions are well separated,
the energy barycenter is over-dispersed compared to the Wasserstein barycenter.
This leads to situations where the centroid of a cluster
might look quite different 
than the distributions in the cluster.
It is tempting to look for a simple modification
of the energy distance that fixes this problem.
We have tried several approaches without success.
Thus, the advantage of energy distance
is that it can be computed quickly but
the disadvantage is that it does not preserve the shape
of the distributions when computing barycenters.

\subsection{Optimal Preconditioning}
The linear transformation that we used, namely,
$\tilde{X}=\Sigma_X^{-1/2}(X-\mu_X)$
matches the first two moments of
$\tilde{X}$ with $U$
where $U\sim R$ is the reference distribution.
This transformation is chosen for convenience.
An alternative approach is to choose an optimal linear
transformation.
That is, we could choose
$a$ and $A$ to minimize
$W^2(a+AX,U)$.
Unfortunately, computing the optimal $a$ and $A$ is non-trivial.
If we permit linear transformations
of both $X_j$ and $U$ then
there is a closed form expression for the optimal
transformation;
see \cite{kuang2017preconditioning}.
This requires a different transformation of $U$ for each $X_j$.
Hence, we lose the idea 
of  a single, fixed, reference distribution.
It would be possible to
construct a composition of maps
$X_j \to \tilde{X}_j \to \tilde{Z}_j \to U$
where
$X_j \to \tilde{X}_j$ is the optimal linear map on $X_j$ and
$\tilde{Z}_j \to U$ is the optimal linear map for $U$.
\cite{kuang2017preconditioning}
show that
the optimal map $T(X)$
is the same as the composition
$M^{-1}\circ T_\dagger \circ L$
where $T_\dagger$
is the optimal transport map from
$\tilde X$ to $\tilde Y$;
see Figures \ref{fig::optimal}
and \ref{fig::pairwise}.

\begin{figure}
\adjustbox{scale=2,center}{%
\begin{tikzcd}
  X \arrow[r, "T"] \arrow[d, "L"]  & Y \arrow[d, "M" ] \\
  \tilde{X} \arrow[r,  "T_\dagger" ]     & \tilde{Y} 
\end{tikzcd}
}
\caption{\em The optimal linear map, preserves the Wasserstein distnce and makes
$\tilde X$ and $\tilde Y$ as close as possible.}
\label{fig::optimal}
\end{figure}

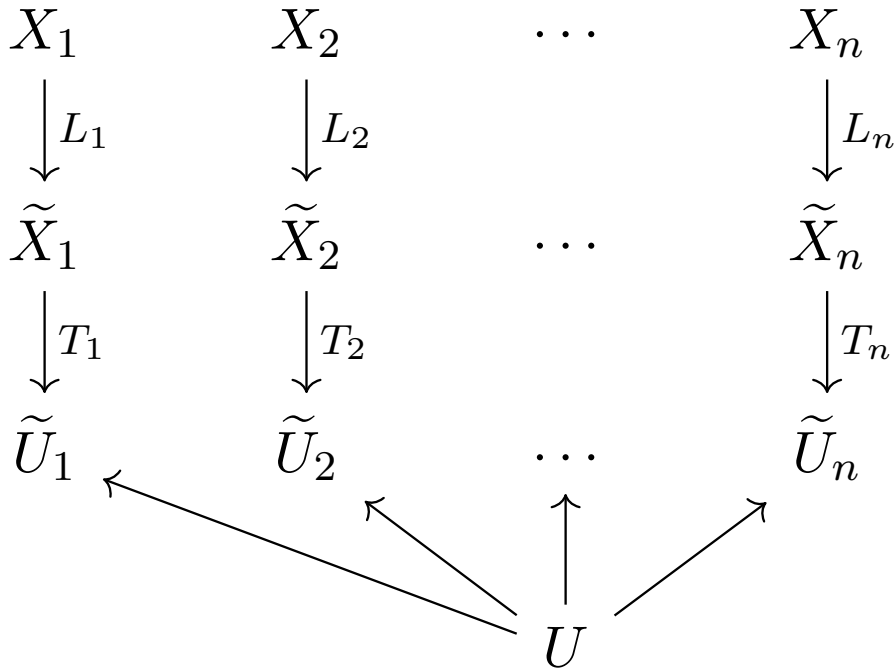
\begin{figure}
\adjustbox{scale=2,center}{
\begin{tikzcd}
X_1 \arrow[d, "L_1"] & X_2 \arrow[d, "L_2"] & \cdots & X_n \arrow[d, "L_n"]\\
\tilde{X}_1 \arrow[d, "T_1"]& \tilde{X}_2 \arrow[d, "T_2"] & \cdots & \tilde{X}_n \arrow[d, "T_n"]\\
\tilde{U}_1 & \tilde{U}_2 & \cdots & \tilde{U}_n\\
            &             & U \arrow[ull] \arrow[ul] \arrow[u]\arrow[ur] & \\
\end{tikzcd}
}
\vspace{-.5in}
\caption{\em The optimal linear preconditioning is 
applied to $X_j$ and $U$ in a pairwise fashion.}
\label{fig::pairwise}
\end{figure}

This might lead to more accurate approximations at
the expense of much more computation.
In the interest of keeping the method as simple as possible,
we have not use this more involved approach.

\subsection{The Choice of Reference Measure and Multiple Tangents}

When we used the
tangent approximation,
we approximated the distance by projecting
one a single tangent space.
An alternative, when there are a large number of diverse distributions,
is to use several tangent approximations.
The datasets could first clustered using some fast approximation, such as the
marginal hybrid distance.
A separate tangent approximation is computed for each perliminary cluster.
Then each distribution can be represented by its local tangent approximation
using the same reference measure.
Another issue is the choice of reference measures.
We do not know if there is an optimal reference measure.
It is not even clear how to define such a notion of optimality.
We conjecture that, at least for distribution clustering,
the choice of reference measure is not critical.

\subsection{Other Applications of the Hybrid Approach}
\label{section::otherapp}

The hybrid distance can be used for other tasks besides
clustering.
In this section we outline how the hybrid distance
can be used for these tasks.

{\bf Nonparametric Shrinkage.}
Suppose that
$P_1,\ldots, P_N$ are random distributions
drawn from a distribution $\Pi$.
Let ${\cal D}_j$ be $n_j$ samples drawn from $P_j$.
Suppose that we want to borrow strength from all the data
to estimate each $P_j$.
Let $R$ be a reference measure
and let
$(\mu_j,\Sigma_j,\psi_j)=\phi(P_j)$
denote the hybrid transform
as in (\ref{eq::transform}).
Let
$(\hat \mu_j,\hat \Sigma_j,\hat \psi_j)$
denote the $j^{\rm th}$ estimate.
We can apply shrinkage estimation separately to
the $\hat \mu_j$'s,
the $\hat \Sigma_j$'s
and the $\hat\psi_j$'s.
Denote these estimates by
$\overline{\mu}_j$,
$\overline{\Sigma}_j$,
$\overline{\psi}_j$.
Inverting $\phi$ 
gives the shrinkage estimator $\overline{P}_j$.
We obtain 
$\overline{\mu}_j$
using standard James-Stein shrinkage,
$\overline{\Sigma}_j$ by covariance shrinkage
as in \cite{nguyen2018distributionally} and
$\overline{\psi}_j$ is obtained by
functional shrinkage as in
\cite{guo2002functional}.

{\bf Multi-Sample Testing.}
Suppose that we want to test the null hypothesis
$H_0: P_1 = \cdots =P_N$.
The hybrid distance allows us to
split this null into three different null hypotheses, namely,
$H_0: \mu_1 = \cdots =\mu_N$,
$H_0: \Sigma_1 = \cdots = \Sigma_N$ and
$H_0: \psi_1 = \cdots = \psi_N$.
This allows us to separate the deviations from the null
in terms of location, scale and shape.

{\bf Multi-Level Clustering.}
After clustering the distributions
$P_1,\ldots, P_N$
into clusters ${\cal C}_1,\ldots, {\cal C}_k$,
we may want to further cluster the datasets
${\cal D}_1,\ldots, {\cal D}_N$.
We could apply any standard clustering algorithm to each dataset.
But we may want the distributions in a cluster
${\cal C}_j$ to have similar clusterings.
For example, we could use Normal mixture clustering on each dataset
then apply the shrinkage ideas mentioned above to
make the clusterings similar.
This is an alternative to simultaneous Wasserstein clustering as in
\cite{ho2017multilevel} which is quite expensive.

\section*{Appendix}
\label{section::theory}

{\bf Consistency.}
In this section
we show that the 1-nearest neighbor regression method
for estimating the transport map
is consistent.
Let
$X_1,\ldots, X_m \sim P$ and
$Y_1,\ldots, Y_m \sim Q$.
Assume that $P$ has density $p$ and
$Q$ has density $q$
and that both are supported on a compact subset
of $\mathbb{R}^d$.
Let $T$ be the optimal transport map from
$P$ to $Q$.
Let $P_n$ be the empirical measure based on
$X_1,\ldots, X_m$ and
let $Q_n$ be the empirical measure based on
$Y_1,\ldots, Y_m$.

Let
$V_1,\ldots, V_m$ denote the Voronoi tesselation defined by
$X_1,\ldots, X_m$.
Let $\pi$ be the permutation that minimizes
$$
\sum_i ||X_i - Y_{\pi(i)}||^2.
$$
Our estimator of $T$ is
\begin{equation}\label{eq::TT}
\hat T(x) = \sum_{s=1}^j I(x\in V_s)Y_{\pi(j(x))}
\end{equation}
where
$j(x)$ is the integer $j$ such that
$||x-X_j|| \leq ||x-X_t||$ for all $t\neq j$.

\begin{theorem}
Let $P$ and $Q$ be supported on a compact set ${\cal X}\subset \mathbb{R}^d$.
Suppose that $P$ and $Q$ have continuous densities $p$ and $q$
and that
$\inf_{x\in {\cal X}}p(x) >0$ and
$\inf_{y\in {\cal Y}}q(x) >0$.
Let
$X_1,\ldots, X_n \sim P$ and
$Y_1,\ldots, Y_n \sim Q$.
Let $T$ be the optimal transport map and let
$\hat T$ be the estimated map in (\ref{eq::TT}).
Then
$$
\int ||x-\hat T(x)||^2dP(x)= \int ||x- T(x)||^2dP(x) + O_P(\log n\, n^{-1/d})
$$
for some positive constant $a>0$.
\end{theorem}

\begin{proof}
Let $V_1,\ldots, V_n$ denote the Voronoi diagram
defined by $X_1,\ldots, X_n$.
From Lemma 
\ref{lemma::Voronoi} below,
$\max_j |P(V_j) - 1/n| + O_P(1/n^2)$
and
$\max_j {\rm diam}(V_j) =O_P(\log n \, n^{-1/d})$.

Let $P_n$ and $Q_n$ be the empirical
distributions of
$X_1,\ldots, X_n$ and
$Y_1,\ldots, Y_n$.
From \cite{weed2017sharp}, we have that $W^2(P_n,Q_n) = O_P(n^{-2/d})$.
Now
\begin{align*}
\int ||x-\hat T(x)||^2 dP(x) &=
\sum_j \int_{V_j} ||x-\hat T(x)||^2 dP(x) =
\sum_j \int_{V_j} ||x-\hat T(X_j)||^2 dP(x)\\
&=A + B + C
\end{align*}
where
$$
A = \sum_j \int_{V_j}||x-X_j||^2 dP(x),\ 
B = \sum_j \int_{V_j}||X_j - \hat T(X_j)||^2,\ 
C = \sum_j \int_{V_j} \langle x-X_j, X_j - \hat T(X_j)\rangle dP(x).
$$
Now
$A \leq \sum_j \max_j {\rm diam}(V_j)P(V_j) = \max_j {\rm diam}(V_j)=O_P(\log n \,n^{-1/d})$
\begin{align*}
B &= \sum_j ||X_j - \hat T(X_j)||^2 P(V_j) = 
\frac{1}{n}\sum_j ||X_j - \hat T(X_j)||^2  +
\sum_j ||X_j - \hat T(X_j)||^2  (P(V_j)-1/n)\\
&=
W^2(P_n,Q_n) + \sum_j ||X_j - \hat T(X_j)||^2  (P(V_j)-1/n)\\
&=
W^2(P_n,Q_n) + O_P(1/n^2) =
W^2(P,Q) + O_P(1/n^2) + O_P(n^{-2/d}) \\
&=\int ||x-T(x)||^2 dP(x) + O_P(n^{-2/d})
\end{align*}
and
\begin{align*}
C &= \sum_j \int_{V_j} \langle x-X_j, X_j - \hat T(X_j)\rangle dP(x) \leq
\sum_j \sqrt{\int_{V_j} || x-X_j||^2} dP(x)\sum_j \sqrt{||X_j - \hat T(X_j)||^2 dP(x)}\\
&= O(\max_j {\rm diam}(V_j)) = O_P(\log n \, n^{-1/d}).
\end{align*}
\end{proof}

The following are standard facts about Voronoi cells
for strictly positive densities.

\begin{lemma}
\label{lemma::Voronoi}
Suppose that $P$ 
is supported on a compact set and has continous density $p$ that
is strictly bounded away from 0.
Let $X_1,\ldots, X_n \sim P$
and let
$V_1,\ldots, V_n$
be the Voronoi cells defined by
$X_1,\ldots, X_n$.
Then
(i) $\mathbb{E}[P(V_i)]=1/n$,
(ii) ${\rm Var}[P(V_i)] \leq 2/n^2$ and
(iii) $\max_j {\rm diam}(V_j) = O_P(\log n\, n^{-1/d})$,
(iv) $\sum_i |P(V_i)-1/n|=o_P(1)$.
\end{lemma}

\vspace{1cm}

{\bf Wasserstein $k$-means in one dimension.}
Now we give the steps
for exact Wasserstein $k$-means clustering in one dimension.
Given datasets
${\cal D}_1,\ldots, {\cal D}_N$
let $F_1,\ldots, F_N$
be the empirical cdf's.
We use the trimmed Wasserstein distance
$$
W^2(F_j,F_k) = \frac{1}{1-2\delta}\int_\delta^{1-\delta} (F_j^{-1}(s) -F_k^{-1}(s))^2 ds
$$
where $\delta>0$ is some specified positive trimming constant.

First, we use $k-{\rm means}^{++}$ seeding to get starting centroids:

\rule{2in}{1mm}\\
{\sf Wasserstein $k$-means seeding}\\
\rule{2in}{1mm}\\
\begin{enum}
\item Input: integer $k$ and data sets
${\cal D}_1,\ldots, {\cal D}_N$.
\item Compute all parwise Wasserstein distances
$$
D_{jk}^2 = W^2(F_j,F_k).
$$
\item Let $F_1,\ldots, F_N$ be the empirical cdf's of the data sets.
\item Find the starting values:
\begin{enum}
\item Let $c_1= F_j$ where $j$ is chosen randomly from $1,\ldots, N$.
Set ${\cal C} = \{j\}$.
\item Choose $c_2$ randomly from
$\{F_1,\ldots, F_N\}$ where $F_s$ is chosen with probability
proportional to $\min_{j\in {\cal C}}D_{sj}^2$.
\item Set ${\cal C} \leftarrow {\cal C} \bigcup\{j\}$.
\item Repeat the last two steps until ${\cal C}=\{c_1,\ldots, c_k\}$ has $k$ elements.
\end{enum}
\end{enum}

The main algorithm is as follows.

\rule{2in}{1mm}\\
{\sf Wasserstein $k$-means}\\
\rule{2in}{1mm}\\
\begin{enum}
\item Run {\sf Wasserstein $k$-means seeding} to get
centroids ${\cal C}$.
\item Compute
$D_{js} = W(F_j,c_s)$ for $j=1,\ldots, N$ and
$s=1,\ldots, k$.
\item Assign each $F_j$ to its nearest centroid:
let ${\cal C}_j = \{F_r:\ W(F_r,F_j) < W(F_r,F_t)\ t\neq j\}$.
\item Let $c_j$ be the centroid of ${\cal C}_j$ using {\sf Wasserstein centroid}.
\item Repeat last two steps until convergence.
\end{enum}

\rule{2in}{1mm}\\
{\sf Wasserstein centroid}\\
\rule{2in}{1mm}\\
\begin{enum}
\item Given one-dimensional cdf's
$F_1,\ldots, F_N$.
\item Let
$c(s) = \frac{1}{N}\sum_j F_j^{-1}(s).$
\item Return $F(x) = c^{-1}(x)$.
\end{enum}

\bibliographystyle{plainnat}
\bibliography{paper}

\begin{thebibliography}{28}
\providecommand{\natexlab}[1]{#1}
\providecommand{\url}[1]{\texttt{#1}}
\expandafter\ifx\csname urlstyle\endcsname\relax
  \providecommand{\doi}[1]{doi: #1}\else
  \providecommand{\doi}{doi: \begingroup \urlstyle{rm}\Url}\fi

\bibitem[{\'A}lvarez-Esteban et~al.(2016){\'A}lvarez-Esteban, del Barrio,
  Cuesta-Albertos, and Matr{\'a}n]{alvarez2016fixed}
Pedro~C {\'A}lvarez-Esteban, E~del Barrio, JA~Cuesta-Albertos, and
  C~Matr{\'a}n.
\newblock A fixed-point approach to barycenters in {W}asserstein space.
\newblock \emph{Journal of Mathematical Analysis and Applications},
  441\penalty0 (2):\penalty0 744--762, 2016.

\bibitem[Arthur and Vassilvitskii(2007)]{arthur2007k}
David Arthur and Sergei Vassilvitskii.
\newblock k-means++: The advantages of careful seeding.
\newblock In \emph{Proceedings of the eighteenth annual ACM-SIAM symposium on
  Discrete algorithms}, pages 1027--1035. Society for Industrial and Applied
  Mathematics, 2007.

\bibitem[Bhatia et~al.(2018)Bhatia, Jain, and Lim]{bhatia2018bures}
Rajendra Bhatia, Tanvi Jain, and Yongdo Lim.
\newblock On the bures--{W}asserstein distance between positive definite
  matrices.
\newblock \emph{Expositiones Mathematicae}, 2018.

\bibitem[Chac{\'o}n and Duong(2018)]{chacon2018multivariate}
Jos{\'e}~E Chac{\'o}n and Tarn Duong.
\newblock \emph{Multivariate Kernel Smoothing and Its Applications}.
\newblock Chapman and Hall/CRC, 2018.

\bibitem[Chac{\'o}n et~al.(2015)]{chacon2015population}
Jos{\'e}~E Chac{\'o}n et~al.
\newblock A population background for nonparametric density-based clustering.
\newblock \emph{Statistical Science}, 30\penalty0 (4):\penalty0 518--532, 2015.

\bibitem[Cheng(1995)]{cheng1995mean}
Yizong Cheng.
\newblock Mean shift, mode seeking, and clustering.
\newblock \emph{IEEE transactions on pattern analysis and machine
  intelligence}, 17\penalty0 (8):\penalty0 790--799, 1995.

\bibitem[del Barrio et~al.(2017)del Barrio, Cuesta-Albertos, Matr{\'a}n, and
  Mayo-{\'I}scar]{del2017robust}
E~del Barrio, JA~Cuesta-Albertos, C~Matr{\'a}n, and A~Mayo-{\'I}scar.
\newblock Robust clustering tools based on optimal transportation.
\newblock \emph{Statistics and Computing}, pages 1--22, 2017.

\bibitem[Duong et~al.(2016)Duong, Beck, Azzag, and Lebbah]{duong2016nearest}
Tarn Duong, Ga{\"e}l Beck, Hanene Azzag, and Mustapha Lebbah.
\newblock Nearest neighbour estimators of density derivatives, with application
  to mean shift clustering.
\newblock \emph{Pattern Recognition Letters}, 80:\penalty0 224--230, 2016.

\bibitem[Ferraty and Vieu(2006)]{ferraty2006nonparametric}
Fr{\'e}d{\'e}ric Ferraty and Philippe Vieu.
\newblock \emph{Nonparametric functional data analysis: theory and practice}.
\newblock Springer Science \& Business Media, 2006.

\bibitem[Givens et~al.(1984)Givens, Shortt, et~al.]{givens1984class}
Clark~R Givens, Rae~Michael Shortt, et~al.
\newblock A class of {W}asserstein metrics for probability distributions.
\newblock \emph{The Michigan Mathematical Journal}, 31\penalty0 (2):\penalty0
  231--240, 1984.

\bibitem[Guo(2002)]{guo2002functional}
Wensheng Guo.
\newblock Functional mixed effects models.
\newblock \emph{Biometrics}, 58\penalty0 (1):\penalty0 121--128, 2002.

\bibitem[Ho et~al.(2017)Ho, Nguyen, Yurochkin, Bui, Huynh, and
  Phung]{ho2017multilevel}
Nhat Ho, XuanLong Nguyen, Mikhail Yurochkin, Hung~Hai Bui, Viet Huynh, and Dinh
  Phung.
\newblock Multilevel clustering via {W}asserstein means.
\newblock \emph{arXiv preprint arXiv:1706.03883}, 2017.

\bibitem[Huang and Huo(2017)]{huang2017efficient}
Cheng Huang and Xiaoming Huo.
\newblock An efficient and distribution-free two-sample test based on energy
  statistics and random projections.
\newblock \emph{arXiv preprint arXiv:1707.04602}, 2017.

\bibitem[Jiang et~al.(2018)Jiang, Jang, and Kpotufe]{jiang2018quickshift++}
Heinrich Jiang, Jennifer Jang, and Samory Kpotufe.
\newblock Quickshift++: Provably good initializations for sample-based mean
  shift.
\newblock \emph{arXiv preprint arXiv:1805.07909}, 2018.

\bibitem[Kuang and Tabak(2017)]{kuang2017preconditioning}
Max Kuang and Esteban~G Tabak.
\newblock Preconditioning of optimal transport.
\newblock \emph{SIAM Journal on Scientific Computing}, 39\penalty0
  (4):\penalty0 A1793--A1810, 2017.

\bibitem[Kuhn(1955)]{kuhn1955hungarian}
Harold~W Kuhn.
\newblock The hungarian method for the assignment problem.
\newblock \emph{Naval research logistics quarterly}, 2\penalty0 (1-2):\penalty0
  83--97, 1955.

\bibitem[Maier et~al.(2007)Maier, Anderson, De~Jager, Wicker, and
  Hafler]{maier2007allelic}
Lisa~M Maier, David~E Anderson, Philip~L De~Jager, Linda~S Wicker, and David~A
  Hafler.
\newblock Allelic variant in ctla4 alters t cell phosphorylation patterns.
\newblock \emph{Proceedings of the National Academy of Sciences}, 104\penalty0
  (47):\penalty0 18607--18612, 2007.

\bibitem[Nguyen et~al.(2018)Nguyen, Kuhn, and
  Esfahani]{nguyen2018distributionally}
Viet~Anh Nguyen, Daniel Kuhn, and Peyman~Mohajerin Esfahani.
\newblock Distributionally robust inverse covariance estimation: The
  {W}asserstein shrinkage estimator.
\newblock \emph{arXiv preprint arXiv:1805.07194}, 2018.

\bibitem[Panaretos and Zemel(2018)]{panaretos2018statistical}
Victor~M Panaretos and Yoav Zemel.
\newblock Statistical aspects of wasserstein distances.
\newblock \emph{Annual Review of Statistics and Its Application}, 2018.

\bibitem[Rippl et~al.(2016)Rippl, Munk, and Sturm]{rippl2016limit}
Thomas Rippl, Axel Munk, and Anja Sturm.
\newblock Limit laws of the empirical {W}asserstein distance: Gaussian
  distributions.
\newblock \emph{Journal of Multivariate Analysis}, 151:\penalty0 90--109, 2016.

\bibitem[Rosenbaum(2005)]{rosenbaum2005exact}
Paul~R Rosenbaum.
\newblock An exact distribution-free test comparing two multivariate
  distributions based on adjacency.
\newblock \emph{Journal of the Royal Statistical Society: Series B (Statistical
  Methodology)}, 67\penalty0 (4):\penalty0 515--530, 2005.

\bibitem[Silverman(2018)]{silverman2018density}
Bernard~W Silverman.
\newblock \emph{Density estimation for statistics and data analysis}.
\newblock Routledge, 2018.

\bibitem[Sommerfeld and Munk(2018)]{sommerfeld2018inference}
Max Sommerfeld and Axel Munk.
\newblock Inference for empirical {W}asserstein distances on finite spaces.
\newblock \emph{Journal of the Royal Statistical Society: Series B (Statistical
  Methodology)}, 80\penalty0 (1):\penalty0 219--238, 2018.

\bibitem[Sommerfeld et~al.(2018)Sommerfeld, Schrieber, and
  Munk]{sommerfeld2018optimal}
Max Sommerfeld, J{\"o}rn Schrieber, and Axel Munk.
\newblock Optimal transport: Fast probabilistic approximation with exact
  solvers.
\newblock \emph{arXiv preprint arXiv:1802.05570}, 2018.

\bibitem[Sz{\'e}kely and Rizzo(2013)]{szekely2013energy}
G{\'a}bor~J Sz{\'e}kely and Maria~L Rizzo.
\newblock Energy statistics: A class of statistics based on distances.
\newblock \emph{Journal of statistical planning and inference}, 143\penalty0
  (8):\penalty0 1249--1272, 2013.

\bibitem[Villani(2003)]{villani2003topics}
C{\'e}dric Villani.
\newblock \emph{Topics in optimal transportation}.
\newblock American Mathematical Soc., 2003.

\bibitem[Wang et~al.(2013)Wang, Slep{\v{c}}ev, Basu, Ozolek, and
  Rohde]{wang2013linear}
Wei Wang, Dejan Slep{\v{c}}ev, Saurav Basu, John~A Ozolek, and Gustavo~K Rohde.
\newblock A linear optimal transportation framework for quantifying and
  visualizing variations in sets of images.
\newblock \emph{International journal of computer vision}, 101\penalty0
  (2):\penalty0 254--269, 2013.

\bibitem[Weed and Bach(2017)]{weed2017sharp}
Jonathan Weed and Francis Bach.
\newblock Sharp asymptotic and finite-sample rates of convergence of empirical
  measures in {W}asserstein distance.
\newblock \emph{arXiv preprint arXiv:1707.00087}, 2017.

\end{thebibliography}

\end{document}